\documentclass[twocolumn,english]{IEEEtran}
\usepackage[T1]{fontenc}
\usepackage[latin9]{inputenc}
\usepackage{babel}
\usepackage{array}
\usepackage{varioref}
\usepackage{refstyle}
\usepackage{float}
\usepackage{calc}
\usepackage{multirow}
\usepackage{amsmath}
\usepackage{amsthm}
\usepackage{amssymb}
\usepackage{graphicx}
\PassOptionsToPackage{normalem}{ulem}
\usepackage{ulem}

\makeatletter

\AtBeginDocument{\providecommand\figref[1]{\ref{fig:#1}}}
\AtBeginDocument{\providecommand\secref[1]{\ref{sec:#1}}}
\AtBeginDocument{\providecommand\eqref[1]{\ref{eq:#1}}}
\AtBeginDocument{\providecommand\probref[1]{\ref{prob:#1}}}
\AtBeginDocument{\providecommand\defref[1]{\ref{def:#1}}}
\AtBeginDocument{\providecommand\lemref[1]{\ref{lem:#1}}}
\AtBeginDocument{\providecommand\exaref[1]{\ref{exa:#1}}}
\AtBeginDocument{\providecommand\propref[1]{\ref{prop:#1}}}
\AtBeginDocument{\providecommand\proref[1]{\ref{pro:#1}}}
\AtBeginDocument{\providecommand\tabref[1]{\ref{tab:#1}}}

\providecommand{\tabularnewline}{\\}
\RS@ifundefined{subsecref}
  {\newref{subsec}{name = \RSsectxt}}
  {}
\RS@ifundefined{thmref}
  {\def\RSthmtxt{theorem~}\newref{thm}{name = \RSthmtxt}}
  {}
\RS@ifundefined{lemref}
  {\def\RSlemtxt{lemma~}\newref{lem}{name = \RSlemtxt}}
  {}

  \theoremstyle{definition}
  \newtheorem{defn}{\protect\definitionname}
  \theoremstyle{plain}
  \newtheorem{lem}{\protect\lemmaname}
 \theoremstyle{definition}
  \newtheorem{example}{\protect\examplename}
  \theoremstyle{remark}
  \newtheorem{rem}{\protect\remarkname}
  \theoremstyle{definition}
  \newtheorem{problem}{\protect\problemname}
  \theoremstyle{plain}
  \newtheorem{prop}{\protect\propositionname}
\theoremstyle{plain}
\newtheorem{thm}{\protect\theoremname}

\@ifundefined{date}{}{\date{}}

\usepackage{ifxetex}
\ifxetex
\usepackage{fontspec}
\fi

\newcommand{\person}[2]{#1\,{\footnotesize<#2>\rm}}

  \newcommand{\mythanks}{
  \protect\thanks{\person{Andrea\,Censi}{censi@mit.edu} is with the
  Laboratory for Information and Decision Systems (LIDS)
  at the Massachusetts Institute of Technology. 
}}

\usepackage{substr}

\usepackage[table]{xcolor}

\newcommand{\tableColors}{\rowcolors{2}{green!4}{blue!4}}

\usepackage{xspace}
\usepackage{amsmath}
\usepackage{amssymb}
\usepackage[mathscr]{eucal}

\newtheoremstyle{plain}
  {}
  {}
  {}
  {}
  {\bfseries}
  {.}
  { }
  {}

\newtheoremstyle{remark}
  {}
  {}
  {}
  {}
  {\bfseries}
  {.}
  { }
  {}

\usepackage[tracking=true,kerning=true,spacing=true,factor=1000,stretch=15,shrink=20]{microtype}
\SetTracking{encoding={*}, shape=sc}{0}

\usepackage{ifxetex}
\ifxetex
\usepackage{etoolbox}
\usepackage{fontspec}
\newfontfeature{Microtype}{protrusion=default;expansion=default;}

\setmathrm[Microtype,Scale=0.9]{Myriad Pro}
\setmathsf[Microtype,Scale=0.9]{Myriad Pro}
\fi

\@ifpackageloaded{prettyref}{
	\errmessage{
	refstyle_ieee_abbreviated.tex is not compatible with prettyref. 
	(LyX: Maybe disable use_ref_style = 0 in your Lyx file.)}
}{}

\usepackage{refstyle}
\newref{prob}{refcmd={Problem\,\ref{#1}}}

\newref{sub}{refcmd={Sec.\,\ref{#1}}}
\newref{sec}{refcmd={Sec.\,\ref{#1}}}

\newref{prop}{refcmd={Prop.\,\ref{#1}}}
\newref{pro}{refcmd={Prop.\,\ref{#1}}}

\newref{def}{refcmd={Def.\,\ref{#1}}}

\newref{app}{refcmd={Appendix\,\ref{#1}}}
\newref{alg}{refcmd={Algorithm\,\ref{#1}}}
\newref{cor}{refcmd={Corollary\,\ref{#1}}}
\newref{thm}{refcmd={Theorem\,\ref{#1}}}
\newref{tab}{refcmd={Table\,\ref{#1}}}
\newref{fig}{refcmd={Fig.\,\ref{#1}}}
\newref{line}{refcmd={line~\ref{#1}}}
\newref{exa}{refcmd={Example\,\ref{#1}}}
\newref{lem}{refcmd={Lemma\,\ref{#1}}}

\newref{eq}{refcmd={(\ref{#1})}}

\definecolor{symbol-highlight}{rgb}{0,0.6,0}

\newcommand{\idFunc}{\aword{Id}}

\newcommand{\reals}{\mathbb{R}}

\newcommand{\nonNegReals}{\reals^{+}_{\bullet}}

\newcommand{\etal}{{et\,al.}\xspace}
\newcommand{\eg}{{e.g.},\xspace}

\newcommand{\aword}[1]{\mathsf{#1}}

\newcommand{\vmath}[1]{\aword{#1}}

\newcommand{\subto}{\text{s.t.}} 
\newcommand{\with}{\text{using}}

\newcommand{\pset}{\mathscr{P}} 
\DeclareMathOperator*{\Min}{Min}

\DeclareMathOperator*{\Max}{Max}

\newcommand{\posleq}{\preceq}

\newcommand{\posgeq}{\succeq}

\newcommand{\posA}{\mathcal{P}} 

\newcommand{\posAleq}{\mathrel{{\posleq_\posA}}}

\newcommand{\posB}{\mathcal{Q}} 
\newcommand{\posBleq}{\mathrel{{\posleq_\posB}}}

\newcommand{\lfp}{\vmath{lfp}}

\newcommand{\CPO}{\textsc{CPO}\xspace}

\newcommand{\DCPO}{\textsc{DCPO}\xspace}

\newcommand{\antichains}{\vmath{A}}

\newcommand{\upsets}{\vmath{U}}

\newcommand{\upit}{{\uparrow\,}}

\newcommand{\posetwidth}{\vmath{width}}
\newcommand{\posetheight}{\vmath{height}}

\newcommand{\ftor}{{h}}

\newcommand{\Rcomp}{\overline{\mathbb{R}}_+}

\usepackage[normalem]{ulem}

\newcommand{\funsp}{\mathscr{F}} 
\newcommand{\funleq}{\posleq_{\funsp}} 
\newcommand{\fun}{\vmath{f}}

\newcommand{\imp}{\vmath{i}} 
\newcommand{\impsp}{\mathscr{I}} 

\newcommand{\exc}{\vmath{exec}} 
\newcommand{\eval}{\vmath{eval}}

\newcommand{\res}{\vmath{r}} 
\newcommand{\resleq}{\posleq_{\ressp}}

\newcommand{\ressp}{\mathscr{R}}

\newcommand{\dprob}{\vmath{dp}}

\newcommand{\dpseries}{\vmath{series}}
\newcommand{\dppar}{\vmath{par}}
\newcommand{\dploop}{\vmath{loop}}

\newcommand{\dploopb}{\vmath{loopb}}

\newcommand{\resMin}{{\Min_{\resleq}}}

\newcommand{\unconnectedfun}{\mathsf{UF}}
\newcommand{\unconnectedres}{\mathsf{UR}}

\newcommand{\Aressp}{{\antichains\ressp}}

\newcommand{\acprod}{\mathbin{\boldsymbol{\times}}} 

\newcommand{\oploop}{\dagger}
\newcommand{\opseries}{\mathbin{\varocircle}}
\newcommand{\oppar}{\mathbin{\varotimes}}
\newcommand{\opcoprod}{\mathbin{\varovee}}

\newcommand{\colR}{\color[rgb]{0.555789,0.000000,0.000000}}
\newcommand{\colF}{\color[rgb]{0.094869,0.500000,0.000000}}
\newcommand{\colH}{\color[rgb]{0.000000,0.400000,1.000000}}
\newcommand{\colI}{\color[RGB]{214,120,5}}

\newcommand{\R}[1]{{\colR #1}}
\newcommand{\F}[1]{{\colF #1}}
\newcommand{\I}[1]{{\colI #1}}

\newcommand{\One}{\mathbbm{1}} 
\newcommand{\cdpiN}{\mathcal{V}} 
\newcommand{\cdpin}{v} 
\newcommand{\cdpinA}{v_1}
\newcommand{\cdpinB}{v_2}
\newcommand{\cdpiresind}{i}
\newcommand{\cdpifunind}{j}
\newcommand{\cdpiresindA}{i_1}
\newcommand{\cdpifunindB}{j_2}
\newcommand{\dpinumf}{\vmath{nf}}
\newcommand{\dpinumr}{\vmath{nr}}

\usepackage{mathtools, cuted}

\usepackage[caption=false,font=footnotesize]{subfig}
\let\l@ENGLISH\l@english

\usepackage{hyperref}

\usepackage[style=numeric-comp,sorting=none,firstinits=true, maxnames=10, bibstyle=numeric,abbreviate=true,defernums=true,eprint=false,backend=bibtex]{biblatex}

\newcommand{\notpresent}[1]{{\color{red}No #1} }
\renewcommand{\notpresent}[1]{}

\bibliography{commons/bib/strings_long}
\bibliography{commons/bib/all_only_one_w_pages}

\AtEveryBibitem{\clearlist{pages}} 
\DeclareFieldFormat{pages}{}

\AtEveryBibitem{\clearlist{address}} 
\AtEveryBibitem{\clearlist{location}}

\AtEveryBibitem{\clearfield{pages}}
\AtEveryBibitem{\clearfield{isbn}}
\AtEveryBibitem{\clearfield{issn}}

\usepackage{wrapfig}

\newref{ass}{refcmd={Assumption~\ref{#1}}}

\renewcommand{\resleq}{\posleq_{\ressp}}

\usepackage{capt-of}

\newcommand*{\vcenteredhbox}[1]{\begingroup
\setbox0=\hbox{#1}\parbox{\wd0}{\box0}\endgroup}

\newcommand{\captionsideleft}[2]{
    \medskip
    \begin{minipage}{1.8cm}{
        \hfill
        \protect\captionof{figure}{#1}}\end{minipage}
    \begin{minipage}{6.6cm}
    
    \vcenteredhbox{{#2}}
    \hfill
    \end{minipage}
    \medskip
}

\usepackage{listings}

\usepackage{fancyvrb}

\newcommand*{\minwidthbox}[2]{
  \makebox[{\ifdim#2<\width\width\else#2\fi}]{#1}
}

\usepackage{eso-pic}

\newcommand{\firstpage}{
\AddToShipoutPictureBG*{
  \AtPageUpperLeft{
    \setlength\unitlength{1in}
    \hspace*{\dimexpr0.5\paperwidth\relax}
    \makebox(0,-0.75)[r]{This version: Sep 2016 - Last version at: \url{http://tiny.cc/co-design}}
}}
}

\renewcommand{\firstpage}{}

\usepackage{placeins}

\let\oldfun\fun     \renewcommand{\fun}{{\colF\oldfun}}
\let\oldres\res     \renewcommand{\res}{{\colR\oldres}}
\let\oldimp\imp     \renewcommand{\imp}{{\colI\oldimp}}

\let\oldexc\exc     \renewcommand{\exc}{{\colF\oldexc}}
\let\oldeval\eval     \renewcommand{\eval}{{\colR\oldeval}}

\let\oldfunsp\funsp \renewcommand{\funsp}{{\colF\oldfunsp}}
\let\oldressp\ressp \renewcommand{\ressp}{{\colR\oldressp}}
\let\oldimpsp\impsp \renewcommand{\impsp}{{\colI\oldimpsp}}

\let\oldftor\ftor \renewcommand{\ftor}{{\colH\oldftor}}

\renewcommand{\Aressp}{{\colR\antichains\ressp}}

\usepackage{stmaryrd}

\usepackage{bbm} 
\usepackage{adjustbox}

\newcommand{\scottcontinuous}{Scott continuous\xspace}
\newcommand{\scottcontinuity}{Scott continuity\xspace}

\renewcommand{\nonNegReals}{\mathbb{R}_+}
\newcommand{\nonNegRealsComp}{\overline{\mathbb{R}}_+}

\newcommand{\triv}{\mathsf{Triv}}

\newtheoremstyle{plain}
  {}
  {}
  {}
  {}
  {\bfseries}
  {.}
  { }
  {}

\newtheoremstyle{remark}
  {}
  {}
  {}
  {}
  {\bfseries}
  {.}
  { }
  {}

\@ifundefined{showcaptionsetup}{}{
 \PassOptionsToPackage{caption=false}{subfig}}
\usepackage{subfig}
\makeatother

  \providecommand{\definitionname}{Definition}
  \providecommand{\examplename}{Example}
  \providecommand{\lemmaname}{Lemma}
  \providecommand{\problemname}{Problem}
  \providecommand{\propositionname}{Proposition}
  \providecommand{\remarkname}{Remark}
\providecommand{\theoremname}{Theorem}

\begin{document}
\firstpage

\title{\vspace{-5mm}\huge A Mathematical Theory of Co-Design}

\author{Andrea Censi\mythanks}
\maketitle
\begin{abstract}
One of the challenges of modern engineering, and robotics in particular,
is designing complex systems, composed of many subsystems, rigorously
and with optimality guarantees. This paper introduces a theory of
co-design that describes ``design problems'', defined as tuples
of ``functionality space'', ``implementation space'', and ``resources
space'', together with a feasibility relation that relates the three
spaces. Design problems can be interconnected together to create ``co-design
problems'', which describe possibly recursive co-design constraints
among subsystems. A co-design problem induces a family of optimization
problems of the type ``find the minimal resources needed to implement
a given functionality''; the solution is an antichain (Pareto front)
of resources. A special class of co-design problems are Monotone Co-Design
Problems (MCDPs), for which functionality and resources are complete
partial orders and the feasibility relation is monotone and Scott
continuous. The induced optimization problems are multi-objective,
nonconvex, nondifferentiable, noncontinuous, and not even defined
on continuous spaces; yet, there exists a complete solution. The antichain
of minimal resources can be characterized as a least fixed point,
and it can be computed using Kleene's algorithm. The computation needed
to solve a co-design problem can be bounded by a function of a graph
property that quantifies the interdependence of the subproblems. These
results make us much more optimistic about the problem of designing
complex systems in a rigorous way.  
\end{abstract}

\section{Introduction}

\IEEEPARstart{O}{ne} of the great engineering challenge of this
century is dealing with the design of ``complex'' systems. A complex
system is complex because its components cannot be decoupled; otherwise,
it would be just a (simple) product of simple systems. The \emph{design}
of a complex system is complicated because of the ``co-design constraints'',
which are the constraints that one subsystem induces on another. This
paper is an attempt towards formalizing and systematically solving
the problem of ``co-design'' of complex systems with recursive design
constraints.

\subsubsection*{Robotic systems as the prototype of complex systems}

Robotics is the prototypical example of a field that includes heterogeneous
multi-domain co-design constraints. The design of a robotic system
involves the choice of physical components, such as the actuators,
the sensors, the power supply, the computing units, the network links,
etc. Not less important is the choice of the software components,
including perception, planning, and control modules. All these components
induce co-design constraints on each other. Each physical component
has SWAP characteristics such as its shape (which must contained somewhere),
weight (which adds to the payload), power (which needs to be provided
by something else), excess heat (which must be dissipated somehow),
etc. Analogously, the software components have similar co-design constraints.
For example, a planner needs a state estimate. An estimator provides
a state estimate, and requires the data from a sensor, which requires
the presence of a sensor, which requires power. Everything costs money
to buy or develop or license. 

What makes system design problems non trivial is that the constraints
might be recursive. This is a form of \emph{feedback} in the problem
of design~(\figref{intro}). For example, a battery provides power,
which is used by actuators to carry the payload. A larger battery
provides more power, but it also increases the payload, so more power
is needed. Extremely interesting trade-offs arise when considering
constraints between the mechanical system and the embodied intelligence.
For control, typically a better state estimate saves energy in the
execution, but requires better sensors (which increase the cost and
the payload) or better computation (which increases the power consumption). 

\begin{figure}[t]
\begin{centering}
\includegraphics[width=8.6cm]{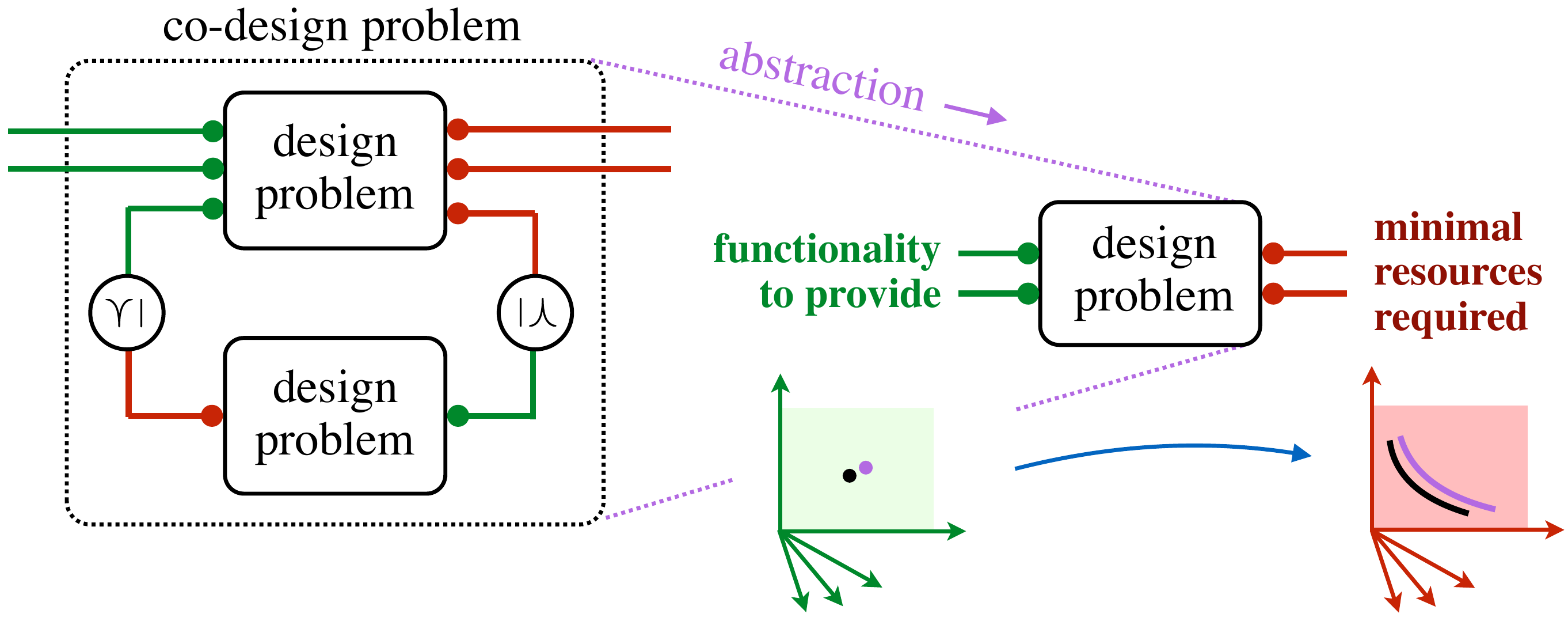}
\par\end{centering}
\caption{\label{fig:intro}A \emph{design problem }is a relation that relates
the implementations available to the \F{functionality provided}
 and the \R{resources required}, both represented as partially ordered
sets. A\emph{ co-design problem }is the interconnection of two or
more design problems. An edge in a co-design diagram like in the figure
represent a \emph{co-design constraint}: the resources required by
the first design problem are a lower bound for the functionality to
be provided by the second. The optimization problem to be solved is:
find the solutions that are minimal in resources usage, given a lower
bound on the functionality to be provided. }
\end{figure}

\subsubsection*{Contribution: A Principled Theory of Co-Design}

This paper describes a theory to deal with arbitrarily complex co-design
problems in a principled way. A~\emph{design problem} is defined
as a tuple of functionality space, implementation space, and a resources
space, plus the two maps that relate an implementation to functionality
provided and resources required. A design problem defines a family
of optimization problems of the type ``find the minimal resources
needed to implement a given functionality''. A \emph{co-design problem}
is an interconnection of design problems according to an arbitrary
graph structure, including feedback connections. Monotone Co-Design
Problems (MCDPs) are the composition of design problems for which
both functionality and resources are complete partial orders, and
the relation between functionality implemented and resources needed
is monotone (order-preserving) and \scottcontinuous. The first main
result in this paper (Theorem~\vref{thm:CDP-monotone}) is that the
class of MCDPs is closed with respect to interconnection. The second
main result (Theorem~\vref{thm:CDP-solvig}) is that there exists
a systematic procedure to solve an MCDP, assuming there is a procedure
to solve the primitive design problems. The solution of an MCDP\textemdash a
Pareto front, or ``antichain'' of minimal resources\textemdash can
be found by solving a least fixed point iteration in the space of
antichains. The complexity of this iteration depends on the structure
of the co-design diagram.

This paper is a generalization of previous work~\cite{censi15monotone},
where the interconnection was limited to one cycle. A conference
version of this work appeared as~\cite{censi15same}.

\subsubsection*{Outline}

\secref{Background} recalls necessary background about partial orders.
\secref{Design-Problems} defines co-design problems. \secref{Optimization}
contains a brief statement of results. \secref{threeoperators} describes
composition operators for design problems. \secref{Decomposition}
shows how any interconnection of design problems can be described
using three composition operators (series, parallel, feedback). \secref{Monotone-Co-Design-Problems}
describes the invariance of a monotonicity property that is preserved
by the composition operators. \secref{Solution-of-Monotone} describes
solution algorithms for MCDPs. \secref{Numerical-examples} shows
numerical examples. \secref{Discussion-of-related} discusses related
work.

\section{Background\label{sec:Background}}

We will use basic facts about order theory. Davey and Priestley~\cite{davey02}
and Roman~\cite{roman08} are possible reference texts.

Let $\left\langle \posA,\posAleq\right\rangle $ be a partially ordered
set (poset), which is a set~$\posA$ together with a partial order~$\posAleq$
(a reflexive, antisymmetric, and transitive relation). The partial
order~``$\posAleq$'' is written as~``$\posleq$'' if the context
is clear. If a poset has a least element, it is called ``bottom''
and it is denoted by~$\bot_{\posA}$. If the poset has a maximum
element, it is called ``top'' and denoted as~$\top_{\posA}$.

\subsubsection*{Chains and antichains}

A \emph{chain} $x\posleq y\posleq z\posleq\dots$ is a subset of a
poset in which all elements are comparable. An \emph{antichain} is
a subset of a poset in which \emph{no} elements are comparable. This
is the mathematical concept that formalizes the idea of ``Pareto
front''.

\begin{defn}[Antichains]
A subset $S\subseteq\posA$ is an antichain iff no elements are comparable:
for~$x,y\in S$, $x\posleq y$ implies~$x=y$. 
\end{defn}
Call~$\antichains\posA$ the set of all antichains in~$\posA$.
By this definition, the empty set is an antichain: $\emptyset\in\antichains\posA$.

\begin{defn}[Width and height of a poset]
\label{def:poset-width-height} $\mathsf{width}(\posA)$ is the maximum
cardinality of an antichain in~$\posA$ and $\mathsf{height}(\posA)$
is the maximum cardinality of a chain in~$\posA$.
\end{defn}

\subsubsection*{Minimal elements}

Uppercase ``$\Min$'' will denote the \emph{minimal} elements of
a set. The minimal elements are the elements that are not dominated
by any other in the set. Lowercase ``$\min$'' denotes\emph{ the
least} element, an element that dominates all others, if it exists.
(If~$\min S$ exists, then~$\Min S=\{\min S\}$.)   

The set of minimal elements of a set are an antichain, so~$\Min$
is a map from the power set $\pset(\posA)$ to the antichains~$\antichains\posA$:
\begin{align*}
\Min\colon\pset(\posA) & \rightarrow\antichains\posA,\\
S & \mapsto\{x\in S:\ (y\in S)\wedge(y\posleq x)\Rightarrow(x=y)\ \}.
\end{align*}

$\Max$ and $\max$ are similarly defined.

\subsubsection*{Upper sets}

An ``upper set'' is a subset of a poset that is closed upward.

\begin{defn}[Upper sets]
A subset $S\subseteq\posA$ is an upper set iff~$x\in S$ and~$x\posleq y$
implies~$y\in S$. 
\end{defn}
Call~$\upsets\posA$ the set of upper sets of~$\posA$. By this
definition, the empty set is an upper set: $\emptyset\in\upsets\posA$.
\begin{lem}
$\upsets\posA$ is a poset itself, with the order given by 
\begin{equation}
A\posleq_{\upsets\posA}B\qquad\equiv\qquad A\supseteq B.\label{eq:up_order}
\end{equation}
\end{lem}
Note in (\ref{eq:up_order}) the use of~``$\supseteq$'' instead
of~``$\subseteq$'', which might seem more natural. This choice
will make things easier later. 

In the poset~$\left\langle \upsets\posA,\posleq_{\upsets\posA}\right\rangle $,
the top is the empty set, and the bottom is the entire poset~$\posA$.

\subsubsection*{Order on antichains}

The upper closure operator ``$\upit$'' maps a subset of a poset
to an upper set.
\begin{defn}[Upper closure]
The operator~$\upit$ maps a subset to the smallest upper set that
includes it: 
\begin{eqnarray*}
\upit\colon\pset(\posA) & \rightarrow & \upsets\posA,\\
S & \mapsto & \{y\in\posA:\exists\,x\in S:x\posleq y\}.
\end{eqnarray*}
\end{defn}

\captionsideleft{\label{fig:antichains_upsets}}{\includegraphics[scale=0.4]{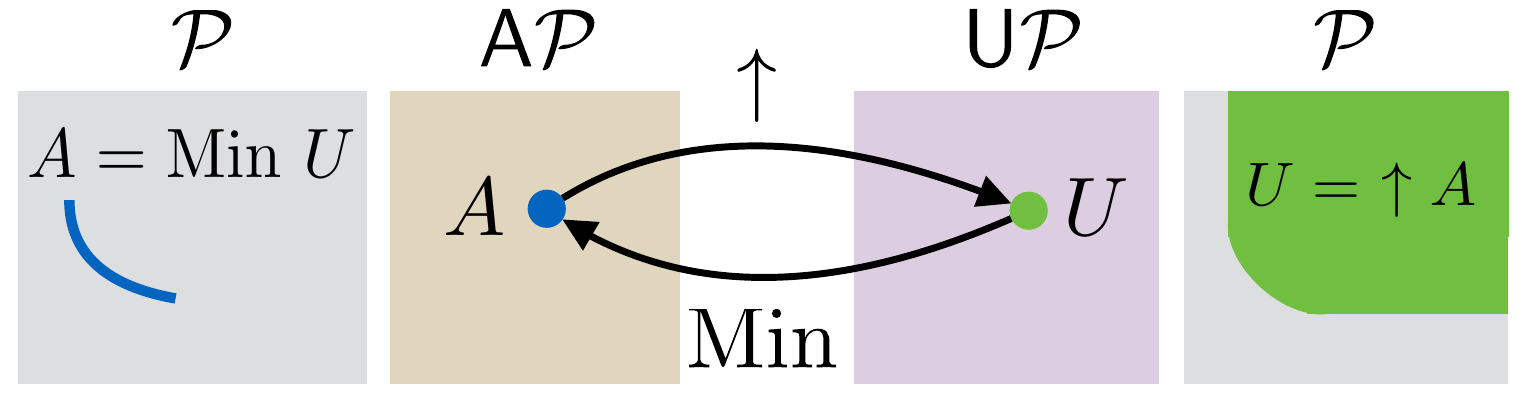}}

By using the upper closure operator, we can define an order on antichains
using the order on the upper sets~(\figref{antichains_upsets}).
\begin{lem}
\label{lem:antichains-are-poset}$\antichains\posA$ is a poset with
the relation~$\posleq_{\antichains\posA}$ defined by
\[
A\posleq_{\antichains\posA}B\qquad\equiv\qquad\upit A\supseteq\upit B.
\]
\end{lem}
In the poset $\left\langle \antichains\posA,\posleq_{\antichains\posA}\right\rangle $,
the top is the empty set:$\top_{\antichains\posA}=\emptyset.$ If
a bottom for $\posA$ exists, then the bottom for~$\antichains\posA$
is the singleton containing only the bottom for~$\posA$: $\bot_{\antichains\posA}=\{\bot_{\posA}\}.$

\subsubsection*{Monotonicity and fixed points\label{sec:Monotonicity-and-fixed}}

We will use Kleene's theorem, a celebrated result that is used in
disparate fields. It is used in computer science for defining denotational
semantics~(see, e.g.,~\cite{manes86}). It is used in embedded systems
for defining the semantics of models of computation~(see, e.g.,~\cite{lee10}).

\begin{defn}[Directed set]
A set~$S\subseteq\posA$ is \emph{directed} if each pair of elements
in~$S$ has an upper bound: for all~$a,b\in S$, there exists~$c\in S$
such that~$a\posleq c$ and~$b\posleq c$. 
\end{defn}

\begin{defn}[Completeness]
\label{def:cpo}A poset is a \emph{directed complete partial order}
(\DCPO) if each of its directed subsets has a supremum (least of
upper bounds). It is a \emph{complete partial order} (\CPO) if it
also has a bottom.

\end{defn}
\begin{example}[Completion of $\nonNegReals$ to~$\nonNegRealsComp$]
\label{exa:Rcomp}The set of real numbers~$\mathbb{R}$ is not
a \CPO, because it lacks a bottom. The nonnegative reals~$\nonNegReals=\{x\in\reals\mid x\geq0\}$
have a bottom~$\bot=0$, however, they are not a \DCPO because some
of their directed subsets do not have an upper bound. For example,
take~$\nonNegReals$, which is a subset of~$\nonNegReals$. Then~$\nonNegReals$
is directed, because for each~$a,b\in\nonNegReals$, there exists~$c=\max\{a,b\}\in\nonNegReals$
for which~$a\leq c$ and~$b\leq c$. One way to make~$\left\langle \nonNegReals,\leq\right\rangle $
a \CPO is by adding an artificial top element~$\top$, by defining~$\nonNegRealsComp\triangleq\nonNegReals\cup\{\top\},$
and extending the partial order~$\leq$ so that~$a\leq\top$ for
all~$a\in\reals^{+}$. 
\end{example}

Two properties of maps that will be important are monotonicity and
the stronger property of \scottcontinuity.
\begin{defn}[Monotonicity]
\label{def:monotone}A map~$f\colon\posA\rightarrow\posB$ between
two posets is \emph{monotone} iff~$x\posAleq y$ implies~$f(x)\posBleq f(y)$. 
\end{defn}

\begin{defn}[\scottcontinuity]
\label{def:scott}A map~$f:\posA\rightarrow\posB$ between DCPOs
is\textbf{ }\emph{\scottcontinuous{}}\textbf{ }iff for each directed
subset~$D\subseteq\posA$, the image~$f(D)$ is directed, and $f(\sup D)=\sup f(D).$
\end{defn}
\begin{rem}
\scottcontinuity implies monotonicity.
\end{rem}

\begin{rem}
\scottcontinuity does not imply topological continuity. A map from
the CPO $\langle\Rcomp,\leq\rangle$ to itself is \scottcontinuous
iff it is nondecreasing and left-continuous. For example, the ceiling
function $x\mapsto\left\lceil x\right\rceil $~ is \scottcontinuous
(\figref{ceil}).
\end{rem}
\captionsideleft{\label{fig:ceil}}{\includegraphics[scale=0.33]{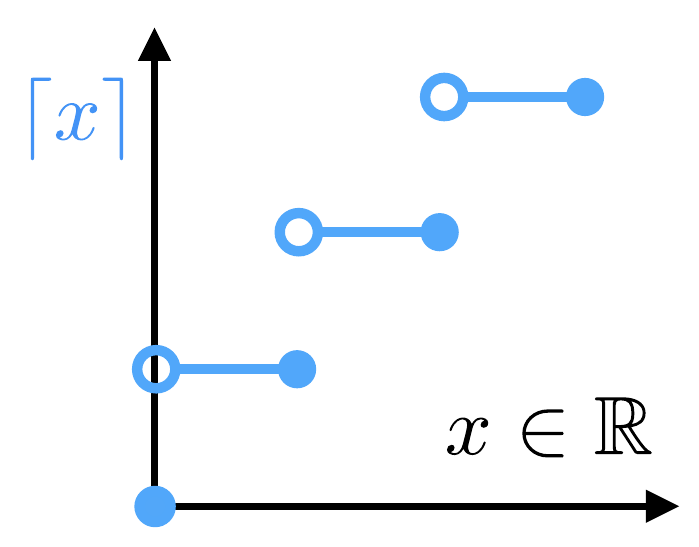}}

\emph{}

A \emph{fixed} \emph{point} of $f:\posA\rightarrow\posA$ is a point~$x$
such that $f(x)=x$. 
\begin{defn}
A \emph{least fixed point} of~$f:\posA\rightarrow\posA$ is the minimum
(if it exists) of the set of fixed points of~$f$:
\begin{equation}
\lfp(f)\,\,\doteq\,\,\min_{\posleq}\,\{x\in\posA\colon f(x)=x\}.\label{eq:lfp-one}
\end{equation}
The equality in \eqref{lfp-one} can be relaxed to ``$\posleq$''.
\end{defn}
The least fixed point need not exist. Monotonicity of the map~$f$
plus completeness is sufficient to ensure existence.
\begin{lem}[{\cite[CPO Fixpoint Theorem II, 8.22]{davey02}}]
\label{lem:CPO-fix-point-2}If~$\posA$ is a \CPO and~$f:\posA\rightarrow\posA$
is monotone, then $\lfp(f)$ exists.
\end{lem}

With the additional assumption of \scottcontinuity, Kleene's algorithm
is a systematic procedure to find the least fixed point.
\begin{lem}[{Kleene's fixed-point theorem \cite[CPO fixpoint theorem I, 8.15]{davey02}}]
\label{lem:kleene-1}Assume $\posA$ is a \CPO, and~$f:\posA\rightarrow\posA$
is \scottcontinuous. Then the least fixed point of~$f$ is the supremum
of the Kleene ascent chain 
\[
\bot\posleq f(\bot)\posleq f(f(\bot))\posleq\cdots\posleq f^{(n)}(\bot)\leq\cdots.
\]
\end{lem}

\section{Co-Design Problems\label{sec:Design-Problems}}

The basic objects considered in this paper are ``design problems'',
of which several classes will be investigated. We start by defining
a ``design problem with implementation'', which is a tuple of ``\F{functionality}
space'', ``\I{implementation} space'', and ``\R{resources}
space'', together with two maps that describe the feasibility relations
between these three spaces~(\figref{setup}).
\begin{defn}
\label{def:design_problem}A \emph{design problem with implementation}
(DPI) is a tuple $\left\langle \funsp,\ressp,\impsp,\exc,\eval\right\rangle $
where:

\begin{itemize}
\item $\funsp$ is a poset, called \emph{\F{functionality} space};
\item $\ressp$ is a poset, called \emph{\R{resources} space};
\item $\impsp$ is a set, called \emph{\I{implementation} space};
\item the map~$\exc\colon\impsp\rightarrow\funsp$, mnemonics for ``execution'',
maps an implementation to the functionality it provides;
\item the map~$\eval\colon\impsp\rightarrow\ressp$, mnemonics for ``evaluation'',
maps an implementation to the resources it requires.
\end{itemize}
\captionsideleft{\label{fig:setup}}{\includegraphics[scale=0.33]{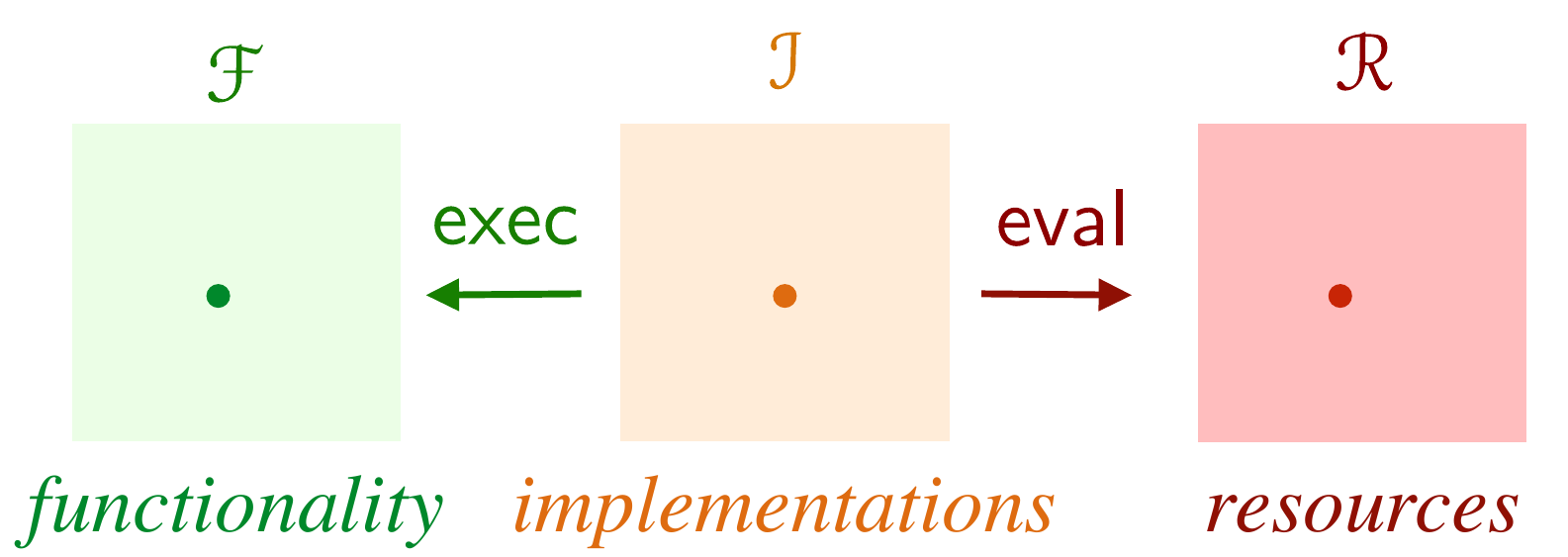}}
\end{defn}

\begin{example}[Motor design]
\label{exa:motor}Suppose we need to choose a motor for a robot from
a given set. The \emph{functionality} of a motor could be parametrized
by \F{torque} and \F{speed}. The \emph{resources} to consider
could include the \R{cost {[}\${]}}, the \R{mass {[}g{]}}, the
input \R{voltage {[}V{]}}, and the input \R{current {[}A{]}}.
The map~$\exc:\impsp\rightarrow\funsp$ assigns to each motor its
functionality, and the map~$\eval:\impsp\rightarrow\ressp$ assigns
to each motor the resources it needs~(\figref{motor}).
\end{example}
\captionsideleft{\label{fig:motor_evalexec}}{\includegraphics[scale=0.33]{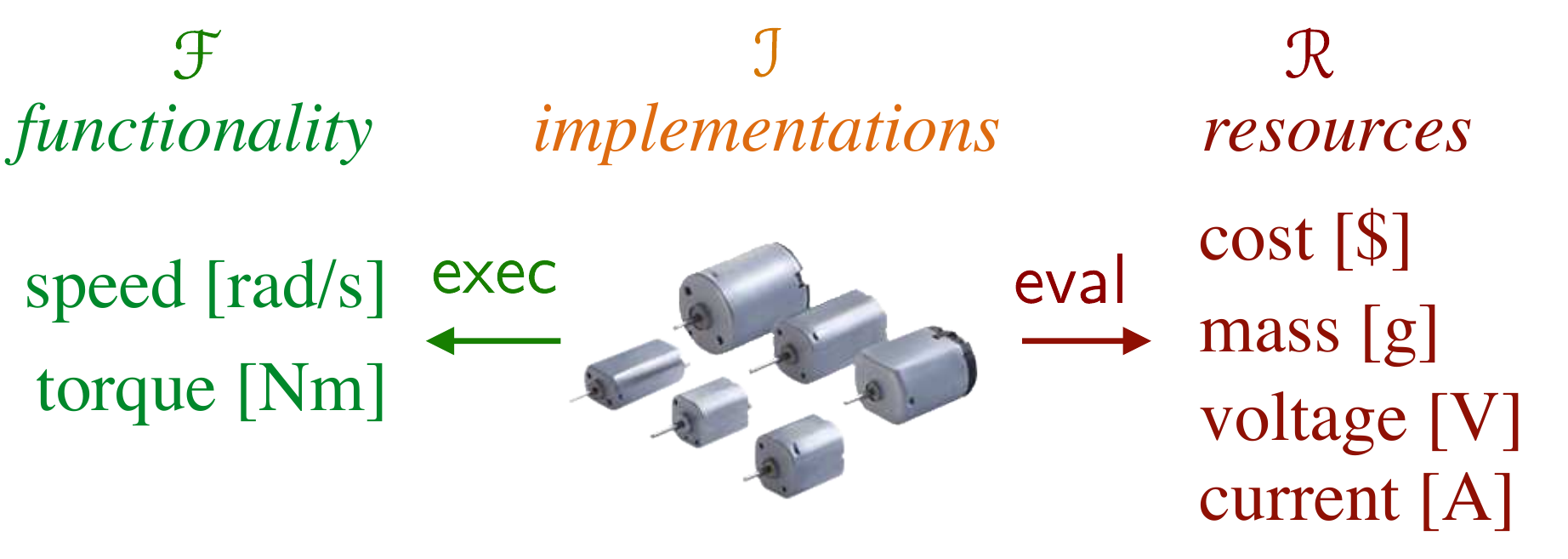}}
\begin{example}[Chassis design]
\label{exa:chassis}Suppose we need to choose a chassis for a robot~(\figref{gmcdp_chassis_eval}).
The implementation space~$\impsp$ could be the set of all chassis
that could ever be designed (in case of a theoretical analysis), or
just the set of chassis available in the catalogue at hand (in case
of a practical design decision). The functionality of a chassis could
be formalized as ``the ability to transport a certain \F{payload
{[}g{]}}'' and ``at a given \F{speed {[}m/s{]}}''. More refined
functional requirements would include maneuverability, the cargo volume,
etc. The resources to consider could be the \R{cost {[}\${]}} of
the chassis; the total mass; and, for each motor to be placed in the
chassis, the required \R{speed {[}rad/s{]}} and \R{torque {[}Nm{]}}.
\end{example}
\captionsideleft{\label{fig:gmcdp_chassis_eval}}{\includegraphics[scale=0.33]{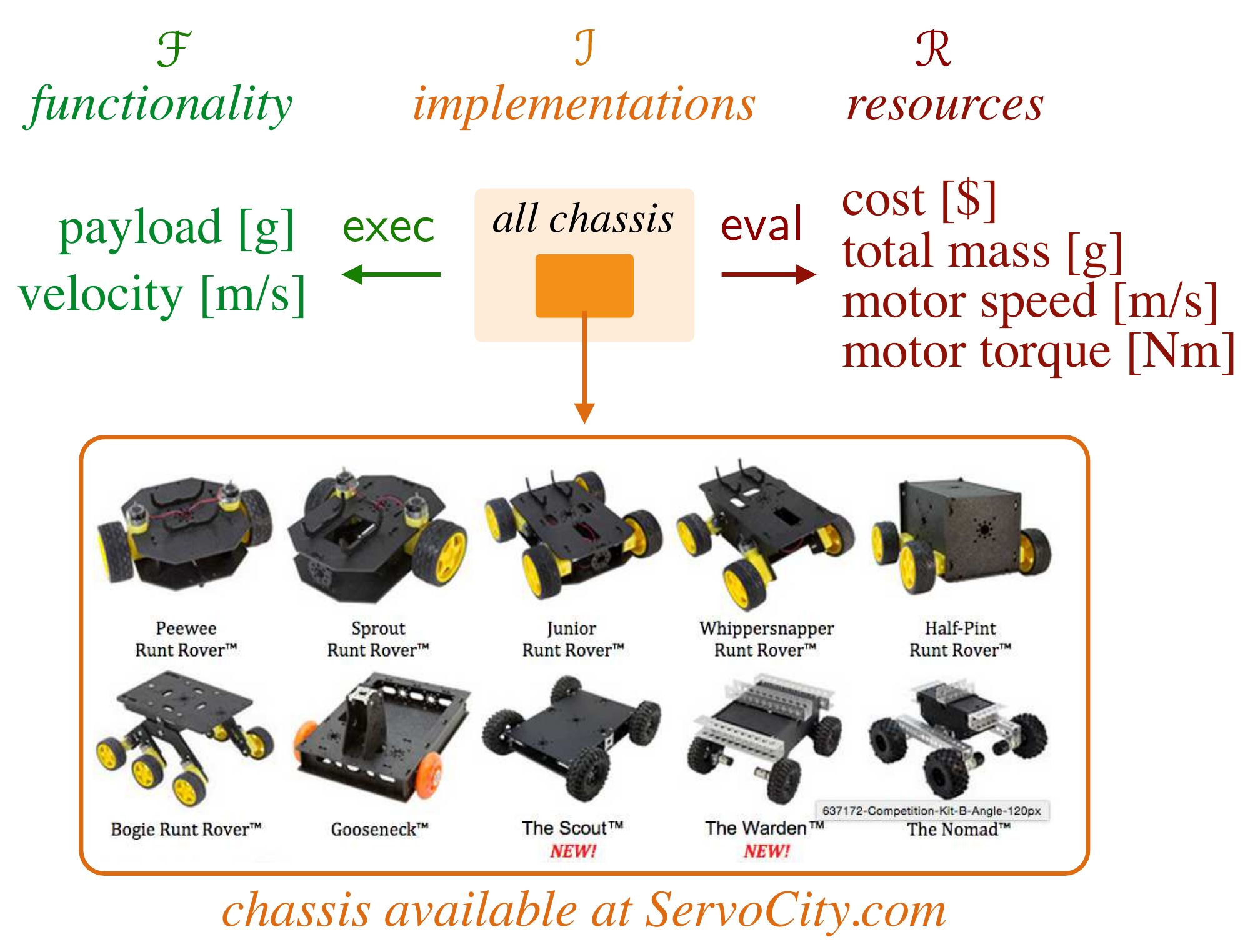}}

\subsubsection{Querying a DPI}

A DPI is a model that induces a family of optimization problems, of
the type ``Given a lower bound on the functionality~$\fun$, what
are the implementations that have minimal resources usage?''~(\figref{setup-1}).
\begin{problem}
\label{prob:problem1}Given~$\fun\in\funsp$, find the implementations
in~$\impsp$ that realize the functionality~$\fun$ (or higher)
with minimal resources, or provide a proof that there are none:
\begin{equation}
\begin{cases}
\with & \imp\in\impsp,\\
\Min_{\resleq} & \res,\\
\subto & \res=\eval(\imp),\\
 & \fun\funleq\exc(\imp).
\end{cases}\label{eq:objective}
\end{equation}
\end{problem}
\captionsideleft{\label{fig:setup-1}}{\includegraphics[scale=0.33]{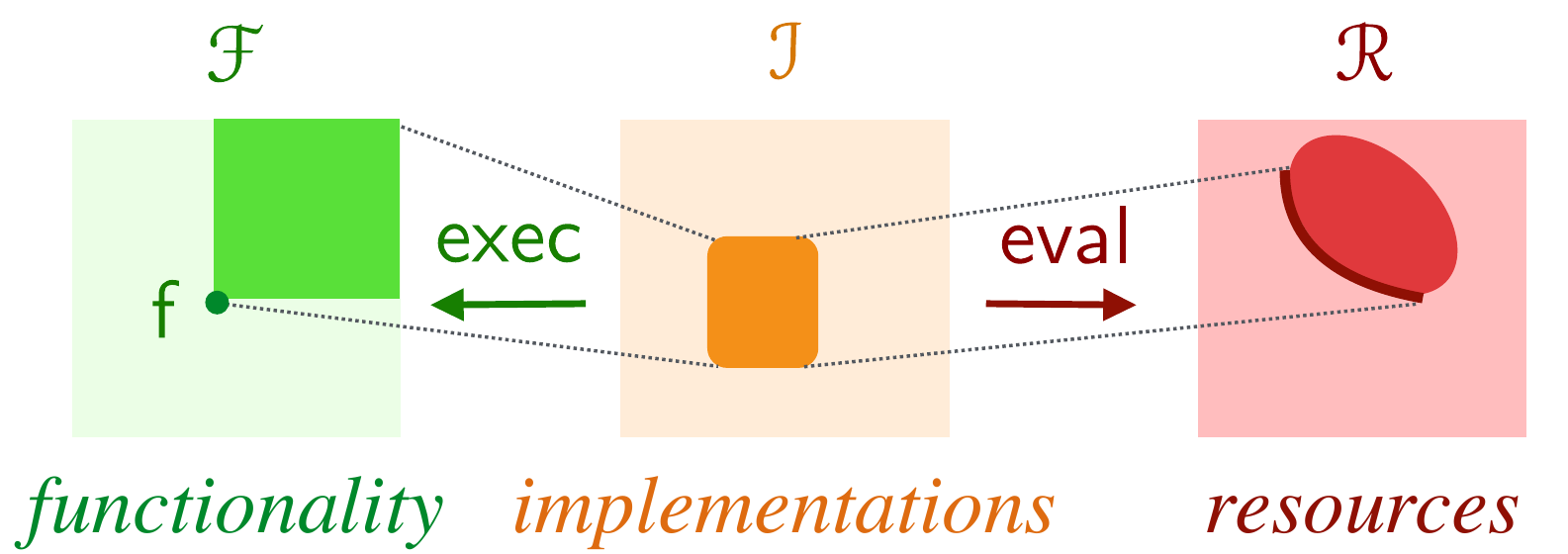}}

\begin{rem}[Minimal \emph{vs} least solutions]
Note the use of~``$\Min_{\resleq}$'' in~(\ref{eq:objective}),
which indicates the set of minimal (non-dominated) elements according
to~$\resleq$, rather than~``$\min_{\resleq}$'', which would
presume the existence of a least element. In all problems in this
paper, the goal is to find the optimal trade-off of resources (``Pareto
front''). So, for each~$\fun$, we expect to find an antichain~${\colR R}\in\Aressp$.
We will see that this formalization allows an elegant way to treat
multi-objective optimization. The algorithm to be developed will directly
solve for the set~${\colR R}$, without resorting to techniques such
as \emph{scalarization}, and therefore is able to work with arbitrary
posets, possibly discrete.
\end{rem}

\begin{rem}[Dual formulation]
In an entirely symmetric fashion, we could fix an upper bound on
the resources usage, and then maximize the functionality provided~(\figref{setup_max_f}).
The formulation is entirely dual, in the sense that it is obtained
from \eqref{objective} by swapping~$\Min$ with~$\Max$, $\funsp$~with~$\ressp$,
and $\exc$~with~$\eval$.
\begin{equation}
\begin{cases}
\with & \imp\in\impsp,\\
\Max_{\funleq} & \fun,\\
\subto & \fun=\exc(\imp),\\
 & \res\posgeq_{\ressp}\eval(\imp).
\end{cases}\label{eq:objective-1}
\end{equation}
\end{rem}
\captionsideleft{\label{fig:setup_max_f}}{\includegraphics[scale=0.33]{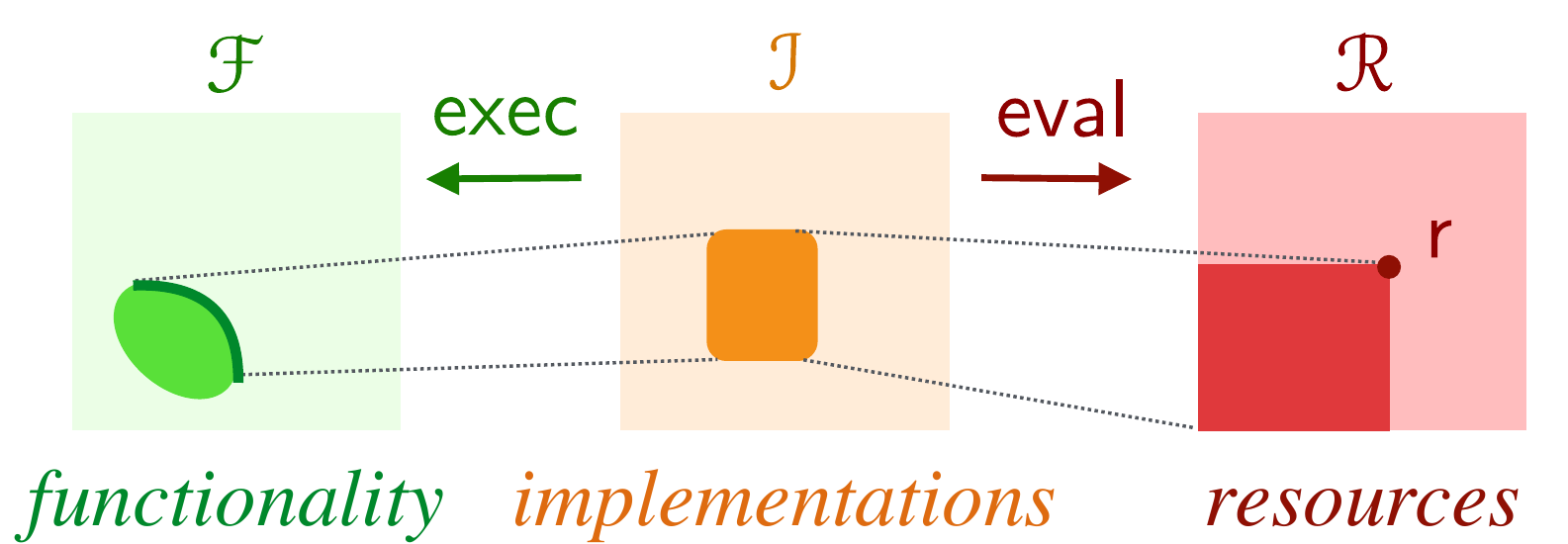}}

\subsubsection{The functionality-to-minimal resources map $\ftor$}

It is useful to also describe a design problem as a map from functionality
to sets of resources that abstracts over implementations. 

(A useful analogy is the state space representation \emph{vs} the
transfer function representation of a linear time-invariant system:
the state space representation is richer, but we only need the transfer
function to characterize the input-output response.)
\begin{defn}
\label{def:ftor}Given a DPI $\left\langle \funsp,\ressp,\impsp,\exc,\eval\right\rangle $,
define the map~$\ftor:\funsp\rightarrow\Aressp$ that associates
to each functionality~$\fun$ the objective function of~\probref{problem1},
which is the set of minimal resources necessary to realize~$\fun$:
\begin{eqnarray*}
\ftor:\funsp & \rightarrow & \Aressp,\\
\fun & \mapsto & \resMin\{\eval(\imp)\mid\left(\imp\in\impsp\right)\,\wedge\,\left(\fun\posleq\exc(\imp)\right)\}.
\end{eqnarray*}
If a certain functionality~$\fun$ is infeasible, then $\ftor(\fun)=\emptyset$.
\end{defn}
\captionsideleft{\label{fig:setup_h-1}}{\includegraphics[scale=0.33]{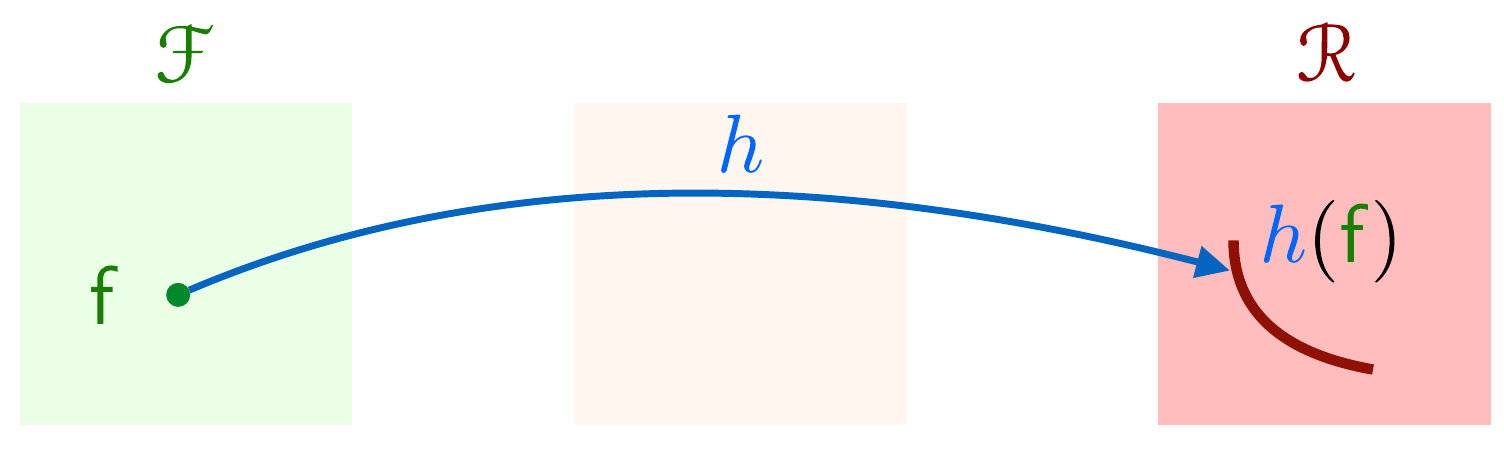}}

\begin{example}
In the case of the motor design problem, the map~$\ftor$ assigns
to each pair of $\left\langle \F{\text{speed}},\F{\text{torque}}\right\rangle $
the achievable trade-off of \R{cost}, \R{mass}, and other resources~(\figref{motor-trade-offs}).
The antichains are depicted as continuous curves, but they could also
be composed by a finite set of points.

\captionsideleft{\label{fig:motor-trade-offs}}{\includegraphics[scale=0.33]{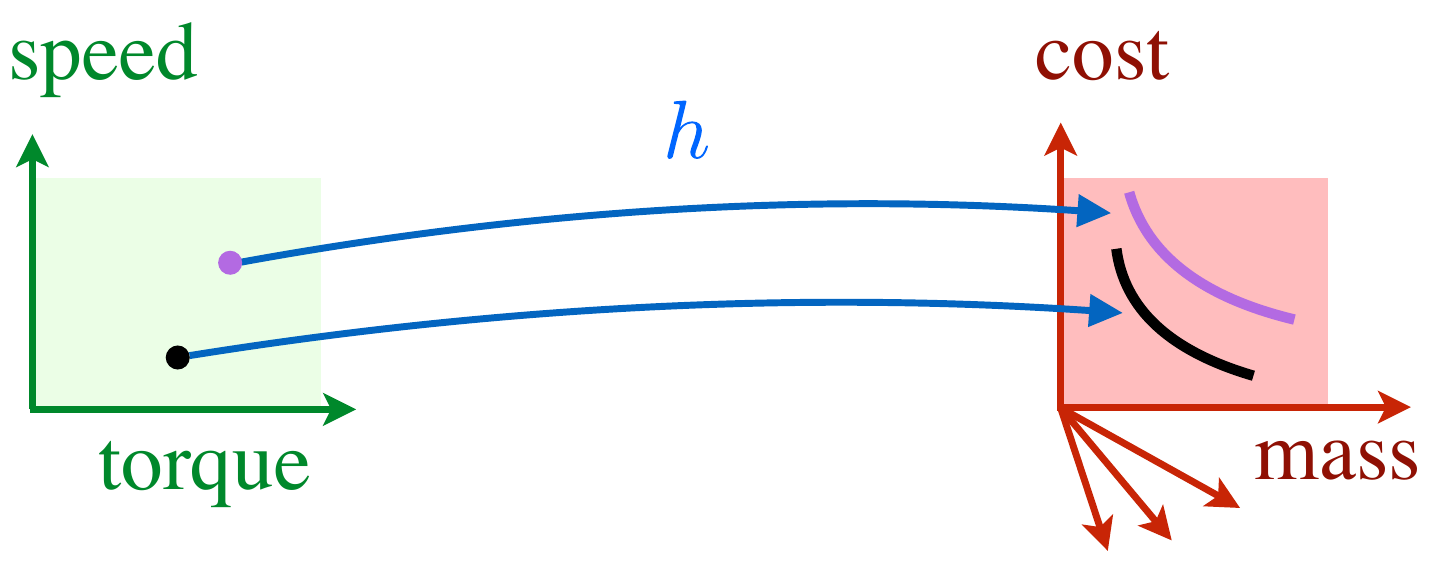}}
\end{example}

By construction, $\ftor$ is monotone (\defref{monotone}), which
means that
\[
\fun_{1}\funleq\fun_{2}\quad\Rightarrow\quad\ftor(\fun_{1})\posleq_{\Aressp}\ftor(\fun_{2}),
\]
where~$\posleq_{\Aressp}$ is the order on antichains defined in
\lemref{antichains-are-poset}. Monotonicity of~$\ftor$ means that
if the functionality~$\fun$ is increased the antichain of resources
will go ``up'' in the poset of antichains~$\Aressp$, and at some
point it might reach the top of~$\Aressp$, which is the empty set,
meaning that the problem is not feasible.

\subsubsection{Co-design problems\label{sec:Co-design-problems}}

A graphical notation will help reasoning about composition. A DPI
is represented as a box with~$\dpinumf$ green edges and~$\dpinumr$
red edges~(\figref{dp_graphical}).

\captionsideleft{\label{fig:dp_graphical}}{\includegraphics[scale=0.33]{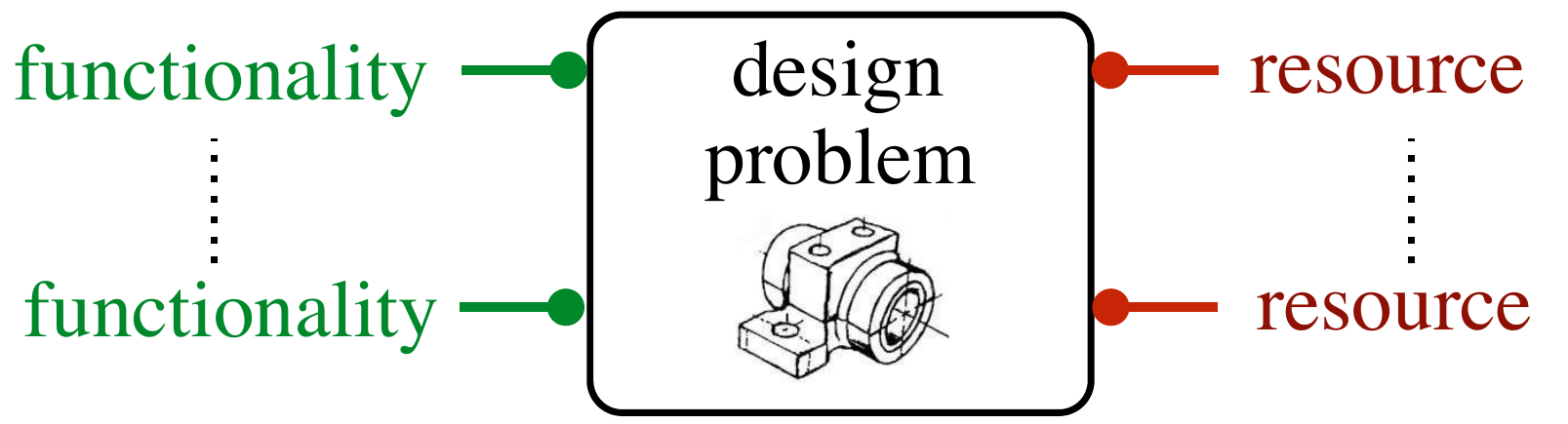}}

\noindent This means that the functionality and resources spaces
can be factorized in~$\dpinumf$ and~$\dpinumr$ components: $\funsp=\prod_{i=1}^{\dpinumf}\pi_{i}\funsp_{i},$
$\ressp=\prod_{j=1}^{\dpinumr}\pi_{j}\ressp$, where ``$\pi_{i}$''
represents the projection to the $i$-th component. If there are no
green (respectively, red) edges, then $\dpinumf$ (respectively, $\dpinumr$)
is zero, and $\funsp$ (respectively, $\ressp$) is equal to~$\One=\{\left\langle \right\rangle \}$,
the set containing one element, the empty tuple~$\left\langle \right\rangle $.

These \emph{co-design diagrams} are not to be confused with signal
flow diagrams, in which the boxes represent oriented systems and the
edges represent signals.

A ``co-design problem'' will be defined as a multigraph of design
problems. Graphically, one is allowed to connect only edges of different
color. This interconnection is indicated with the symbol~``$\posleq$''
in a rounded box~(\figref{connection}). 

\captionsideleft{\label{fig:connection}}{\includegraphics[scale=0.33]{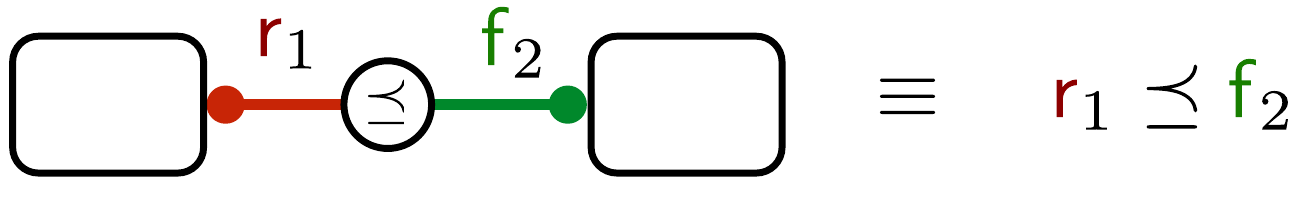}}

\noindent The semantics of the interconnection is that the resources
required by the first DPI are provided by the second DPI. This is
a partial order inequality constraint of the type~$\res_{1}\posleq\fun_{2}$.

\begin{defn}
\label{def:cdpi}A \emph{Co-Design Problem with Implementation} (CDPI)
is a tuple $\left\langle \funsp,\ressp,\left\langle \cdpiN,\mathcal{E}\right\rangle \right\rangle ,$
where~$\funsp$ and~$\ressp$ are two posets, and~$\left\langle \cdpiN,\mathcal{E}\right\rangle $
is a\emph{ }multigraph of DPIs. Each node~$\cdpin\in\cdpiN$ is a
DPI $\cdpin=\left\langle \funsp_{\cdpin},\ressp_{\cdpin},\impsp_{\cdpin},\exc_{\cdpin},\eval_{\cdpin}\right\rangle .$
An edge~$e\in\mathcal{E}$ is a tuple $e=\left\langle \left\langle \cdpinA,\cdpiresindA\right\rangle ,\left\langle \cdpinB,\cdpifunindB\right\rangle \right\rangle $,
where~$\cdpinA,\cdpinB\in\cdpiN$ are two nodes and~$\cdpiresindA$
and~$\cdpifunindB$ are the indices of the components of the functionality
and resources to be connected, and it holds that~$\pi_{\cdpiresindA}\ressp_{\cdpinA}=\pi_{\cdpifunindB}\funsp_{\cdpinB}$~(\figref{mcdps}). 

\captionsideleft{\label{fig:mcdps}}{\includegraphics[scale=0.33]{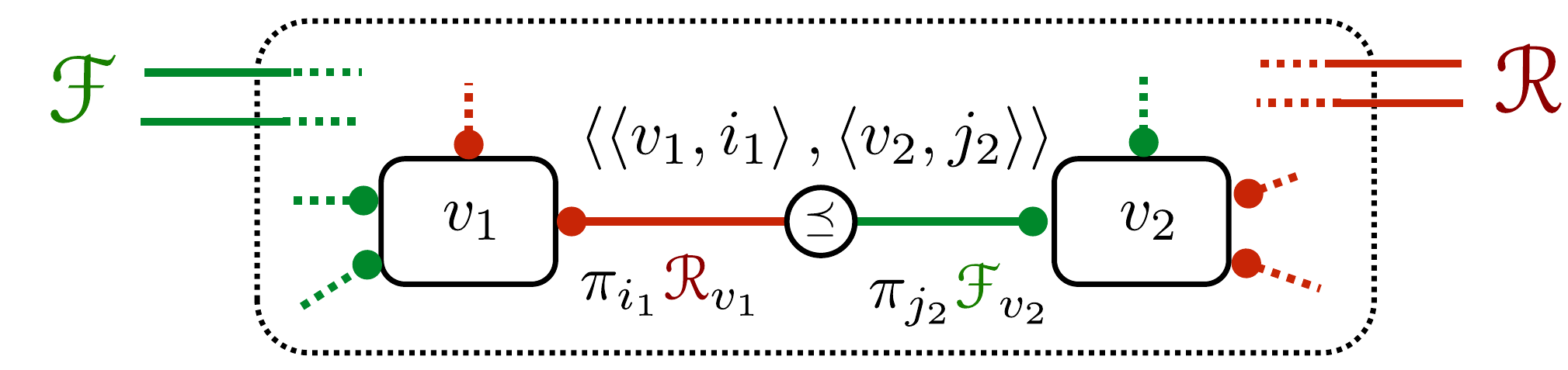}}
\end{defn}

A CDPI is equivalent to a DPI with an implementation space~$\impsp$
that is a subset of the product $\prod_{\cdpin\in\cdpiN}\impsp_{\cdpin}$,
and contains only the tuples that satisfy the co-design constraints.
An implementation tuple~$\imp\in\prod_{\cdpin\in\cdpiN}\impsp_{\cdpin}$
belongs to~$\impsp$ iff it respects all functionality\textendash resources
constraints on the edges, in the sense that, for all edges~$\left\langle \left\langle \cdpinA,\cdpiresindA\right\rangle ,\left\langle \cdpinB,\cdpifunindB\right\rangle \right\rangle $
in~$\mathcal{E}$, it holds that 
\[
\pi_{\cdpiresindA}\eval_{\cdpinA}(\pi_{\cdpinA}\imp)\posleq\pi_{\cdpifunindB}\exc_{\cdpinB}(\pi_{\cdpinB}\imp).
\]
The posets~$\funsp,\ressp$ for the entire CDPI are the products
of the functionality and resources of the nodes that remain unconnected.
For a node~$\cdpin$, let~$\unconnectedfun_{\cdpin}$ and~$\unconnectedres_{\cdpin}$
be the set of unconnected functionalities and resources. Then~$\funsp$
and~$\ressp$ for the CDPI are defined as the product of the unconnected
functionality and resources of all DPIs: $\funsp=\prod_{\cdpin\in\cdpiN}\prod_{\cdpifunind\in\unconnectedfun_{\cdpin}}\pi_{\cdpifunind}\funsp_{\cdpin}$
and $\ressp=\prod_{\cdpin\in\cdpiN}\prod_{\cdpiresind\in\unconnectedres_{\cdpin}}\pi_{\cdpiresind}\ressp_{\cdpin}.$
The maps $\exc,\eval$ return the values of the unconnected functionality
and resources:
\begin{align*}
\exc:\imp & \mapsto{\scriptstyle {\displaystyle \prod_{\cdpin\in\cdpiN}\prod_{\cdpifunind\in\unconnectedfun_{\cdpin}}}}\pi_{\cdpifunind}\exc_{\cdpin}(\pi_{\cdpin}\imp),\\
\eval:\imp & \mapsto{\displaystyle \prod_{\cdpin\in\cdpiN}\prod_{\cdpiresind\in\unconnectedres_{\cdpin}}}\pi_{\cdpiresind}\eval_{\cdpin}(\pi_{\cdpin}\imp).
\end{align*}

\begin{example}
\label{exa:chassis_plus_motor}Consider the co-design of chassis (\exaref{chassis})
plus motor (\exaref{motor}). The design problem for a motor has \F{speed}
and \F{torque} as the provided functionality (what the motor must
provide), and \R{cost}, \R{mass}, \R{voltage}, and \R{current}
as the required resources~(\figref{motor}).

\captionsideleft{\label{fig:motor}}{\includegraphics[scale=0.33]{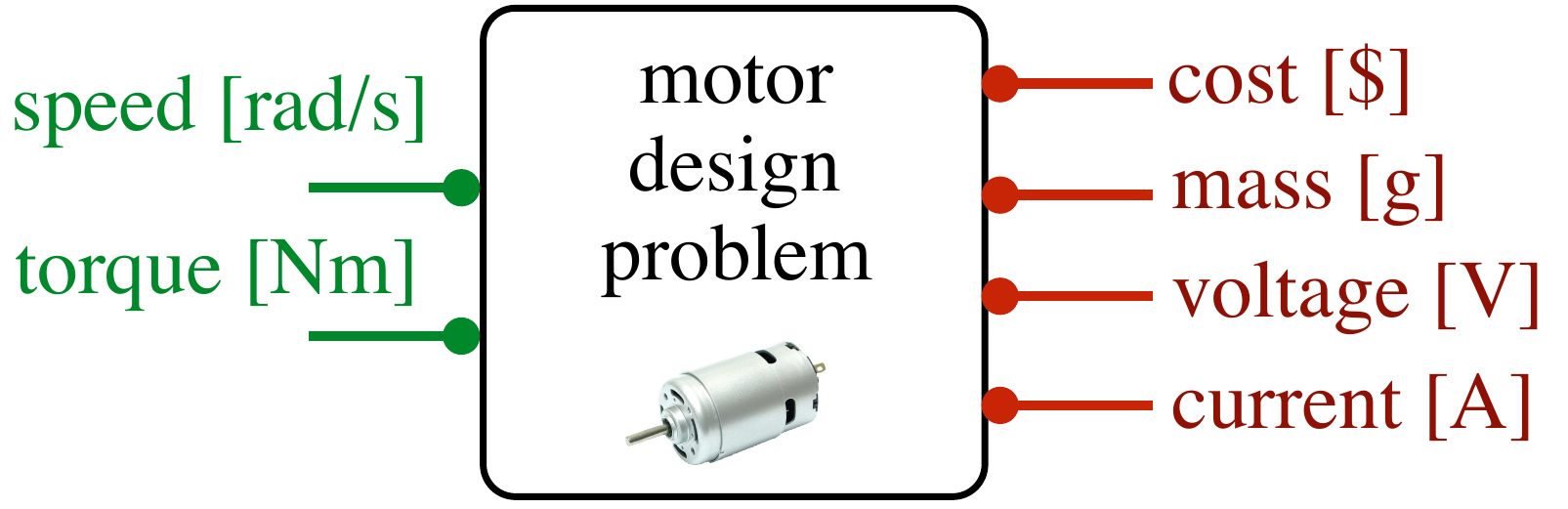}}

\noindent For the chassis (\figref{gmcdp_chassis}), the provided
functionality is parameterized by the \F{mass} of the payload and
the platform \F{velocity}. The required resources include the \R{cost},
\R{total mass}, and what the chassis needs from its motor(s), such
as \R{speed} and \R{torque}.

\captionsideleft{\label{fig:gmcdp_chassis}}{\includegraphics[scale=0.33]{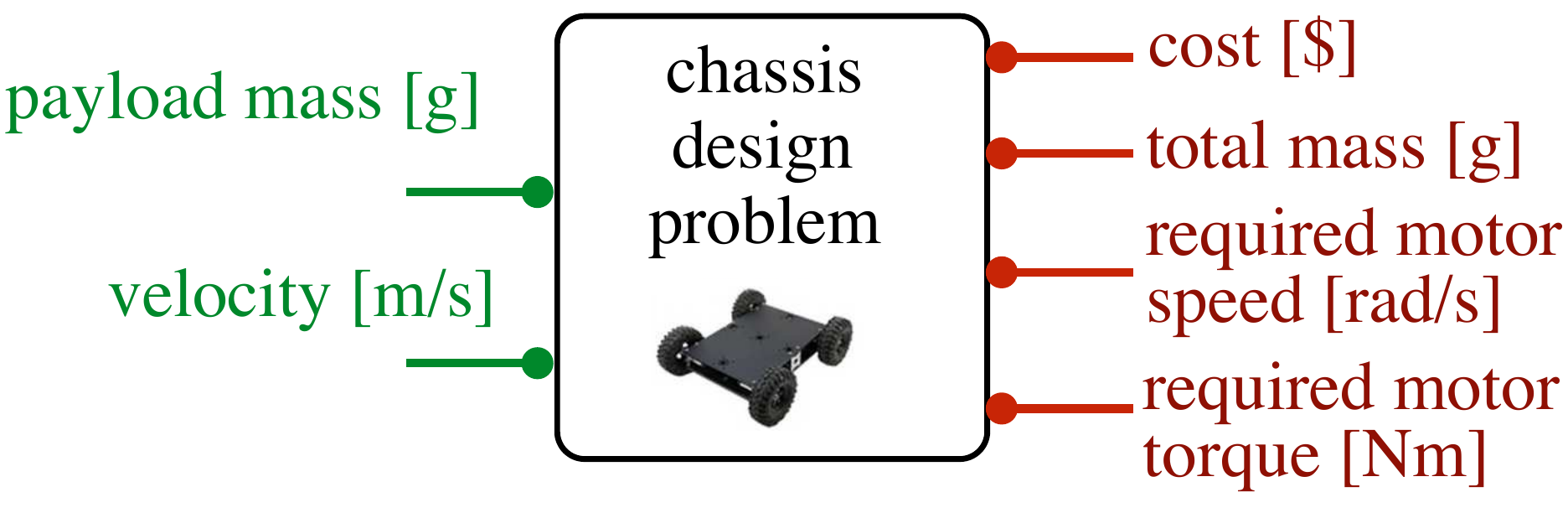}}

\noindent The two design problem can be connected at the edges for
torque and speed~(\figref{gmcdp_chassis_plus_motor_series}). The
semantics is that the motor needs to have\emph{ at least }the given
torque and speed. 

\captionsideleft{\label{fig:gmcdp_chassis_plus_motor_series}}{\includegraphics[scale=0.33]{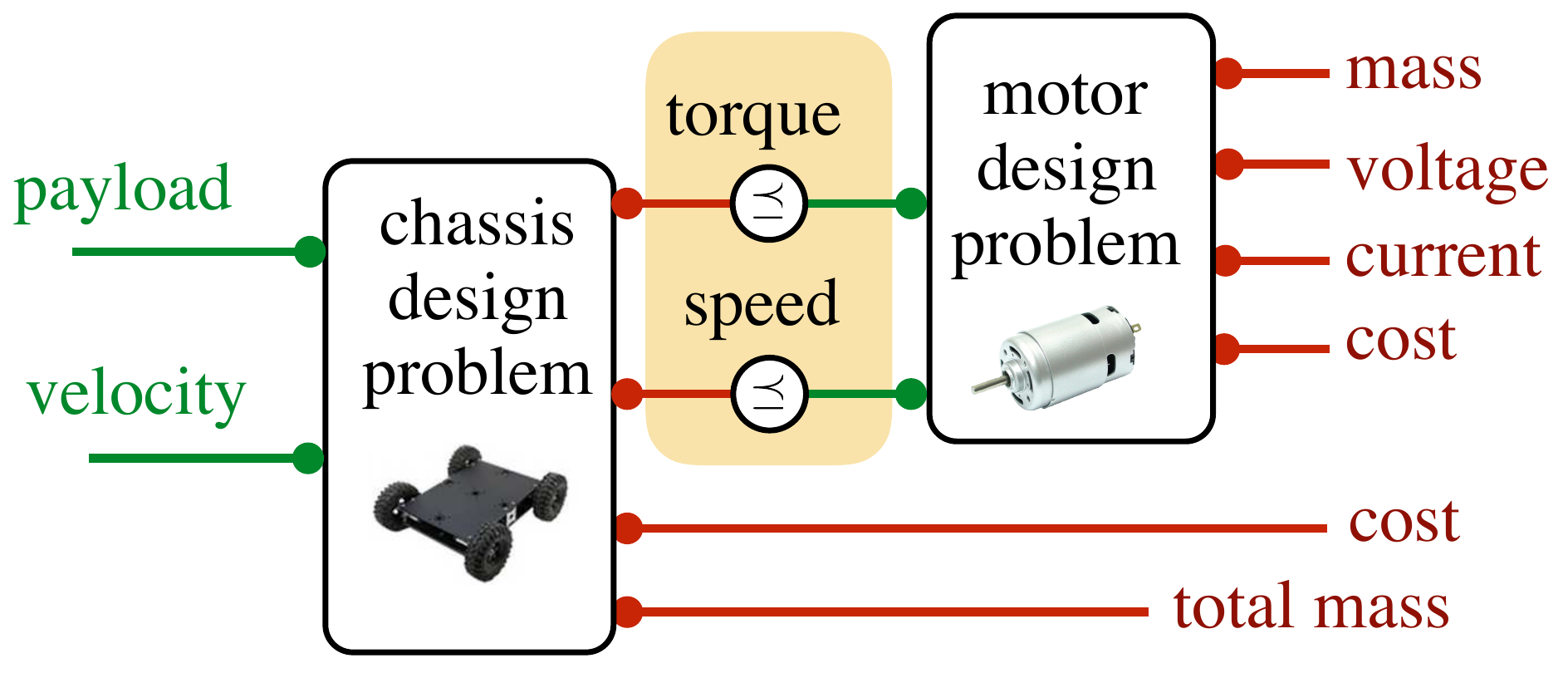}}

\noindent Resources can be summed together using a trivial DP corresponding
to the map $\ftor:\left\langle \fun_{1},\fun_{2}\right\rangle \mapsto\{\fun_{1}+\fun_{2}\}$
(\figref{total_cost}).

\captionsideleft{\label{fig:total_cost}}{\includegraphics[scale=0.33]{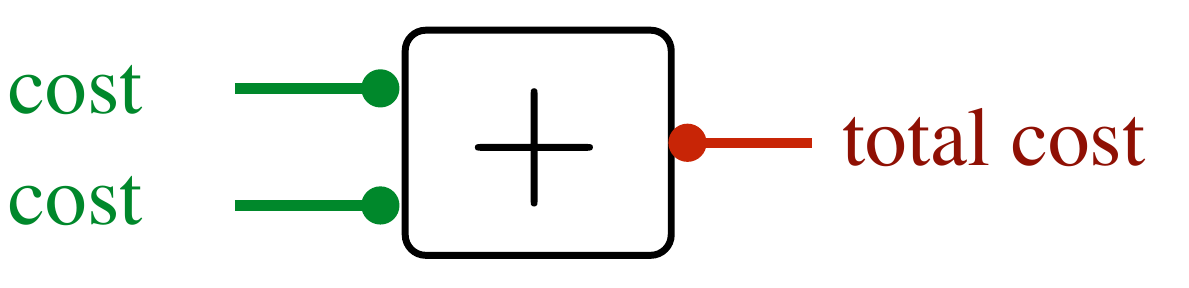}}

\noindent A co-design problem might contain recursive co-design constraints.
For example, if we set the payload to be transported to be the sum
of the motor mass plus some extra payload, a cycle appears in the
graph~(\figref{gmcdp_chassis_plus_motor}). 

\noindent 
\begin{figure}[H]
\centering{}\includegraphics[scale=0.33]{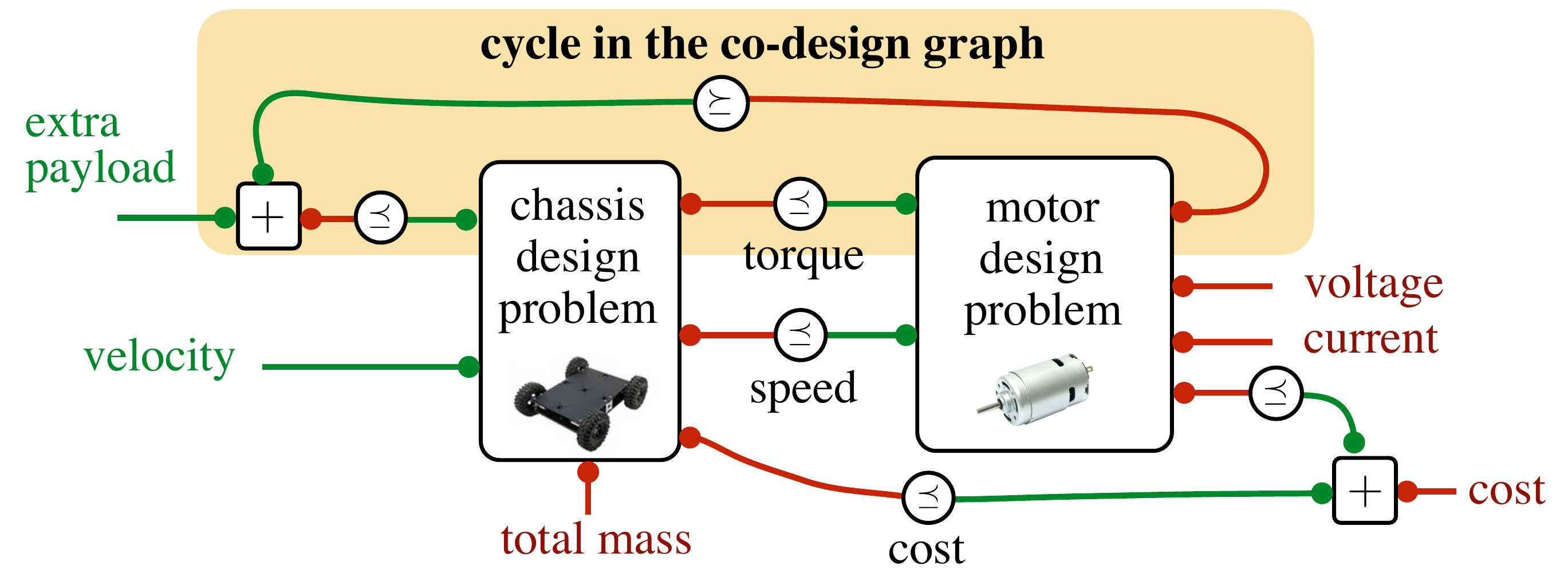}\caption{\label{fig:gmcdp_chassis_plus_motor}}
\end{figure}

\noindent This formalism makes it easy to abstract away the details
in which we are not interested. Once a diagram like~\figref{gmcdp_chassis_plus_motor}
is obtained, we can draw a box around it and consider the abstracted
problem~(\figref{gmcdp_chassis_plus_motor-1}). 

\captionsideleft{\label{fig:gmcdp_chassis_plus_motor-1}}{\includegraphics[scale=0.33]{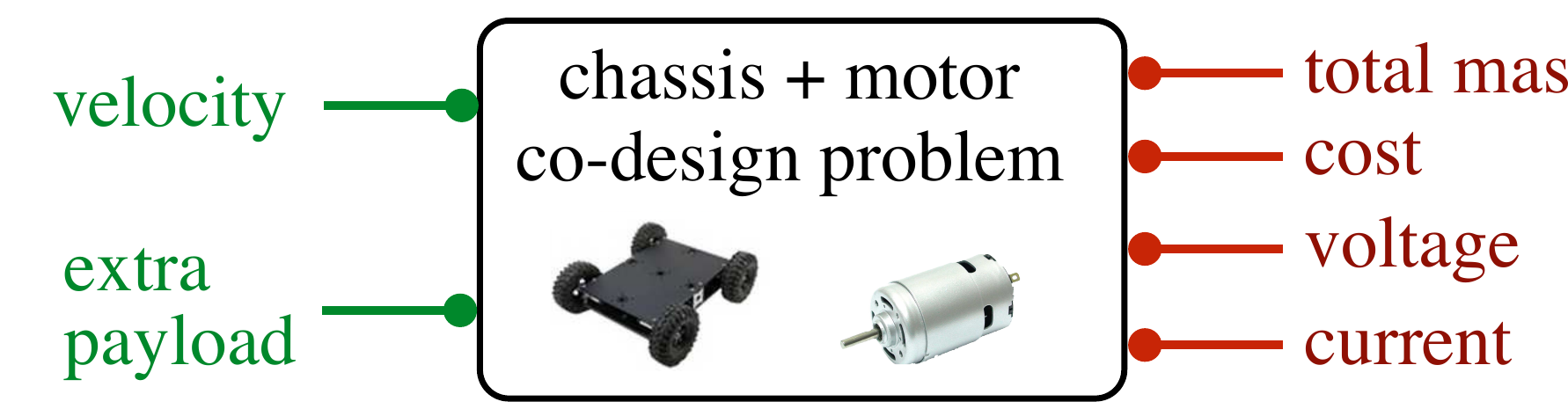}}

\label{exa:finish}Let us finish assembling our robot. A motor needs
a motor control board. The functional requirements are the (peak)
\F{output current} and the \F{output voltage range}~(\figref{mcb}).

\captionsideleft{\label{fig:mcb}}{\includegraphics[scale=0.33]{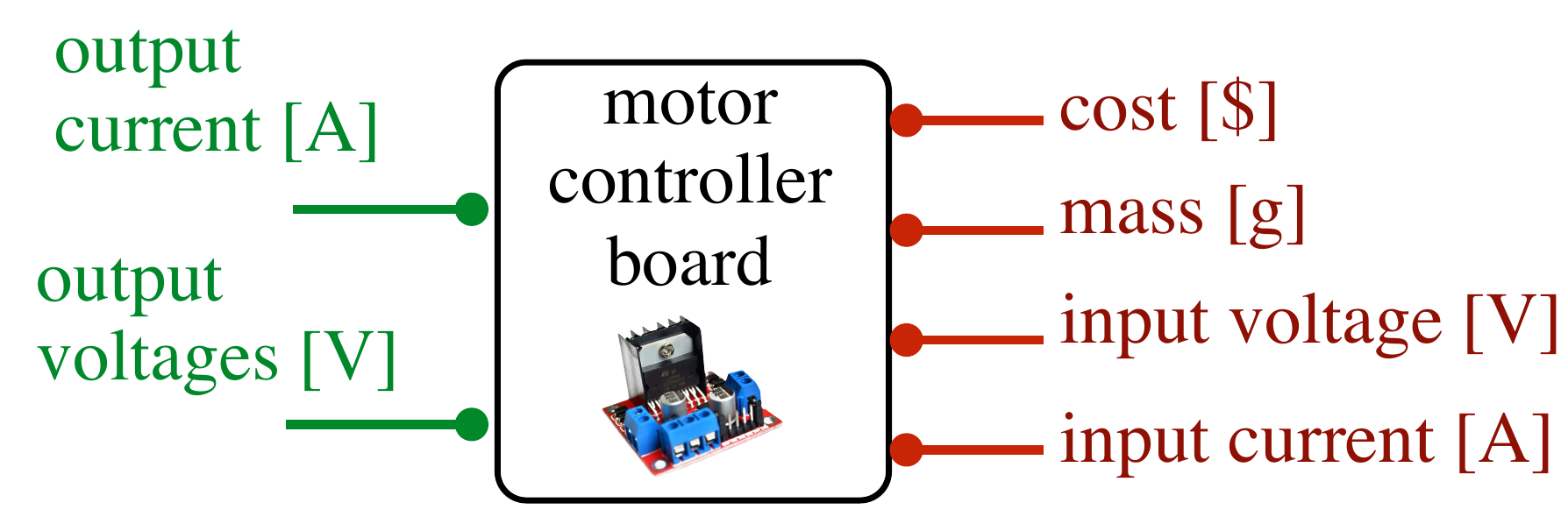}}

\noindent The functionality for a power supply could be parameterized
by the \F{output current}, the \F{output voltages}, and the \F{capacity}.
The resources could include \R{cost} and \R{mass} (\figref{example-ba}).

\captionsideleft{\label{fig:example-ba}}{\includegraphics[scale=0.33]{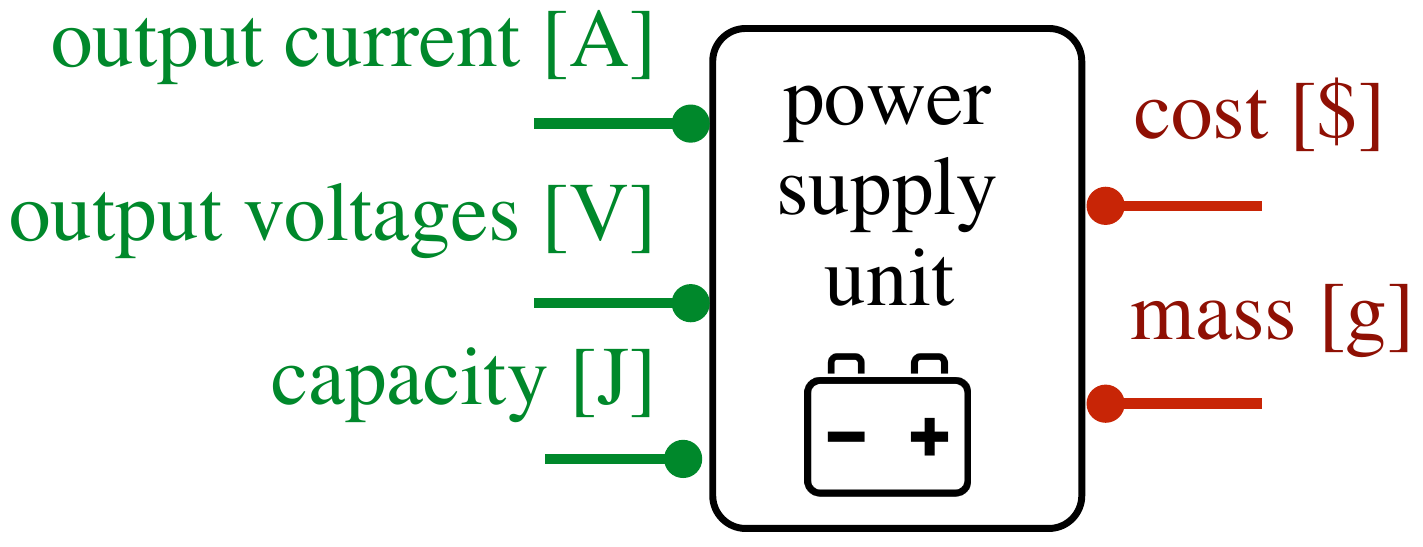}} 

\noindent 

\noindent Relations such as ${\colF\mbox{current}}\times{\colF\mbox{voltage}}\leq{\colR\mbox{power required}}$
and ${\colF\mbox{power}}\times{\colF\mbox{endurance}}\leq{\colR\mbox{energy required}}$
can be modeled by a trivial ``multiplication'' DPI (\figref{current_times_voltage}).

\captionsideleft{\label{fig:current_times_voltage}}{\includegraphics[scale=0.33]{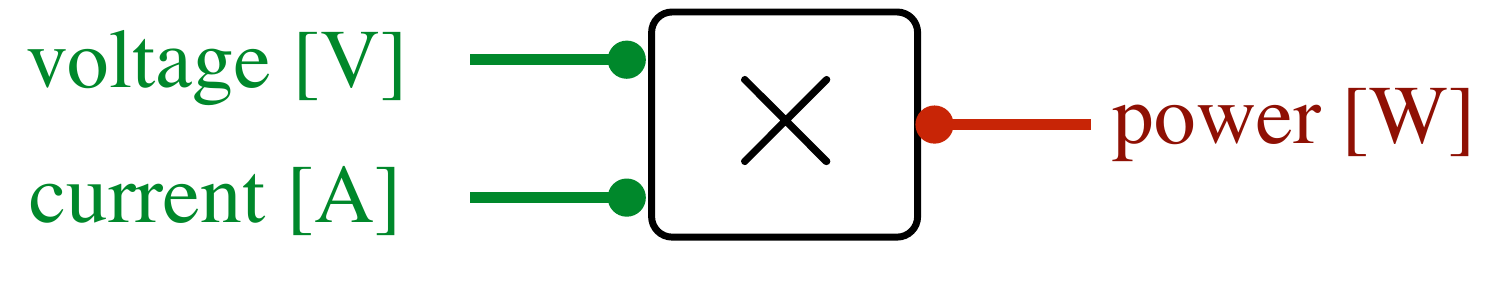}}

\noindent We can connect these DPs to obtain a co-design problem with
functionality \F{voltage}, \F{current}, \F{endurance} and resources
\R{mass} and \R{cost}~(\figref{connect}).

\captionsideleft{\label{fig:connect}}{\includegraphics[scale=0.29]{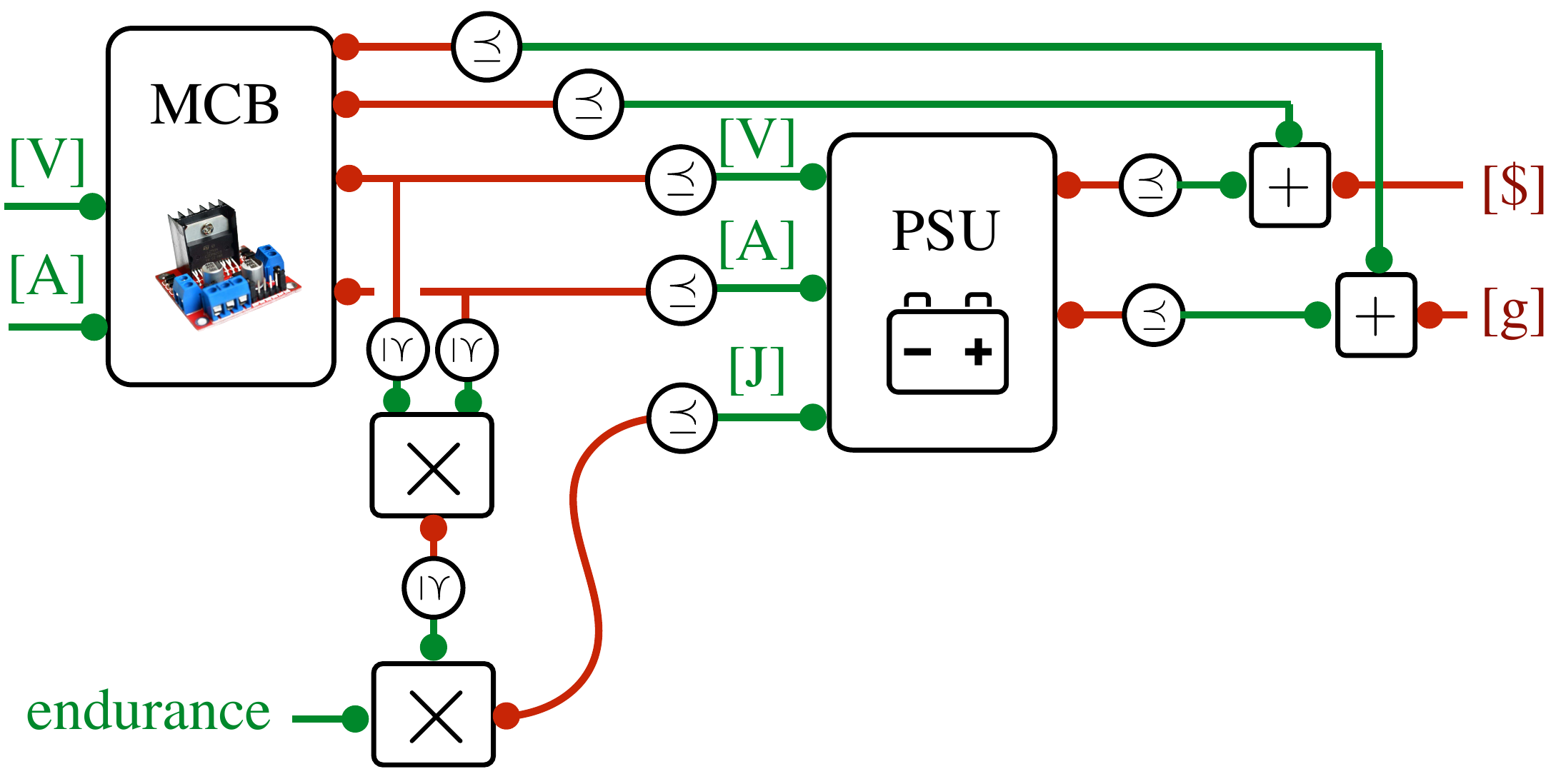}}

\noindent Draw a box around the diagram, and call it ``MCB+PSU'';
then interconnect it with the ``chassis+motor'' diagram in~\figref{another}.

\noindent 
\begin{figure}[H]
\begin{centering}
\includegraphics[scale=0.33]{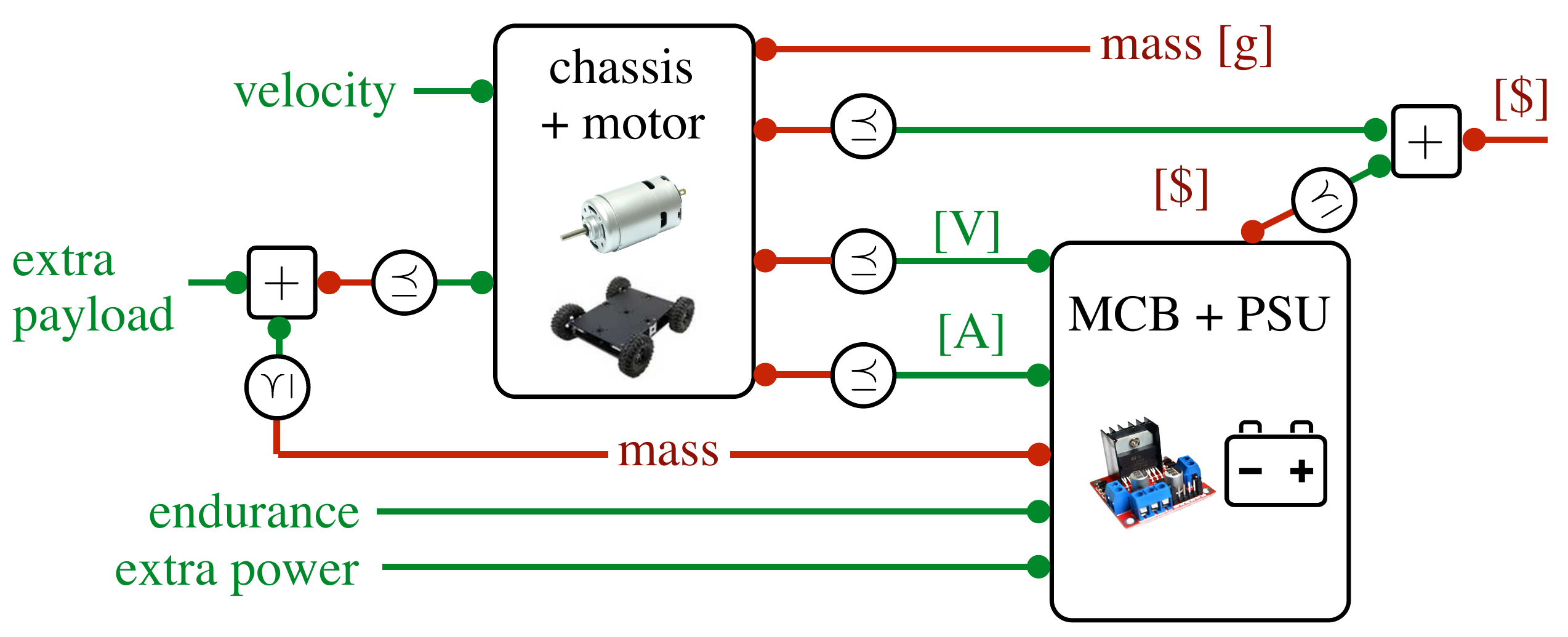}
\par\end{centering}
\caption{\label{fig:another}}
\end{figure}

We can further abstract away the diagram in~\figref{another} as
a ``mobility+power'' CDPI, as in \figref{shipping}. The formalism
allows to consider \R{mass} and \R{cost} as independent resources,
meaning that we wish to obtain the Pareto frontier for the minimal
resources. Of course, one can always reduce everything to a scalar
objective. For example, a conversion from mass to cost exists and
it is called ``shipping''. Depending on the destination, the conversion
factor is between~$\$0.5/\mbox{lbs}$, using USPS, to~$\$10\mbox{k}/\mbox{lbs}$
for sending your robot to low Earth orbit. 

\noindent 
\begin{figure}[H]
\centering{}\includegraphics[scale=0.33]{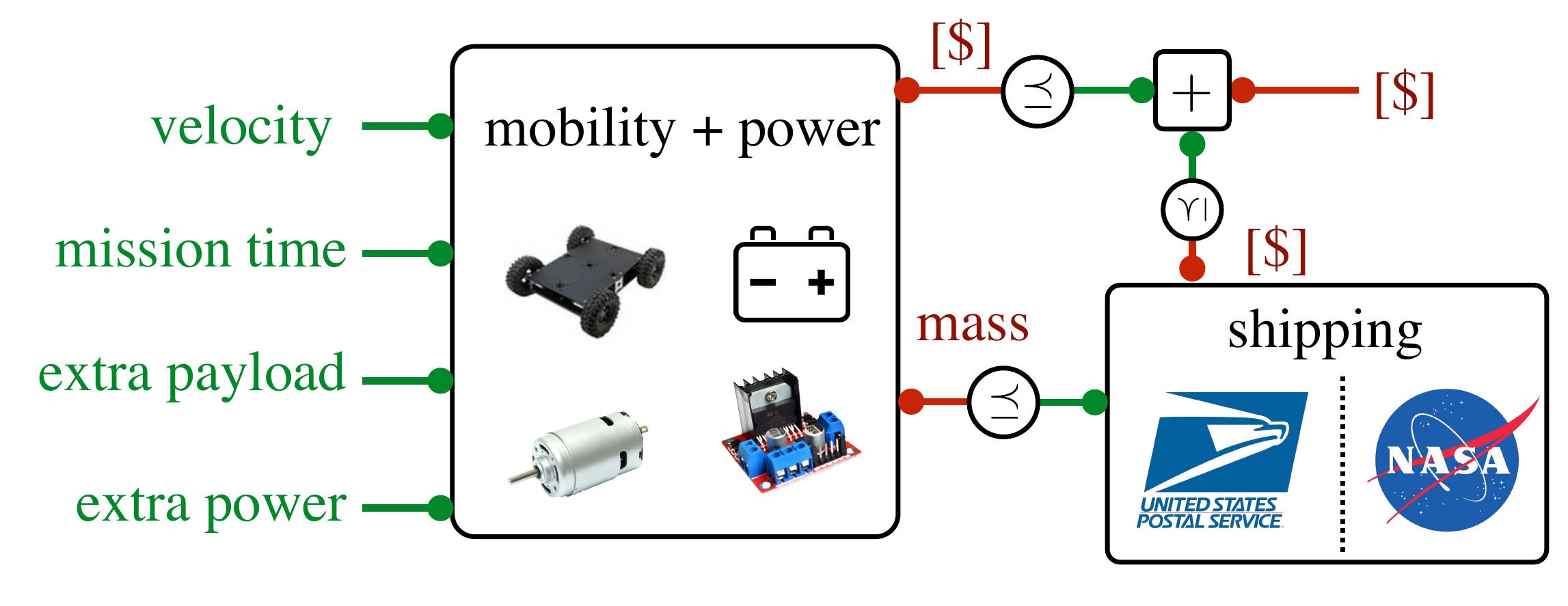}\caption{\label{fig:shipping}}
\end{figure}

\end{example}

\subsubsection*{Examples from the literature}

Many recent works in robotics and neighboring fields that deal with
minimality and resource constraints can be incorporated in this framework.

\begin{example}
Svorenova\,\,\etal~\cite{svorenova16resource} consider a joint
sensor scheduling and control synthesis problem, in which a robot
can decide to not perform sensing to save power, given performance
objectives on the probability of reaching the target and the probability
of collision. The method outputs a Pareto frontier of all possible
operating points. This can be cast as a design problem with functionality
equal to the \F{probability of reaching the target} and (the inverse
of) \F{the collision probability}, and with resources equal to the
\R{actuation power}, \R{sensing power}, and \R{sensor accuracy}.
 
\end{example}
\captionsideleft{\label{fig:progressive-1-1}}{\includegraphics[scale=0.33]{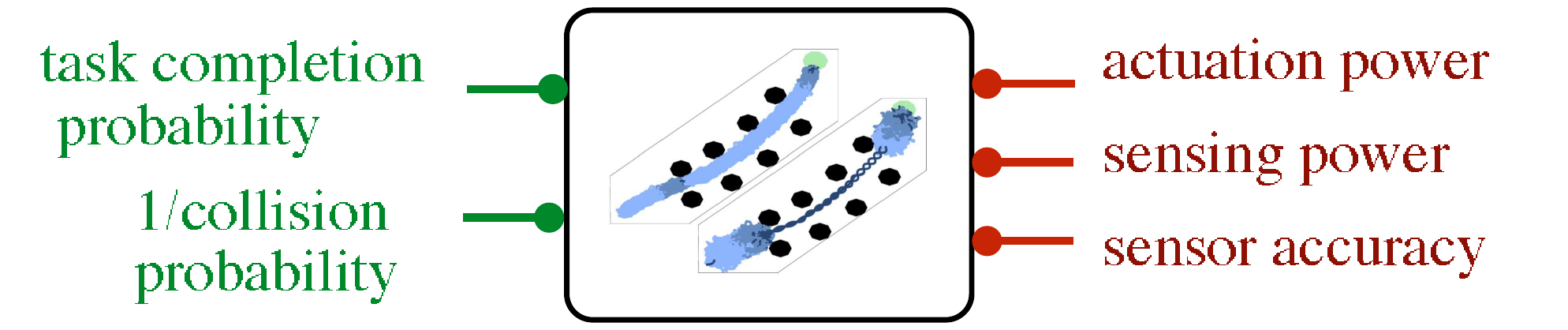}}

\begin{example}
Nardi\,\,\etal~\cite{zia16comparative} describe a benchmarking
system for visual SLAM that provides the empirical characterization
of the monotone relation between \F{the accuracy} of the visual
SLAM solution, the \F{throughput {[}frames/s{]}} and \R{the energy
for computation {[}J/frame{]}}. The implementation space is the product
of algorithmic parameters, compiler flags, and architecture choices,
such as the number of GPU cores active. This is an example of a design
problem whose functionality-resources map needs to be experimentally
evaluated.
\end{example}
\captionsideleft{}{\includegraphics[scale=0.33]{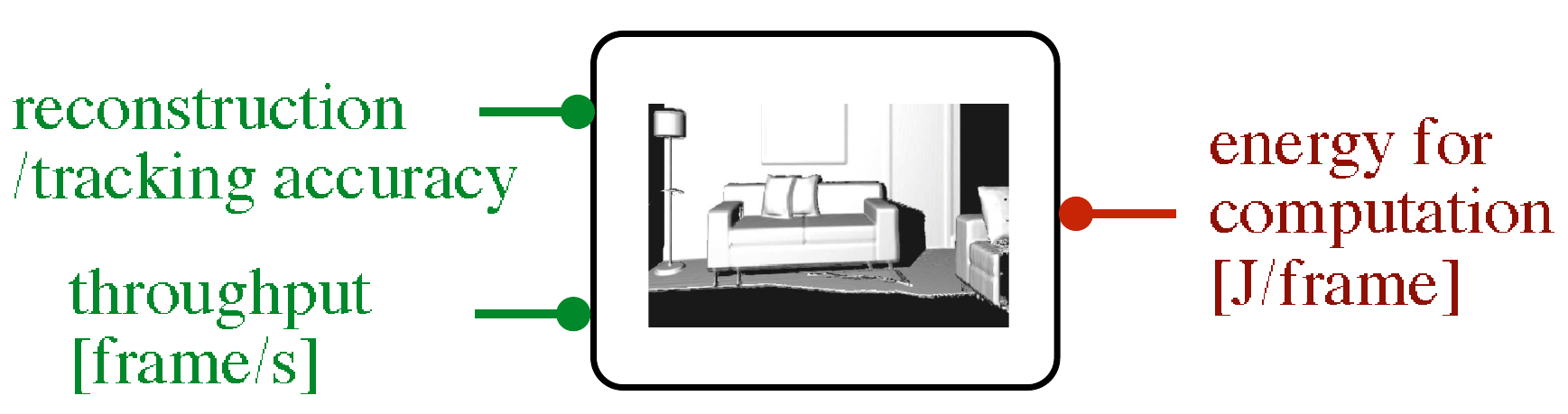}}

\subsubsection*{Other examples in minimal robotics}

Many works have sought to find ``minimal'' designs for robots, and
can be understood as characterizing the relation between the poset
of \F{tasks} and the poset of physical resources, which is the product
of \R{sensing}, \R{actuation}, and \R{computation} resources,
plus other non-physical resources, such as \R{prior knowledge}~(\figref{robot-generic}).
Given a task, there is a minimal antichain in the resources poset
that describes the possible trade-offs (e.g., compensating lousier
sensors with more computation). 

\captionsideleft{\label{fig:robot-generic}}{\includegraphics[scale=0.33]{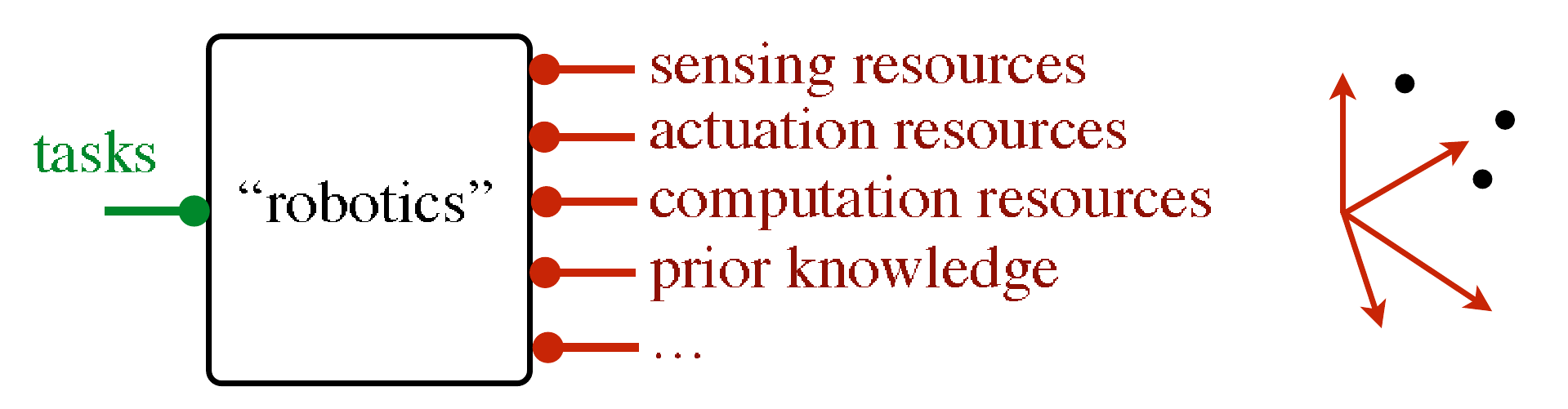}}

The poset structure arises naturally: for example, in the \emph{sensor
lattice}~\cite{lavalle12sensing}\emph{,} a sensor dominates another
if it induces a finer partition of the state space. Similar dominance
relations can be defined for actuation and computation. O'Kane and
Lavalle~\cite{okane08comparing} define a robot as a union of ``robotic
primitives'', where each primitive is an abstraction for a set of
sensors, actuators, and control strategies that can be used together
(e.g., a compass plus a contact sensor allow to ``drive North until
a wall is hit''). The effect of each primitive is modeled as an operator
on the robot's information space. It is possible to work out what
are the minimal combinations of robotic primitives (minimal antichain)
that are sufficient to perform a task (e.g., global localization),
and describe a dominance relation (partial order) of primitives. Other
works have focused on minimizing the complexity of the controller.
Egerstedt~\cite{egerstedt03motion} studies the relation between
the \F{complexity of the environment} and a notion of \R{minimum
description length of control strategies}, which can be taken as
a proxy for the computation necessary to perform the task. Soatto~\cite{soatto11steps}
studies the relation between the \F{performance of a visual task},
and the \R{ minimal representation} that is needed to perform that
task. 

The hope is that the theory of co-design presented in this paper will
help to integrate all this previous work in the same theoretical and
quantitative framework.

\section{Problem statement and summary of results\label{sec:Optimization}}

Given an arbitrary graph of design problems, and assuming we know
how to solve each problem separately, we ask whether we can solve
the \emph{co-design} problem optimally. 
\begin{problem}
\label{prob:MCDP}Suppose that we are given a CDPI~$\left\langle \funsp,\ressp,\left\langle \cdpiN,\mathcal{E}\right\rangle \right\rangle ,$
and that we can evaluate the map~$\ftor_{\cdpin}$ for all~$\cdpin\in\cdpiN$.
Given a required functionality~$\fun\in\funsp$, we wish to find
the \emph{minimal} resources in~$\ressp$ for which there exists
a feasible implementation vector that makes all sub-problems feasible
at the same time and all co-design constraints satisfied; or, if
none exist, provide a certificate of infeasibility. 
\end{problem}
In other words, given the maps~$\{\ftor_{\cdpin},\,\cdpin\in\cdpiN\}$
for the subproblems, one needs to evaluate the map~$\ftor:\funsp\rightarrow\Aressp$
for the entire CDPI (\figref{question}). 

\captionsideleft{\label{fig:question}}{\hspace{-4mm}\includegraphics[width=7cm]{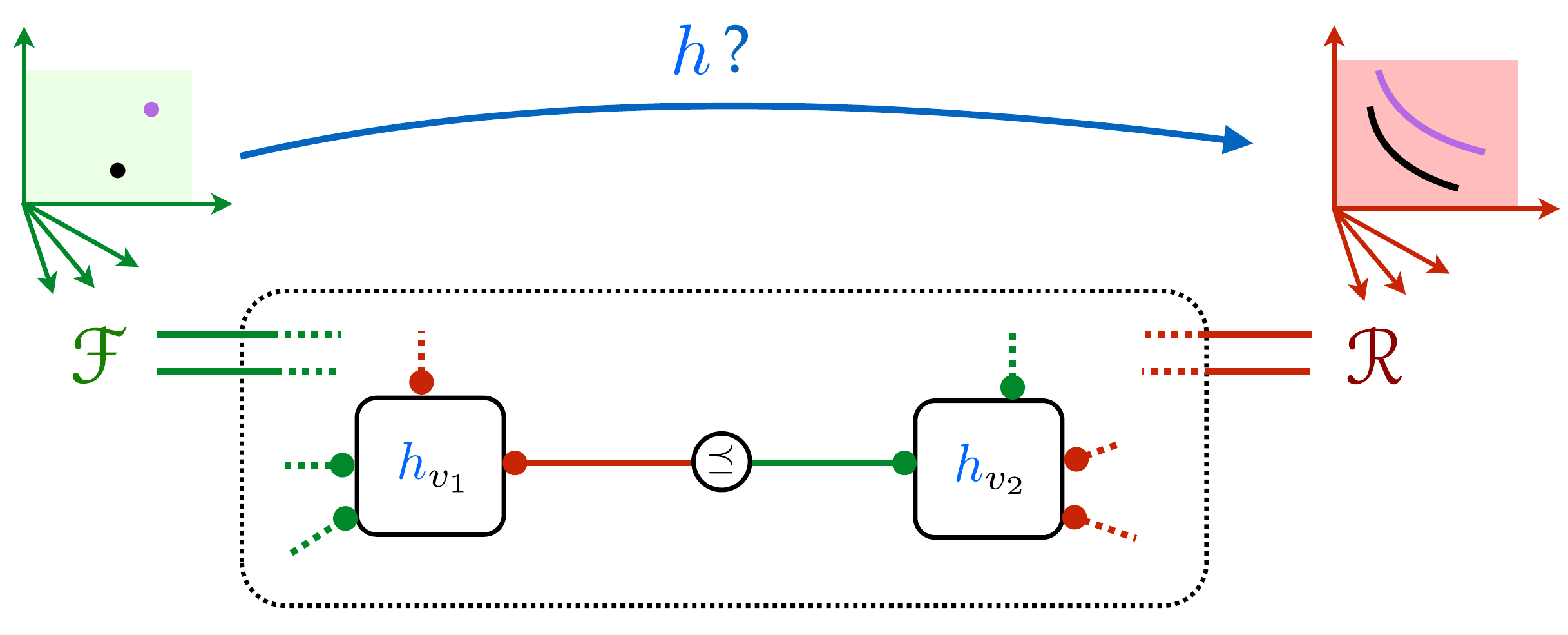}}

The rest of the paper will provide a solution to~\probref{MCDP},
under the assumption that all the DPIs inside the CDPI are ``monotone'',
in the sense of~\defref{DPI-monotone}.
\begin{defn}
\label{def:DPI-monotone}A DPI $\left\langle \funsp,\ressp,\impsp,\exc,\eval\right\rangle $
is ``monotone\emph{''} if
\begin{enumerate}
\item The posets $\funsp,\ressp$ are complete partial orders (\defref{cpo}).
\item The map~$\ftor$ is \scottcontinuous (\defref{scott}).
\end{enumerate}
\end{defn}

Call ``Monotone Co-Design Problems'' (MCDPs) the set of CDPIs for
which all subproblems respect the conditions in~\defref{DPI-monotone}.
 I will show two main results:

1) A \textbf{modeling result} (Theorem~\vref{thm:CDP-monotone})
says that the class of MCDPs is closed with respect to arbitrary interconnections.
Therefore, given a co-design diagram, such as the one in~\figref{shipping},
if we know that each design problem is an MCDP, we can conclude that
the diagram represents an MCDP as well.

2) An \textbf{algorithmic result }(Theorem~\vref{thm:CDP-solvig})
says that the functionality-resources map~$\ftor$ for the entire
MCDP has an explicit expression in terms of the maps~$\{\ftor_{\cdpin},\,\cdpin\in\cdpiN\}$
for the subproblems. If there are cycles in the co-design diagram,
the map~$\ftor$ involves the solution of a\emph{ least fixed point}
equation in the space of antichains. This equation can be solved using
Kleene's algorithm to find the antichain containing all minimal solutions
at the same time.

\subsubsection*{Approach}

The strategy to obtain these results consists in reducing an arbitrary
interconnection of design problems to considering\emph{ }only a finite
number of composition operators (series, parallel, and feedback).
\secref{threeoperators} defines these composition operators. \secref{Decomposition}~shows
how to turn a graph into a tree, where each junction is one of the
three operators. Given the tree representation of an MCDPs, we will
be able to give inductive arguments to prove the results.

\subsubsection*{Expressivity of MCDPs}

The results are significant because MCDPs induce a rich family of
optimization problems. 

We are not assuming, let alone strong properties like convexity, even
weaker properties like differentiability or continuity of the constraints.
In fact, we are not even assuming that functionality and resources
are continuous spaces; they could be arbitrary discrete posets. (In
that case, completeness and \scottcontinuity are trivially satisfied.)

Moreover, even assuming topological continuity of all spaces and maps
considered, MCDPs are strongly not convex. What makes them nonconvex
is the possibility of introducing feedback interconnections. To show
this, I will give an example of a 1-dimensional problem with a continuous~$\ftor$
for which the feasible set is disconnected.

\begin{figure}[H]
\hfill{}\subfloat[\label{fig:Simple-DP}]{\centering{}\includegraphics[scale=0.33]{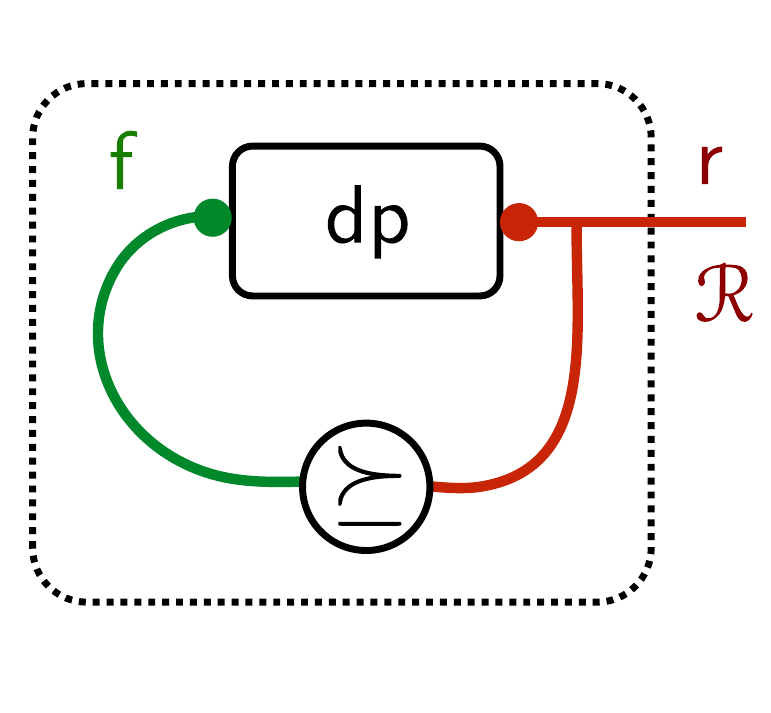}}\hfill{}\subfloat[\label{fig:nonconvex3}]{\centering{}\includegraphics[scale=0.33]{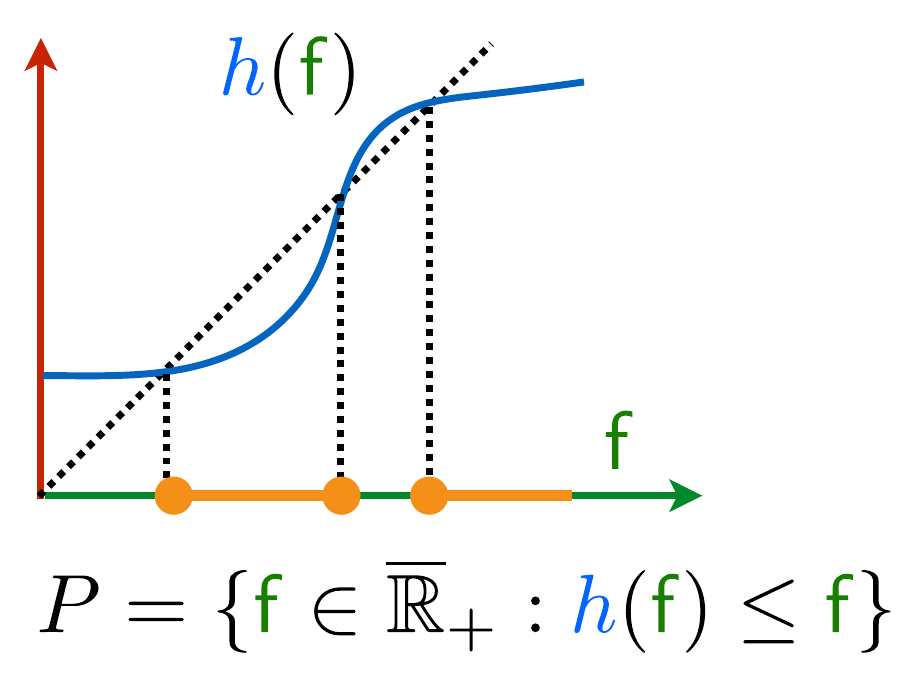}}\hfill{}

\caption{\label{fig:ceil-1}One feedback connection and a topologically continuous~$\ftor$
are sufficient to induce a disconnected feasible set.}
\end{figure}

\medskip{}

\begin{example}
\label{exa:one}Consider the CDPI in \figref{Simple-DP}. The \uline{m}inimal
resources~$M\subseteq\Aressp$ are the objectives of this optimization
problem:
\[
M\doteq\begin{cases}
\with & \fun,\res\in\funsp=\ressp,\\
\Min_{\posleq} & \res,\\
 & \res\in\ftor(\fun),\\
 & \res\posleq\fun.
\end{cases}
\]
The \uline{fea}sible set~$\Phi\subseteq\funsp\times\ressp$ is
the set of functionality and resources that satisfy the constraints~$\res\in\ftor(\fun)$
and~$\res\posleq\fun$:
\begin{equation}
\Phi=\left\{ \langle\fun,\res\rangle\in\funsp\times\ressp:(\res\in\ftor(\fun))\wedge(\res\posleq\fun)\right\} .\label{eq:feasible}
\end{equation}
The \uline{p}rojection~$P$ of~$\Phi$ to the functionality space
is:
\[
P=\left\{ \fun\mid\left\langle \fun,\res\right\rangle \in\Phi\right\} .
\]

In the scalar case ($\funsp=\ressp=\langle\nonNegRealsComp,\leq\rangle$),
the map~$\ftor\colon\funsp\rightarrow\Aressp$ is simply a map~$\ftor\colon\F{\nonNegRealsComp}\rightarrow\R{\nonNegRealsComp}$.
The set~$P$ of feasible functionality is described by
\begin{equation}
P=\{\fun\in\nonNegRealsComp:\ftor(\fun)\leq\fun\}.\label{eq:Pfeasible}
\end{equation}
\figref{nonconvex3} shows an example of a continuous map~$\ftor$
that gives a disconnected feasible set~$P$. Moreover, $P$ is disconnected
under any order-preserving nonlinear re-parametrization.

\end{example}

\section{Composition operators for design problems}

This section defines a handful of composition operators for design
problems. Later, \secref{Decomposing2} will prove that any co-design
problem can be described in terms of a subset of these operators. 

\label{sec:threeoperators}
\begin{defn}[$\dpseries$]
\label{def:series-composition}The series composition of two DPIs
$\dprob_{1}=\langle\funsp_{1},\ressp_{1},\impsp_{1},\exc_{1},\eval_{1}\rangle$
and $\dprob_{2}=\langle\funsp_{2},\ressp_{2},\impsp_{2},\exc_{2},$
$\eval_{2}\rangle$, for which~$\funsp_{2}=\ressp_{1}$, is 
\[
\dpseries(\dprob_{1},\dprob_{2})\doteq\left\langle \funsp_{1},\ressp_{2},\impsp,\exc,\eval\right\rangle ,
\]
where:
\begin{eqnarray*}
\impsp & = & \{\left\langle \imp_{1},\imp_{2}\right\rangle \in\impsp_{1}\times\impsp_{2}\mid\eval_{1}(\imp_{1})\posleq_{\ressp_{1}}\exc_{2}(\imp_{2})\},\\
\exc & : & \left\langle \imp_{1},\imp_{2}\right\rangle \mapsto\exc_{1}(\imp_{1}),\\
\eval & : & \left\langle \imp_{1},\imp_{2}\right\rangle \mapsto\eval_{2}(\imp_{2}).
\end{eqnarray*}
\end{defn}
\captionsideleft{\label{fig:composition-2}}{\includegraphics[scale=0.33]{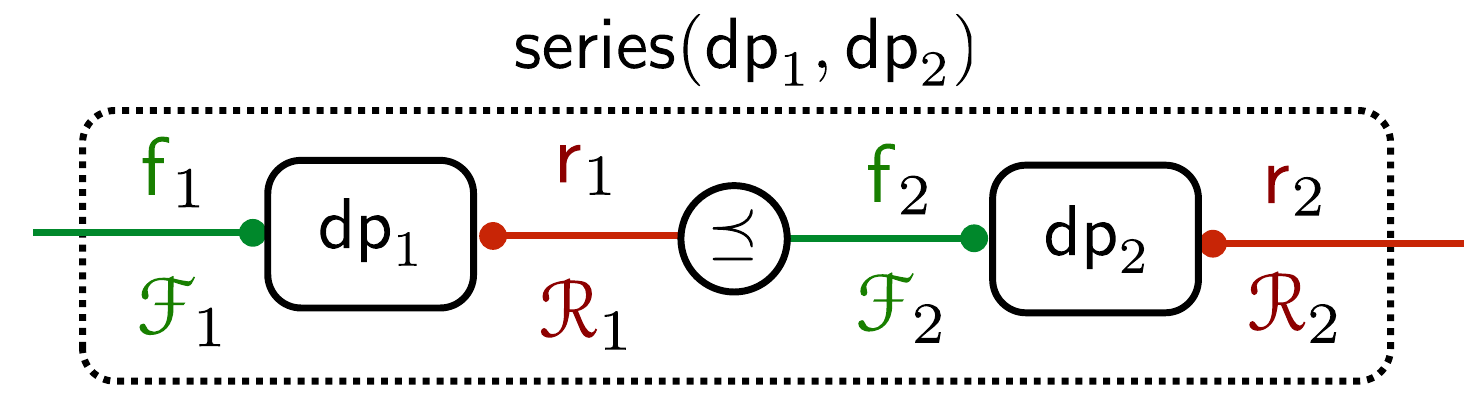}}
\begin{defn}[$\dppar$]
\label{def:parallel}The parallel composition of two DPIs $\dprob_{1}=\left\langle \funsp_{1},\ressp_{1},\impsp_{1},\exc_{1},\eval_{1}\right\rangle $
and $\dprob_{2}=\langle\funsp_{2},\ressp_{2},\impsp_{2},\exc_{2},$
$\eval_{2}\rangle$ is
\[
\dppar(\dprob_{1},\dprob_{2})\doteq\left\langle \funsp_{1}\times\funsp_{2},\ressp_{1}\times\ressp_{2},\impsp_{1}\times\impsp_{2},\exc,\eval\right\rangle ,
\]
where:
\begin{eqnarray}
\exc & : & \left\langle \imp_{1},\imp_{2}\right\rangle \mapsto\left\langle \exc_{1}(\imp_{1}),\exc_{2}(\imp_{2})\right\rangle ,\label{eq:dppar-exec}\\
\eval & : & \left\langle \imp_{1},\imp_{2}\right\rangle \mapsto\left\langle \eval_{1}(\imp_{1}),\eval_{2}(\imp_{2})\right\rangle .\nonumber 
\end{eqnarray}

\captionsideleft{\label{fig:dppar}}{\includegraphics[scale=0.33]{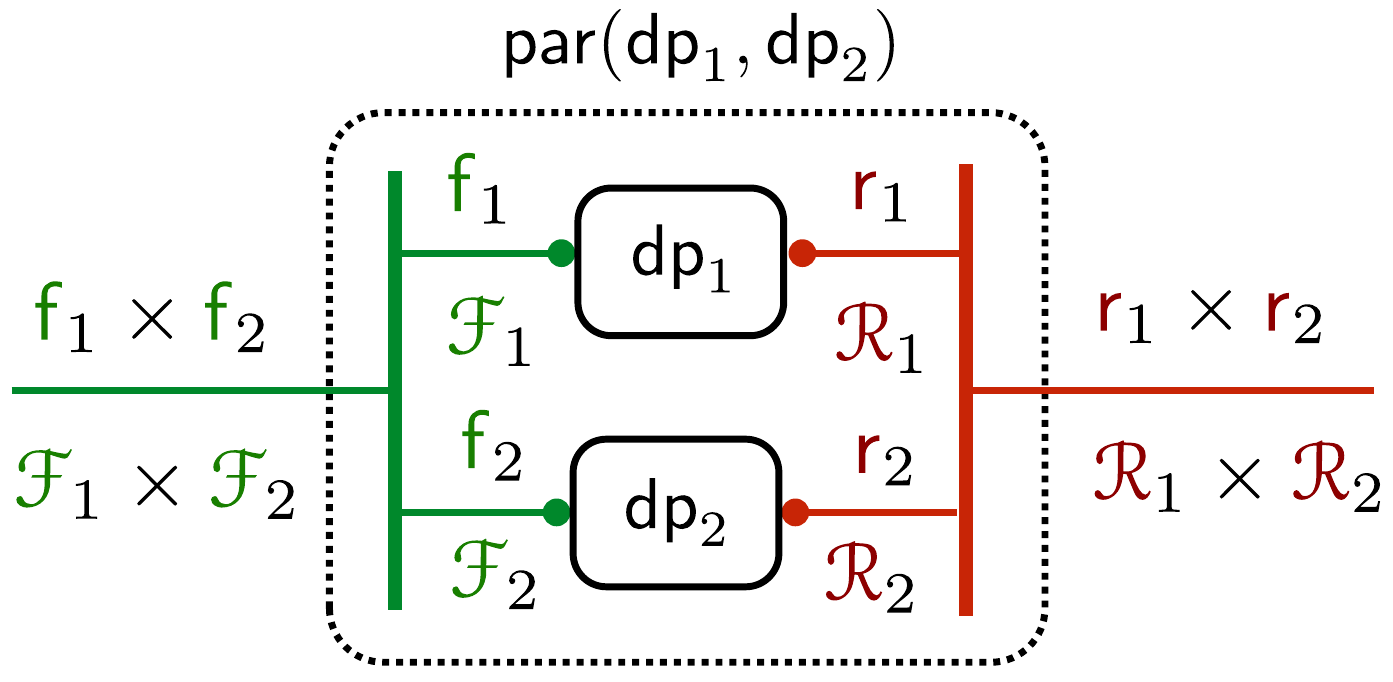}}
\end{defn}

\begin{defn}[$\dploop$]
\label{def:dp_loop}Suppose~$\dprob$ is a DPI with factored functionality
space~$\funsp_{1}\times\ressp$: 
\[
\dprob=\left\langle \funsp_{1}\times\ressp,\ressp,\impsp,\left\langle \exc_{1},\exc_{2}\right\rangle ,\eval\right\rangle .
\]
Then we can define the DPI~$\dploop(\dprob)$ as 
\[
\dploop(\dprob)\doteq\left\langle \funsp_{1},\ressp,\impsp',\exc_{1},\eval\right\rangle ,
\]
where~$\impsp'\subseteq\impsp$ limits the implementations to those
that respect the additional constraint~$\eval(\imp)\posleq\exc_{2}(\imp)$:
\[
\impsp'=\{\imp\in\impsp:\eval(\imp)\posleq\exc_{2}(\imp)\}.
\]
This is equivalent to ``closing a loop'' around~$\dprob$ with
the constraint~$\fun_{2}\posgeq\res$~(\figref{sloop}).
\end{defn}
\captionsideleft{\label{fig:sloop}\label{fig:sloop2}}{\includegraphics[scale=0.33]{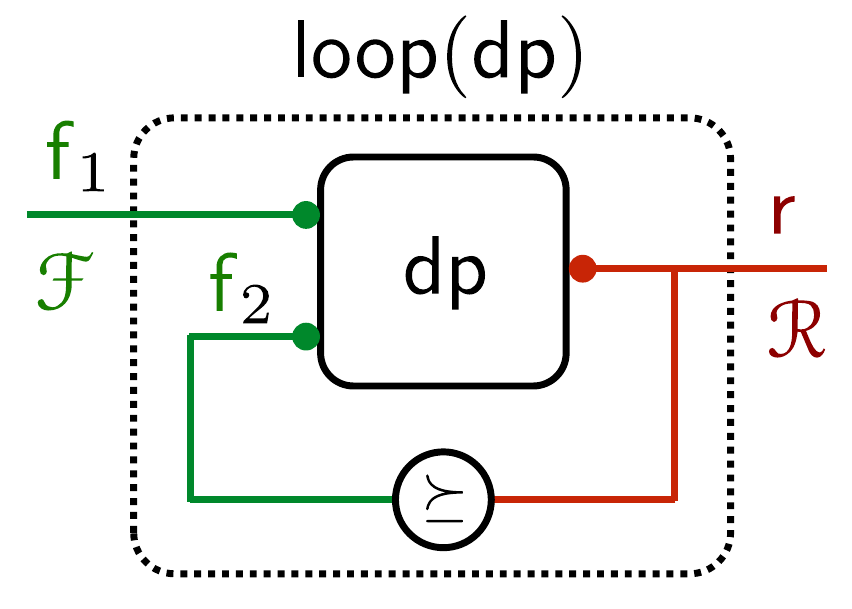}}

The operator~$\dploop$ is asymmetric because it acts on a design
problem with 2 functionalities and 1 resources. We can define a symmetric
feedback operator $\dploopb$ as in \figref{loop_general}, which
can be rewritten in terms of $\dploop$, using the construction in~\figref{loop_general2}\emph{.}

\begin{figure}[H]
\hspace*{\fill}\subfloat[\label{fig:loop_general}]{\includegraphics[scale=0.33]{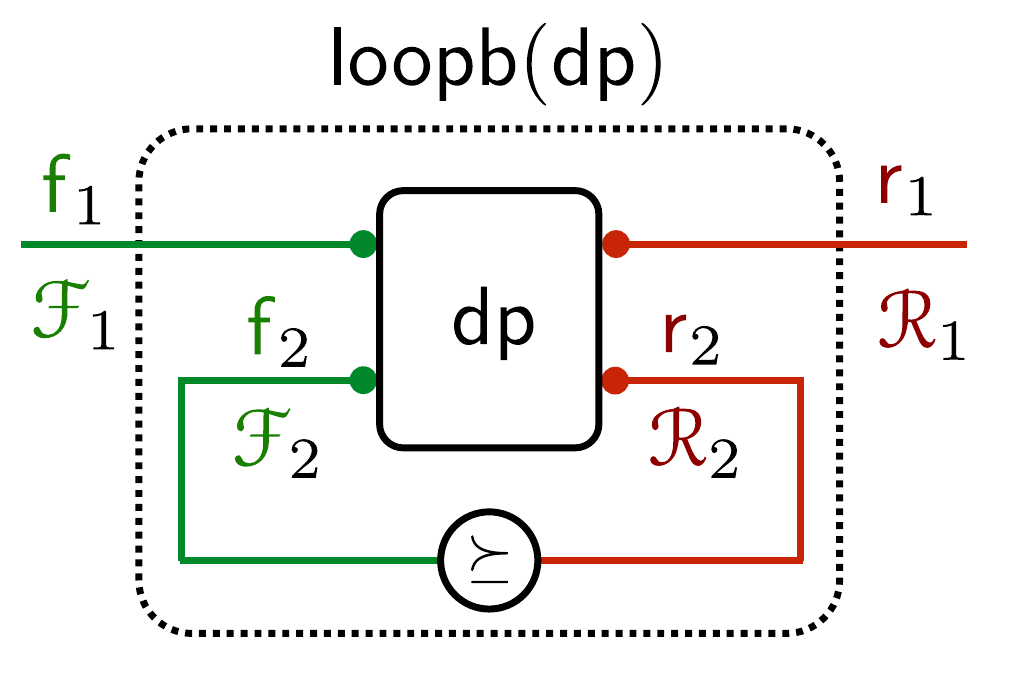}

}\hspace*{\fill}\subfloat[\label{fig:loop_general2}]{\includegraphics[scale=0.33]{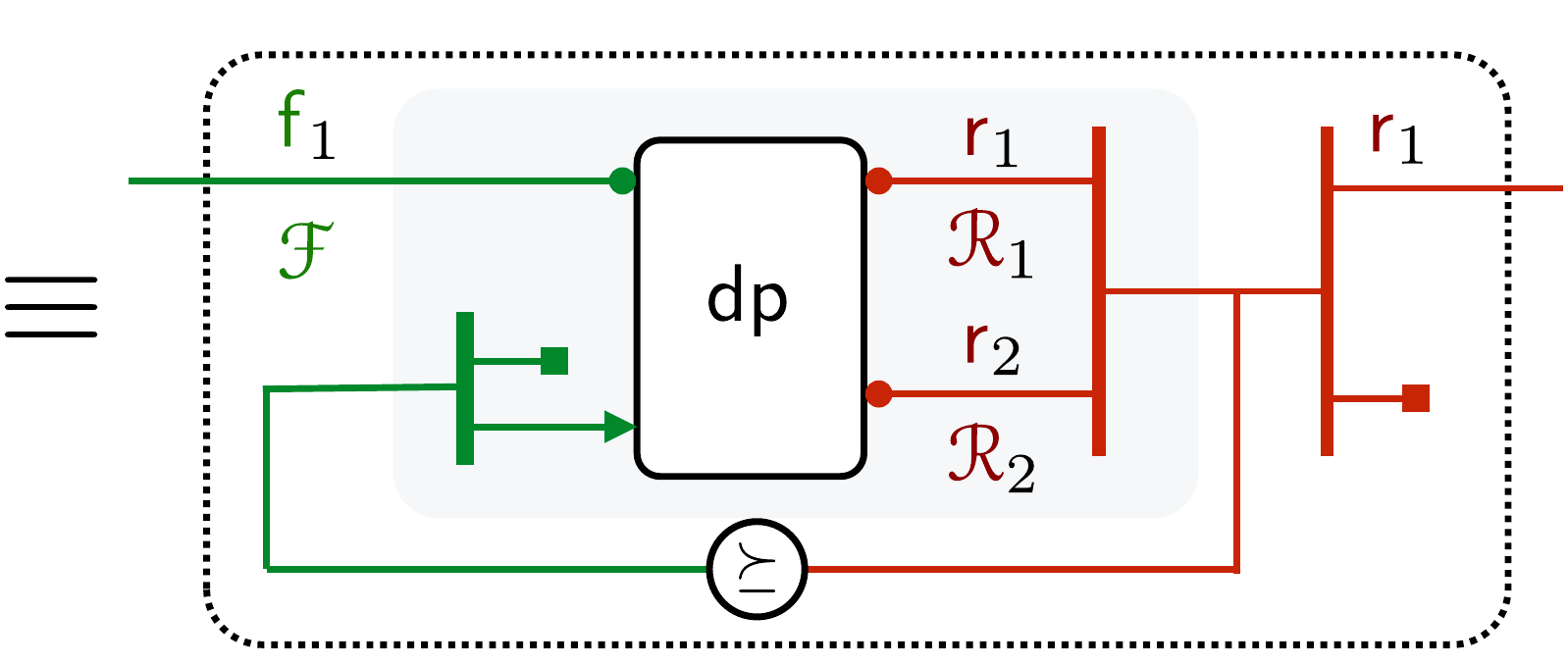}}\hspace*{\fill}\caption{A symmetric operator $\dploopb$ can be defined in terms of $\dploop$.}
\end{figure}

A ``co-product'' (see,~\eg\,\cite[Section 2.4]{spivak14category})
of two design problems is a design problem with the implementation
space~$\impsp=\impsp_{1}\sqcup\impsp_{2}$, and it represents the
exclusive choice between two possible alternative families of designs.
\begin{defn}[Co-product]
\label{def:parallel-1}Given two DPIs with same functionality and
resources $\dprob_{1}=\left\langle \funsp,\ressp,\impsp_{1},\exc_{1},\eval_{1}\right\rangle $
and $\dprob_{2}=\left\langle \funsp,\ressp,,\impsp_{2},\exc_{2},\eval_{2}\right\rangle ,$
define their co-product as
\[
\dprob_{1}\sqcup\dprob_{2}\doteq\left\langle \funsp,\ressp,\impsp_{1}\sqcup\impsp_{2},\exc,\eval\right\rangle ,
\]
where
\begin{eqnarray}
\exc & : & \imp\mapsto\begin{cases}
\exc_{1}(\imp), & \text{if }\imp\in\impsp_{1},\\
\exc_{2}(\imp), & \text{if }\imp\in\impsp_{2},
\end{cases}\label{eq:dppar-exec-1}\\
\eval & : & \imp\mapsto\begin{cases}
\eval_{1}(\imp), & \text{if }\imp\in\impsp_{1},\\
\eval_{2}(\imp), & \text{if }\imp\in\impsp_{2}.
\end{cases}\nonumber 
\end{eqnarray}

\captionsideleft{\label{fig:dpcoproduct}}{\includegraphics[scale=0.33]{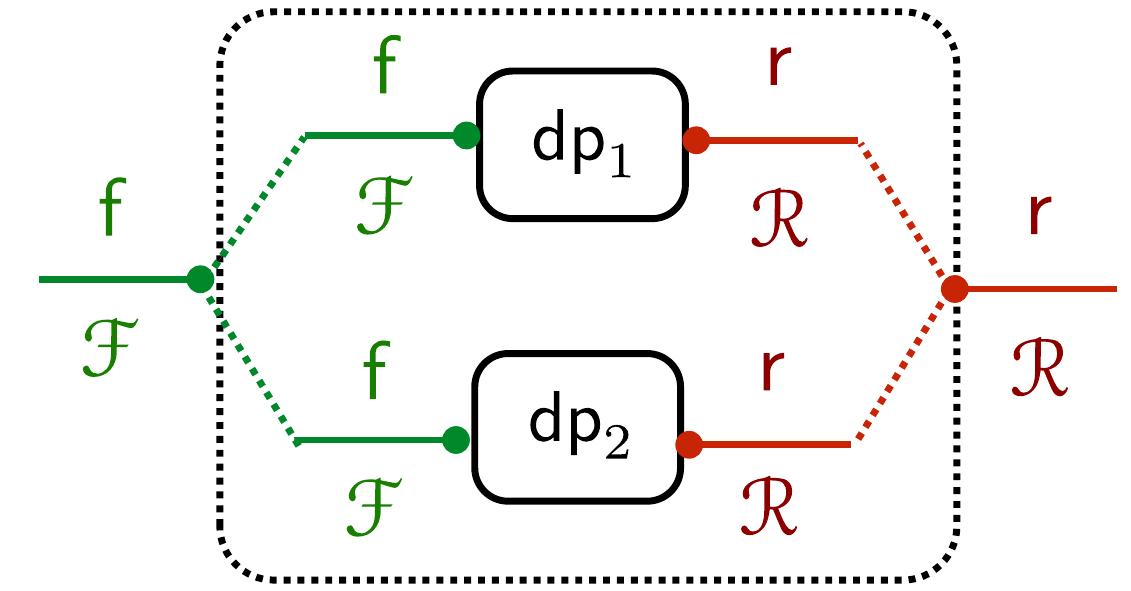}}
\end{defn}

\section{Decomposition of MCDPs\label{sec:Decomposition}}

\label{sec:Decomposing2}This section shows how to describe an arbitrary
interconnection of design problems using only three composition operators.
More precisely, for each CDPI with a set of atoms~$\cdpiN$, there
is an equivalent one that is built from $\dpseries$/$\dppar$/$\dploop$
applied to the set of atoms~$\cdpiN$ plus some extra ``plumbling''
(identities, multiplexers).

\subsubsection*{Equivalence}

The definition of equivalence below ensures that two equivalent DPIs
have the same map from functionality to resources, while one of the
DPIs can have a slightly larger implementation space. 
\begin{defn}
Two DPIs $\left\langle \funsp,\ressp,\impsp_{1},\exc_{1},\eval_{1}\right\rangle $
and $\langle\funsp,\ressp,$ $\impsp_{2},\exc_{2},\eval_{2}\rangle$
are \emph{equivalent} if there exists a map~$\varphi:\impsp_{2}\rightarrow\impsp_{1}$
such that~$\exc_{2}=\exc_{1}\circ\varphi$ and~$\eval_{2}=$ $\eval_{1}\circ\varphi$.
\end{defn}

\subsubsection*{Plumbing}

We also need to define ``trivial DPIs'', which serve as ``plumbing''.
These can be built by taking a map $f:\funsp\rightarrow\ressp$ and
lifting it to the definition of a DPI. The implementation space of
a trivial DPI is a copy of the functionality space and there is a
1-to-1 correspondence between functionality and implementation.
\begin{defn}[Trivial DPIs]
Given a map~$f:\funsp\rightarrow\ressp$, we can lift it to define
a trivial DPI $\triv(f)=\left\langle \funsp,\ressp,\funsp,\idFunc_{\funsp},f\right\rangle $,
where~$\idFunc_{\funsp}$ is the identity on~$\funsp$. 
\end{defn}
\begin{prop}
\label{prop:reduction}Given a CDPI $\left\langle \funsp,\ressp,\left\langle \cdpiN,\mathcal{E}\right\rangle \right\rangle $,
we can find an equivalent CDPI obtained by applying the operators
$\dppar/\dpseries/\dploop$ to a set of atoms~$\cdpiN'$ that contains~$\cdpiN$
plus a set of trivial DPIs. Furthermore, one instance of~$\dploop$
is sufficient.
\end{prop}
\begin{proof}
We show this constructively. We will temporarily remove all cycles
from the graph, to be reattached later. To do this, find an \emph{arc
feedback set} (AFS) $F\subseteq\mathcal{E}$. An AFS is a set of edges
that, when removed, remove all cycles from the graph (see, e.g.,~\cite{golovach15incremental}).
For example, the CDPI represented in~\figref{cdpi_comp1} has a minimal
AFS that contains the edge~$\text{c}\rightarrow\text{a}$~(\figref{cdpi_comp2}). 

\begin{figure}[H]
\subfloat[\label{fig:cdpi_comp1}]{\centering{}\includegraphics[scale=0.33]{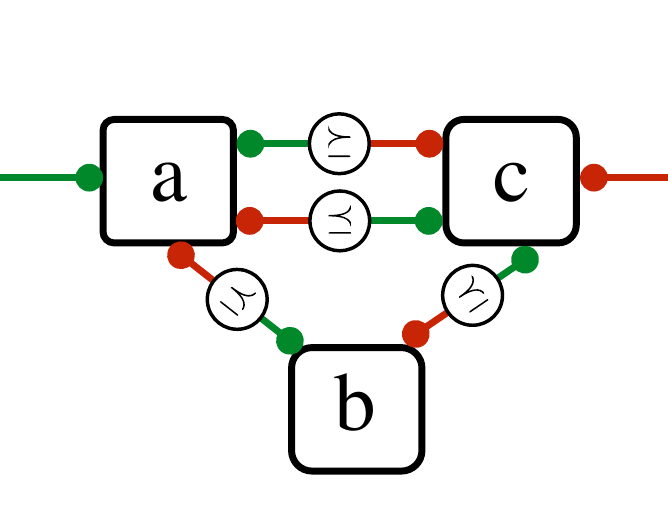}}\hfill{}\subfloat[\label{fig:cdpi_comp2}]{\centering{}\includegraphics[scale=0.33]{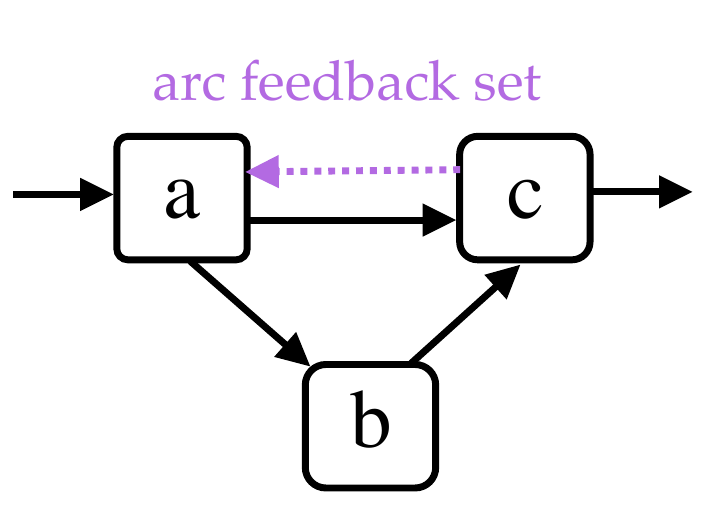}}\hfill{}\subfloat[\label{fig:cdpi_comp2b}]{\begin{centering}
\includegraphics[scale=0.33]{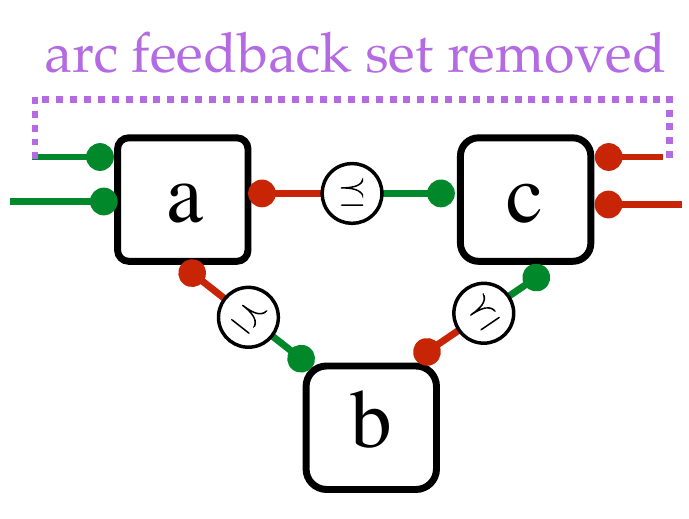}
\par\end{centering}

}

\caption{An example co-design diagram with three nodes~$\cdpiN=\{\text{a},\text{b},\text{c}\}$,
in which a minimal arc feedback set is $\{\text{c}\rightarrow\text{a}$\}.}
\end{figure}

Remove the~AFS~$F$ from~$\mathcal{E}$ to obtain the reduced edge
set~$\mathcal{E}'=\mathcal{E}\backslash F$.  The resulting graph~$\left\langle \cdpiN,\mathcal{E}'\right\rangle $
does not have cycles, and can be written as a series-parallel graph,
by applying the operators~$\dppar$ and~$\dpseries$ from a set
of nodes~$\cdpiN'$. The nodes~$\cdpiN'$ will contain~$\cdpiN$,
plus some extra ``connectors'' that are trivial DPIs. Find a weak
topological ordering of~$\cdpiN$. Then the graph~$\left\langle \cdpiN,\mathcal{E}'\right\rangle $
can be written as the series of~$|\cdpiN|$ subgraphs, each containing
one node of~$\cdpiN$. In the example, the weak topological ordering
is~$\left\langle \text{a},\text{b},\text{c}\right\rangle $ and there
are three subgraphs (\figref{cdpi_comp3}). 

\captionsideleft{\label{fig:cdpi_comp3}}{\includegraphics[scale=0.33]{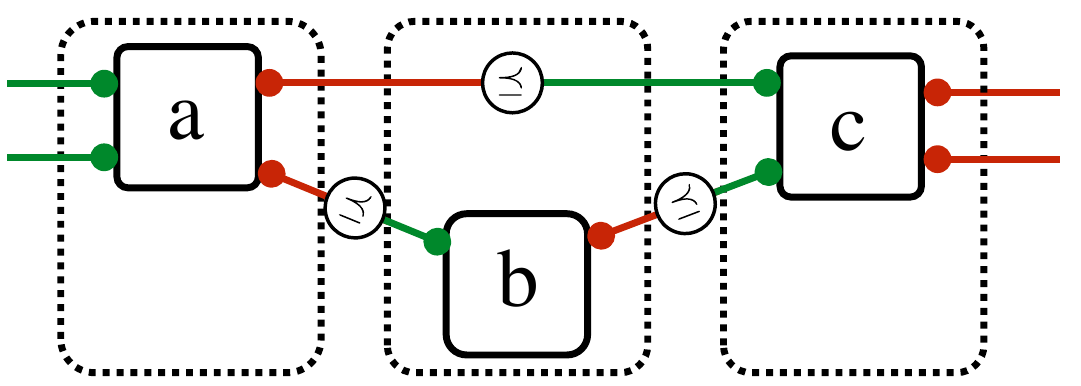}}

\noindent Each subgraph can be described as the parallel interconnection
of a node~$\cdpin\in\cdpiN$ and some extra connectors. For example,
the second subgraph in the graph can be written as the parallel interconnection
of node~$\text{b}$ and the identity $\triv(\idFunc)$ (\figref{cdpi_comp4}).

\captionsideleft{\label{fig:cdpi_comp4}}{\includegraphics[scale=0.33]{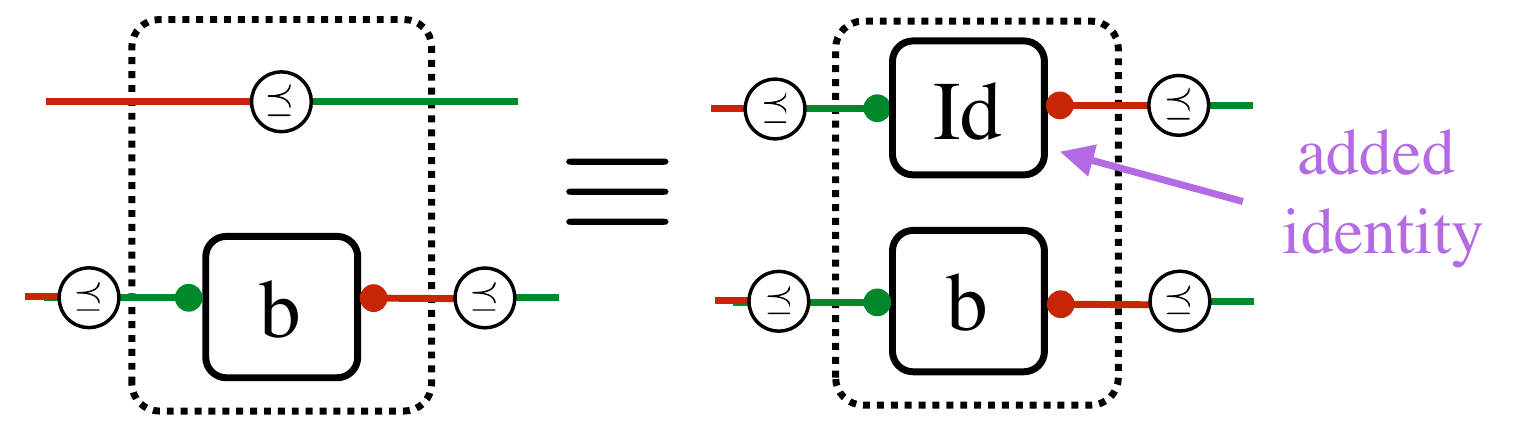}}

After this is done, we just need to ``close the loop'' around the
edges in the AFS~$F$ to obtain a CDPI that is equivalent to the
original one. Suppose the AFS~$F$ contains only one edge. Then one
instance of the~$\dploopb$ operator is sufficient~(\figref{cdpi_comp5}).
In this example, the tree representation (\figref{cdpi_comp6}) is
\[
\dploopb(\dpseries(\dpseries(\text{a},\dppar(\idFunc,\text{b})),\text{c}).
\]

\begin{figure}[H]
\hfill{}\subfloat[\label{fig:cdpi_comp5}]{\centering{}\includegraphics[scale=0.33]{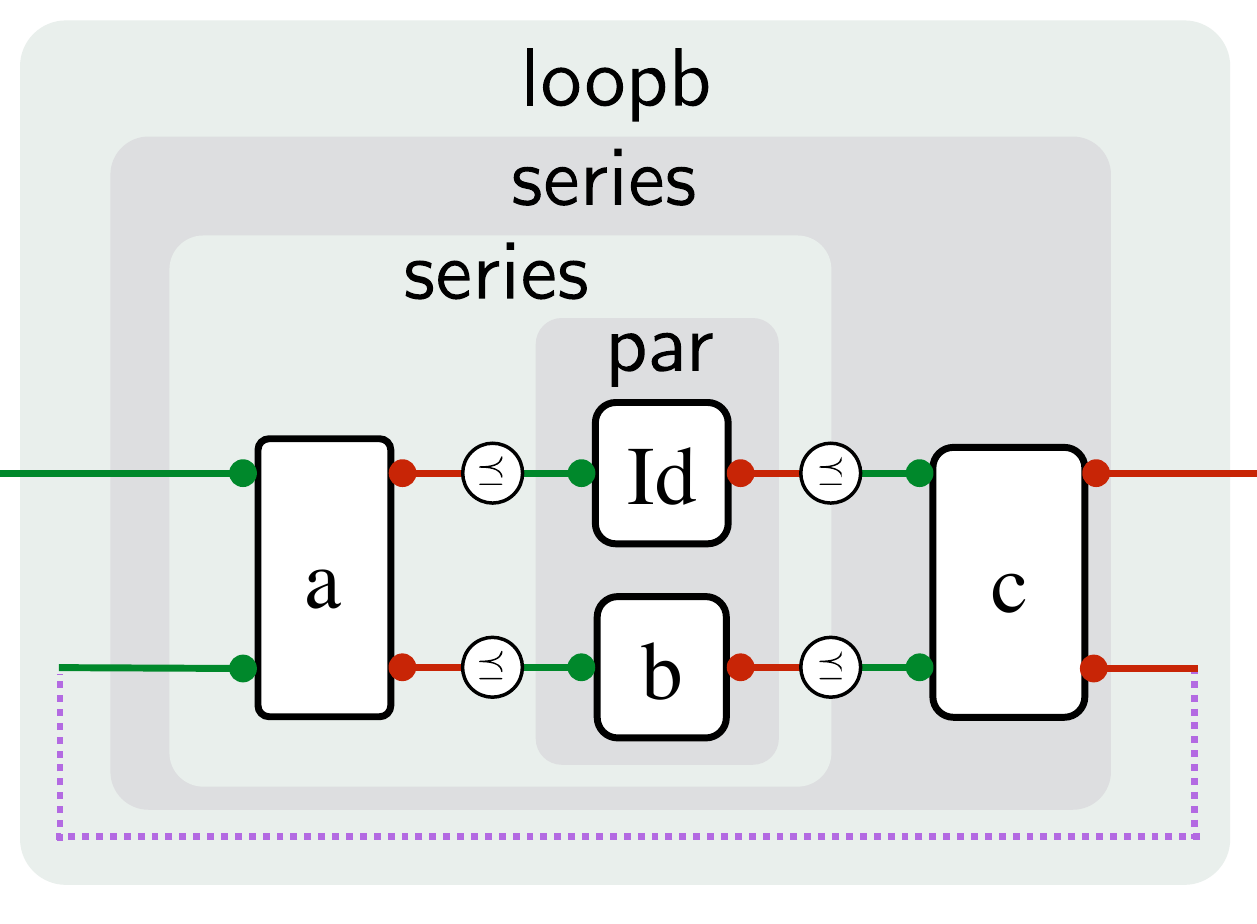}}\hfill{}\subfloat[\label{fig:cdpi_comp6}]{\centering{}\includegraphics[scale=0.33]{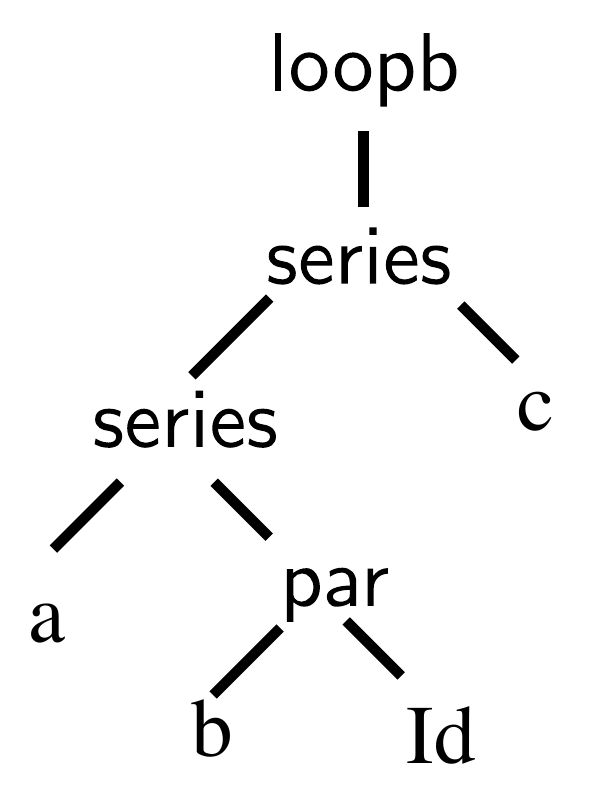}}\hfill{}

\caption{Tree representation for the co-design diagram in~\figref{cdpi_comp1}.}
\end{figure}

If the AFS contains multiple edges, then, instead of closing one loop
at a time, one can can always rewrite multiple nested loops as only
one loop by taking the product of the edges. For example, a diagram
like the one in~\figref{nested1} can be rewritten as~\figref{nested2}.
This construction is analogous to the construction used for the analysis
of process networks~\cite{lee10} (and any other construct involving
a traced monoidal category). Therefore, it is possible to describe
an arbitrary graph of design problems using only one instance of the
$\dploop$ operator.
\end{proof}

\begin{figure}[H]
\subfloat[\label{fig:nested1}]{\begin{centering}
\includegraphics[scale=0.33]{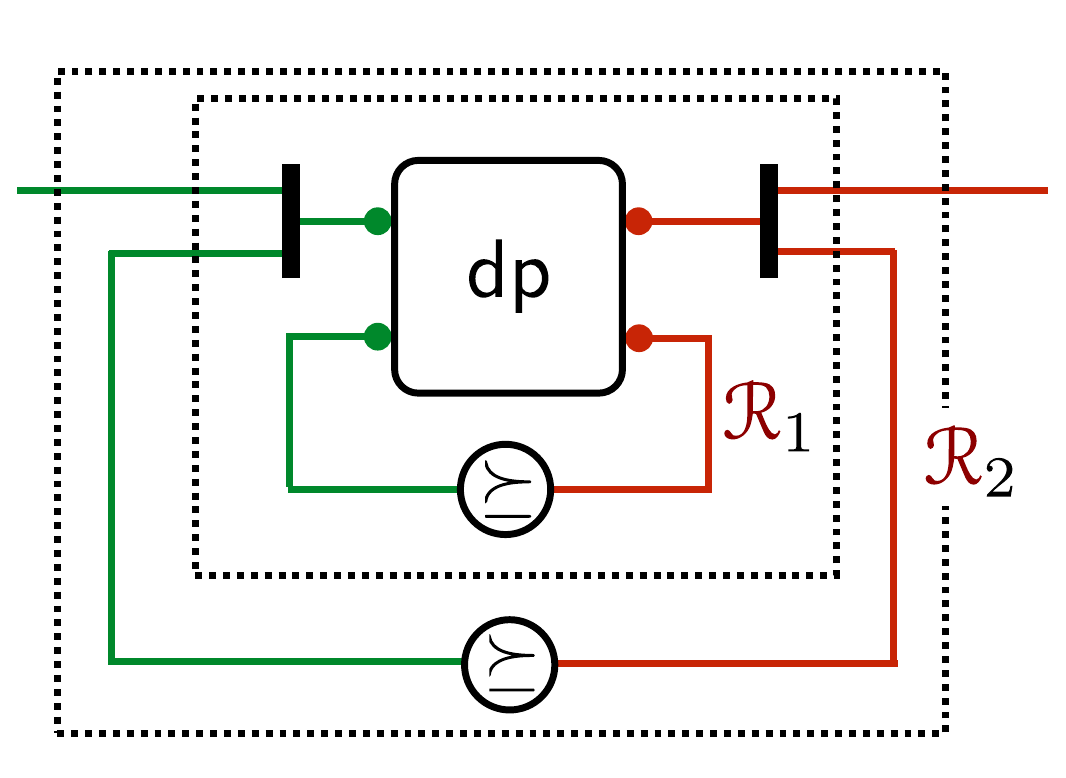}
\par\end{centering}
}\subfloat[\label{fig:nested2}]{\begin{centering}
\includegraphics[scale=0.33]{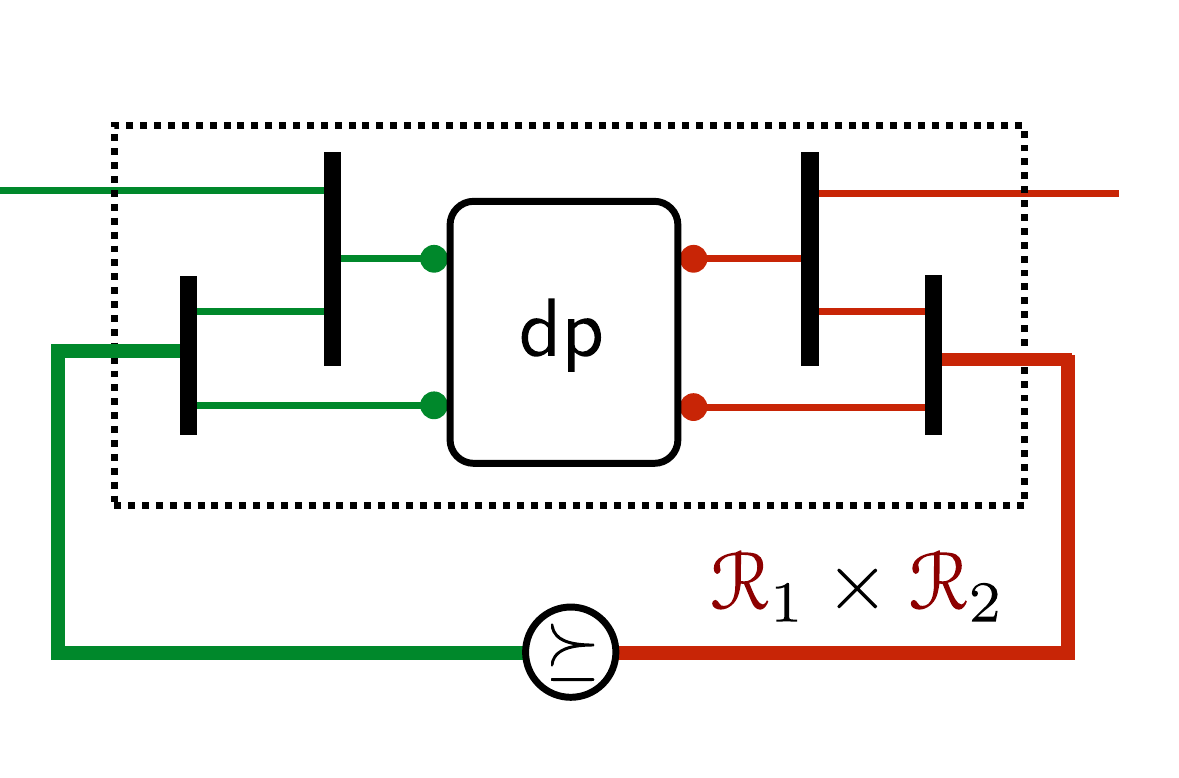}
\par\end{centering}
}

\caption{\label{fig:If-there-are}If there are nested loops in a co-design
diagram, they can be rewritten as one loop, by taking the product
of the edges.}
\end{figure}

\section{Monotonicity as compositional property\label{sec:Monotone-Co-Design-Problems}}

The first main result of this paper is an invariance result.

\noindent 
\fbox{\begin{minipage}[t]{0.95\columnwidth}
\begin{thm}
\label{thm:CDP-monotone}The class of MCDPs is closed with respect
to interconnection.
\end{thm}

\end{minipage}}

\begin{proof}
\propref{reduction} has shown that any interconnection of design
problems can be described using the three operators~$\dppar$, $\dpseries$,
and~$\dploop$. Therefore, we just need to check that monotonicity
in the sense of~\defref{DPI-monotone} is preserved by each operator
separately. This is done below in~\proref{dppar-monotone}\textendash \ref{pro:loop-continuous}. 
\end{proof}

\begin{prop}
\label{pro:dppar-monotone}If~$\dprob_{1}$ and~$\dprob_{2}$ are
monotone (\defref{DPI-monotone}), then also the composition~$\dppar(\dprob_{1},\dprob_{2})$
is monotone.
\end{prop}
\begin{proof}
We need to refer to the definition of $\dppar$ in \defref{parallel}
and check the conditions in \defref{DPI-monotone}. If~$\funsp_{1},\funsp_{2},\ressp_{1},\ressp_{2}$
are CPOs, then~$\funsp_{1}\times\funsp_{2}$ and~$\ressp_{1}\times\ressp_{2}$
are CPOs as well.

From \defref{ftor} and (\ref{eq:dppar-exec}) we know $\ftor$ can
be written as
\begin{align*}
\ftor:\funsp_{1}\times\funsp_{2} & \rightarrow\antichains(\ressp_{1}\times\ressp_{2})\\
\left\langle \fun_{1},\fun_{2}\right\rangle  & \mapsto\resMin\{\left\langle \eval_{1}(\imp_{1}),\eval_{2}(\imp_{2})\right\rangle \mid\\
 & \qquad\left(\left\langle \imp_{1},\imp_{2}\right\rangle \in\impsp_{1}\times\impsp_{2}\right)\\
 & \qquad\,\wedge\,\left(\left\langle \fun_{1},\fun_{2}\right\rangle \posleq\left\langle \exc_{1}(\imp_{1}),\exc_{2}(\imp_{2})\right\rangle \right)\}.
\end{align*}
All terms factorize in the two components, giving:{\small{}
\begin{align*}
\ftor\!:\!\left\langle \fun_{1},\fun_{2}\right\rangle \mapsto & \!\!\!\quad\Min_{\ressp_{1}}\{\left\langle \eval_{1}(\imp_{1})\right\rangle \mid\left(\imp\in\impsp_{1}\right)\,\wedge\,\left(\fun_{1}\posleq\exc_{1}(\imp_{1})\right)\}\\
 & \!\!\!\times\Min_{\ressp_{2}}\{\left\langle \eval_{2}(\imp_{2})\right\rangle \mid\left(\imp\in\impsp_{2}\right)\,\wedge\,\left(\fun_{2}\posleq\exc_{2}(\imp_{2})\right)\},
\end{align*}
}which reduces to 
\begin{eqnarray}
\ftor\colon\left\langle \fun_{1},\fun_{2}\right\rangle  & \mapsto & \ftor_{1}(\fun_{1})\acprod\ftor_{2}(\fun_{2}).\label{eq:isproduct}
\end{eqnarray}
The map $\ftor$ is \scottcontinuous iff $\ftor_{1}$ and $\ftor_{2}$
are~\cite[Lemma II.2.8]{gierz03continuous}.
\end{proof}

\begin{prop}
\label{pro:series-monotone}If $\dprob_{1}$ and $\dprob_{2}$ are
monotone (\defref{DPI-monotone}), then also the composition $\dpseries(\dprob_{1},\dprob_{2})$
is monotone.
\end{prop}
\begin{proof}
From the definition of $\dpseries$ (\defref{series-composition}),
the semantics of the interconnection is captured by this problem:
\begin{equation}
\ftor:\fun_{1}\mapsto\begin{cases}
\with & \res_{1},\fun_{2}\in\ressp_{1},\quad\res_{2}\in\ressp_{2},\\
\Min_{\posleq_{\ressp_{2}}} & \res_{2},\\
\subto & \res_{1}\in\ftor_{1}(\fun_{1}),\\
 & \res_{1}\posleq_{\ressp_{1}}\fun_{2},\\
 & \res_{2}\in\ftor_{2}(\fun_{2}).
\end{cases}\label{eq:dede}
\end{equation}
The situation is described by \figref{series_mono1-2}. The point~$\fun_{1}$
is fixed, and thus~$\ftor(\fun_{1})$ is a fixed antichain in~$\ressp_{1}$.
For each point $\res_{1}\in\ftor(\fun_{1})$, we can choose a $\fun_{2}\posgeq\res_{1}$.
For each~$\fun_{2}$, the antichain~$\ftor_{2}(\fun_{2})$ traces
the solution in $\ressp_{2}$, from which we can choose~$\res_{2}$. 

\captionsideleft{\label{fig:series_mono1-2}}{\includegraphics[scale=0.45]{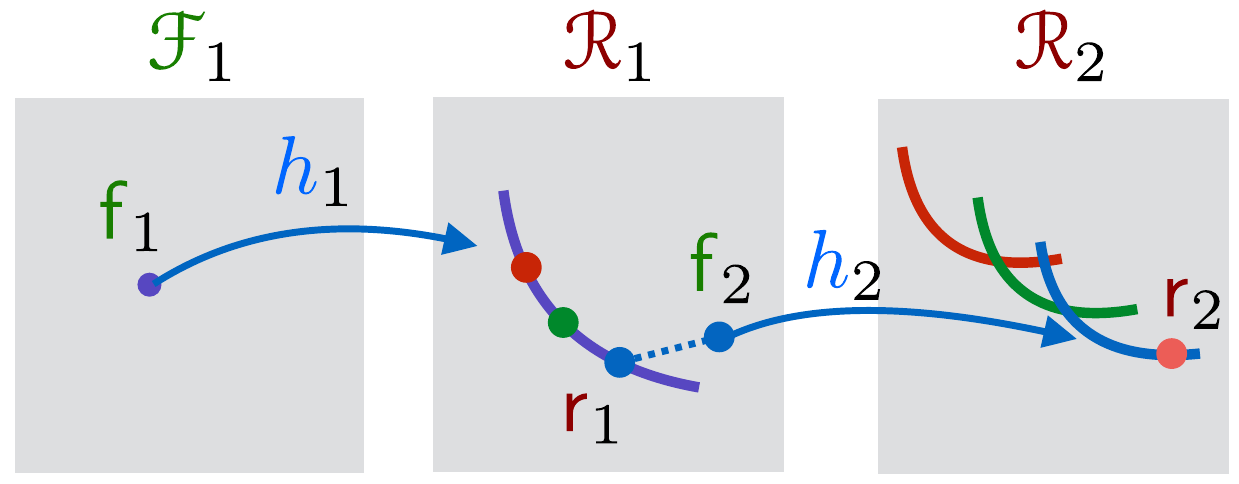}}

\noindent Because~$\ftor_{2}$ is monotone, $\ftor_{2}(\fun_{2})$
is minimized when~$\fun_{2}$ is minimized, hence we know that the
constraint~$\res_{1}\posleq\fun_{2}$ will be tight. We can then
conclude that the objective does not change introducing the constraint~$\res_{1}=\fun_{2}$.
The problem is reduced to:

\begin{equation}
\ftor:\fun_{1}\mapsto\begin{cases}
\with & \fun_{2}\in\ressp_{1},\quad\res_{2}\in\ressp_{2},\\
\Min_{\posleq_{\ressp_{2}}} & \res_{2},\\
\subto & \fun_{2}\in\ftor_{1}(\fun_{1}),\\
 & \res_{2}\in\ftor_{2}(\fun_{2}).
\end{cases}\label{eq:dede-2}
\end{equation}
Minimizing $\res_{2}$ with the only constraint being~$\res_{2}\in\ftor_{2}(\fun_{2})$,
and with~$\ftor_{2}(\fun_{2})$ being an antichain, the solutions
are all and only~$\ftor_{2}(\fun_{2})$. Hence the problem is reduced
to
\begin{equation}
\ftor:\fun_{1}\mapsto\begin{cases}
\with & \fun_{2}\in\ressp_{1},\\
\Min_{\posleq_{\ressp_{2}}} & \ftor_{2}(\fun_{2}),\\
\subto & \fun_{2}\in\ftor_{1}(\fun_{1}).
\end{cases}\label{eq:dede-2-1}
\end{equation}
The solution is simply
\begin{equation}
\ftor:\fun_{1}\mapsto\Min_{\posleq_{\ressp_{2}}}\bigcup_{\fun_{2}\in\ftor_{1}(\fun_{1})}\ftor_{2}(\fun_{2}).\label{eq:ora}
\end{equation}
This map is \scottcontinuous because it is the composition of \scottcontinuous
maps.
\end{proof}

\begin{prop}
\label{pro:loop-continuous}If $\dprob$ is monotone (\defref{DPI-monotone}),
so is~$\dploop(\dprob)$.
\end{prop}
\begin{proof}
The diagram in \figref{sloop} implies that the map~$\ftor_{\dploop(\dprob)}$
can be described as:
\begin{align}
\ftor_{\dploop(\dprob)}\colon\funsp_{1} & \rightarrow\Aressp,\label{eq:loopproblem}\\
\fun_{1} & \mapsto\begin{cases}
\with & \res,\fun_{2}\in\ressp,\\
\Min_{\resleq} & \res,\\
\subto & \res\in\ftor_{\dprob}(\fun_{1},\fun_{2}),\\
 & \res\resleq\fun_{2}.
\end{cases}
\end{align}
Denote by~$\ftor_{\fun_{1}}$ the map~$\ftor_{\dprob}$ with the
first element fixed:
\[
\ftor_{\fun_{1}}\colon\fun_{2}\mapsto\ftor_{\dprob}(\fun_{1},\fun_{2}).
\]
Rewrite $\res\in\ftor_{\dprob}(\fun_{1},\fun_{2})$ in \eqref{loopproblem}
as
\begin{equation}
\res\in\ftor_{\fun_{1}}(\fun_{2}).\label{eq:h2}
\end{equation}
Let~$\res$ be a feasible solution, but not necessarily minimal.
Because of \lemref{antichain-write}, the constraint \eqref{h2} can
be rewritten as 
\begin{equation}
\{\res\}=\ftor_{\fun_{1}}(\fun_{2})\cap\upit\res.\label{eq:h3}
\end{equation}
Because $\fun_{2}\posgeq\res$, and $\ftor_{\fun_{1}}$ is \scottcontinuous,
it follows that~$\ftor_{\fun_{1}}(\fun_{2})\posgeq_{\Aressp}\ftor_{\fun_{1}}(\res)$.
Therefore, by \lemref{antichain_inter}, we have
\begin{equation}
\{\res\}\posgeq_{\Aressp}\ftor_{\fun_{1}}(\res)\cap\upit\res.\label{eq:fea}
\end{equation}
This is a recursive condition that all feasible~$\res$ must satisfy.

Let $\R{R}\in\Aressp$ be an antichain of feasible resources, and
let~$\res$ be a generic element of~$\R{R}$. Tautologically, rewrite~$\R{R}$
as the minimal elements of the union of the singletons containing
its elements: 
\begin{equation}
\R{R}=\Min_{\resleq}\bigcup_{\res\in\R{R}}\ \{\res\}.\label{eq:condition3}
\end{equation}
Substituting (\ref{eq:fea}) in (\ref{eq:condition3}) we obtain (cf
\lemref{antichain_union})
\begin{equation}
\R{R}\posgeq_{\Aressp}\Min_{\resleq}\bigcup_{\res\in\R{R}}\ftor_{\fun_{1}}(\res)\ \cap\ \upit\res.\label{eq:recursive}
\end{equation}

{[}Converse: It is also true that if an antichain~$\R{R}$ satisfies~\eqref{recursive}
then all~$\res\in\R{R}$ are feasible. The constraint \eqref{recursive}
means that for any~$\res_{0}\in\R{R}$ on the left side, we can find
a~$\res_{1}$ in the right side so that~$\res_{0}\posgeq_{\ressp}\res_{1}$.
The point~$\res_{1}$ needs to belong to one of the sets of which
we take the union; say that it comes from $\res_{2}\in\R{R}$, so
that $\res_{1}\in\ftor_{\fun_{1}}(\res_{2})\ \cap\ \upit\res_{2}$.
Summarizing:

{\footnotesize{}
\begin{equation}
\forall\res_{0}\in\R{R}:\ \exists\res_{1}\colon\ (\res_{0}\posgeq_{\ressp}\res_{1})\ \wedge\ (\exists\res_{2}\in\R{R}\colon\ \res_{1}\in\ftor_{\fun_{1}}(\res_{2})\ \cap\ \upit\res_{2}).\label{eq:conc}
\end{equation}
}Because~$\res_{1}\in\ftor_{\fun_{1}}(\res_{2})\,\cap\,\upit\res_{2}$,
we can conclude that~$\res_{1}\in\upit\res_{2}$, and therefore~$\res_{1}\posgeq_{\ressp}\res_{2}$,
which together with~$\res_{0}\posgeq_{\ressp}\res_{1}$, implies~$\res_{0}\posgeq_{\ressp}\res_{2}$.
We have concluded that there exist two points~$\res_{0},\res_{2}$
in the antichain~$\R{R}$ such that~$\res_{0}\posgeq_{\ressp}\res_{2}$;
therefore, they are the same point:~$\res_{0}=\res_{2}$. Because~$\res_{0}\posgeq_{\ressp}\res_{1}\posgeq_{\ressp}\res_{2}$,
we also conclude that~$\res_{1}$ is the same point as well. We can
rewrite~\eqref{conc} by using~$\res_{0}$ in place of~$\res_{1}$
and~$\res_{2}$ to obtain~$\forall\res_{0}\in\R{R}:\res_{0}\in\ftor_{\fun_{1}}(\res_{0})$,
which means that~$\res_{0}$ is a feasible resource.{]}

We have concluded that all antichains of feasible resources~$\R{R}$
satisfy~\eqref{recursive}, and conversely, if an antichain~$\R{R}$
satisfies~\eqref{recursive}, then it is an antichain of feasible
resources. 

Equation \eqref{recursive} is a recursive constraint for~$\R{R}$,
of the kind 
\[
\Phi_{\fun_{1}}(\R{R})\posleq_{\Aressp}\R{R},
\]
with the map~$\Phi_{\fun_{1}}$ defined by
\begin{eqnarray}
\Phi_{\fun_{1}}:\Aressp & \rightarrow & \Aressp,\label{eq:bigphi}\\
\R{R} & \mapsto & \Min_{\resleq}\bigcup_{\res\in\R{R}}\ftor_{\fun_{1}}(\res)\ \cap\ \upit\res.\nonumber 
\end{eqnarray}
If we want the \emph{minimal} resources, we are looking for the \emph{least}
antichain: 
\[
\min_{\posleq_{\Aressp}}\{\,\R{R}\in\Aressp\colon\ \Phi_{\fun_{1}}(\R{R})\posleq_{\Aressp}\R{R}\,\},
\]
which is equal to the \emph{least fixed point }of~$\Phi_{\fun_{1}}$.
Therefore, the map $\ftor_{\dploop(\dprob)}$ can be written as
\begin{equation}
\ftor_{\dploop(\dprob)}\colon\fun_{1}\mapsto\lfp(\Phi_{\fun_{1}}).\label{eq:loop_fixpoint}
\end{equation}
\lemref{dagger} shows that $\lfp(\Phi_{\fun_{1}})$ is \scottcontinuous
in~$\fun_{1}$. 
\end{proof}

\begin{lem}
\label{lem:antichain-write}Let~$A$ be an antichain in $\posA$.
Then
\[
a\in A\qquad\equiv\qquad\{a\}=A\,\cap\upit a.
\]
\end{lem}

\begin{lem}
\label{lem:antichain_inter}For $A,B\in\antichains\posA$, and $S\subseteq P$,
$A\posleq_{\Aressp}B$ implies $A\cap S\posleq_{\Aressp}B\cap S$. 
\end{lem}

\begin{lem}
\label{lem:antichain_union}For $A,B,C,D\in\antichains\posA$, $A\posleq_{\Aressp}C$
and $B\posleq_{\Aressp}D$ implies $A\cup B\posleq_{\Aressp}C\cup D.$
\end{lem}

\begin{lem}
\label{lem:dagger}Let~$f\colon\posA\times\posB\rightarrow\posB$
be \scottcontinuous. For each~$x\in\posA$, define $f_{x}:y\mapsto f(x,y).$
Then $f^{\dagger}:x\mapsto\lfp(f_{x})$ is \scottcontinuous.
\end{lem}
\begin{proof}
Davey and Priestly~\cite{davey02} leave this as Exercise~8.26.
A proof is found in Gierz~\etal~\cite[Exercise II-2.29]{gierz03continuous}.
\end{proof}

\section{Solution of MCDPs\label{sec:Solution-of-Monotone}}

The second main result is that the map $\ftor$ of a MCDP has an explicit
expression in terms of the maps~$\ftor$ of the subproblems.  

\noindent 
\fbox{\begin{minipage}[t]{0.95\columnwidth}
\begin{thm}
\label{thm:CDP-solvig}The map~$\ftor$ for an MCDP has an explicit
expression in terms of the maps $\ftor$ of its subproblems, defined
recursively using the rules in \tabref{Correspondence}.

\begin{table}[H]
\begin{centering}
\caption{Recursive expressions for $\ftor$\label{tab:Correspondence}}
\par\end{centering}
\centering{}\setlength\extrarowheight{5pt}\normalsize
\begin{tabular}{ccc}
series & $\dprob=\dpseries(\dprob_{1},\dprob_{2})$ & $\ftor=\ftor_{1}\opseries\ftor_{2}$\tabularnewline
parallel & $\dprob=\dppar(\dprob_{1},\dprob_{2})$ & $\ftor=\ftor_{1}\oppar\ftor_{2}$\tabularnewline
feedback & $\dprob=\dploop(\dprob_{1})$ & $\ftor=\ftor_{1}^{\oploop}$\tabularnewline
co-product & $\dprob=\dprob_{1}\sqcup\dprob_{2}$ & $\ftor=\ftor_{1}\opcoprod\ftor_{2}$\tabularnewline
\end{tabular}
\end{table}
\end{thm}

\end{minipage}}

\begin{proof}
These expressions were derived in the proofs of~\proref{dppar-monotone}\textendash \ref{pro:loop-continuous}.
 The operators $\opseries,\oppar,\oploop,\opcoprod$ are defined
in \defref{opseries}\textendash \ref{def:opcoprod}. 
\end{proof}
\begin{defn}[Series operator~$\opseries$]
\label{def:opseries}For two maps~$\ftor_{1}\colon\funsp_{1}\rightarrow\Aressp_{1}$
and~$\ftor_{2}\colon\funsp_{2}\rightarrow\Aressp_{2}$, if~$\ressp_{1}=\funsp_{2}$
, define
\begin{align*}
{\displaystyle \ftor_{1}\opseries\ftor_{2}\colon\funsp_{1}} & \rightarrow\Aressp_{2},\\
\fun_{1} & \mapsto\Min_{\posleq_{\ressp_{2}}}\bigcup_{s\in\ftor_{1}(\fun)}\ftor_{2}(s).
\end{align*}
\end{defn}

\begin{defn}[Parallel operator $\oppar$]
\label{def:opmaps}For two maps $\ftor_{1}\colon\funsp_{1}\rightarrow\Aressp_{1}$
and $\ftor_{2}\colon\funsp_{2}\rightarrow\Aressp_{2}$, define
\begin{align}
\ftor_{1}\oppar\ftor_{2}:(\funsp_{1}\times\funsp_{2}) & \rightarrow\antichains(\ressp_{1}\times\ressp_{2}),\label{eq:oppar}\\
\left\langle \fun_{1},\fun_{2}\right\rangle  & \mapsto\ftor_{1}(\fun_{1})\acprod\ftor_{2}(\fun_{2}),\nonumber 
\end{align}
where $\acprod$ is the product of two antichains.
\end{defn}

\begin{defn}[Feedback operator $\oploop$]
\label{def:oploop}For $\ftor:\funsp_{1}\times\ressp\rightarrow\Aressp$,
define
\begin{align}
\ftor^{\oploop}:\funsp_{1} & \rightarrow\Aressp,\nonumber \\
\fun_{1} & \mapsto\lfp\left(\Psi_{\fun_{1}}^{\ftor}\right),\label{eq:lfp}
\end{align}
where~$\Psi_{\fun_{1}}^{\ftor}$ is defined as
\begin{align}
\Psi_{\fun_{1}}^{\ftor}:\Aressp & \rightarrow\Aressp,\nonumber \\
{\colR R} & \mapsto\Min_{\posleq_{\ressp}}\bigcup_{\res\in{\colR R}}\ftor(\fun_{1},\res)\ \cap\upit\res.\label{eq:phi}
\end{align}
\end{defn}

\begin{defn}[Coproduct operator $\opcoprod$]
\label{def:opcoprod}For $\ftor_{1},\ftor_{2}:\funsp\rightarrow\Aressp$,
define
\begin{align*}
\ftor_{1}\opcoprod\ftor_{2}:\funsp & \rightarrow\Aressp,\\
\fun & \mapsto\Min_{\posleq_{\ressp}}\left(\ftor_{1}(\fun)\cup\ftor_{2}(\fun)\right).
\end{align*}
\end{defn}

\subsection{Example: Optimizing over the natural numbers}

This is the simplest example that can show two interesting properties
of MCDPs: 
\begin{enumerate}
\item the ability to work with discrete posets; and 
\item the ability to treat multi-objective optimization problems.
\end{enumerate}
Consider the family of optimization problems indexed by~${\colF c}\in\mathbb{N}$:
\begin{equation}
\begin{cases}
\Min_{\posleq_{\overline{\mathbb{N}}\times\overline{\mathbb{N}}}} & \left\langle {\colR x},{\colR y}\right\rangle ,\\
\subto & x+y\geq\lceil\sqrt{x}\,\rceil+\lceil\sqrt{y}\,\rceil+{\colF c}.
\end{cases}\label{eq:example-1}
\end{equation}
One can show that this optimization problem is an MCDP by producing
a co-design diagram with an equivalent semantics, such as the one
in~\figref{toydiagram}. 

\begin{figure}[H]
\subfloat[\label{fig:toydiagram}]{\centering{}\includegraphics[width=6.2cm]{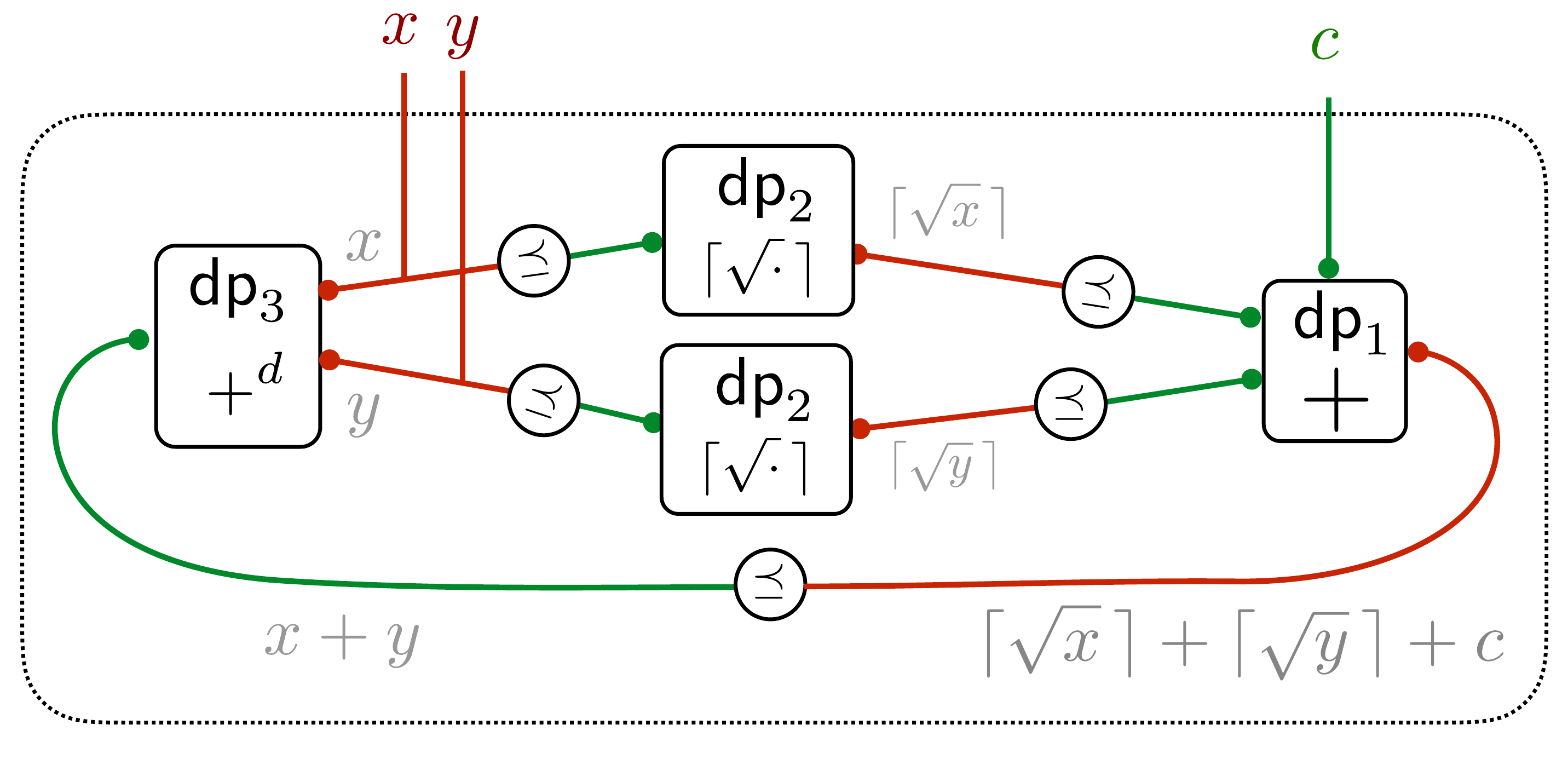}}\subfloat[\label{fig:toydiagram_tree}]{\begin{centering}
\includegraphics[scale=0.33]{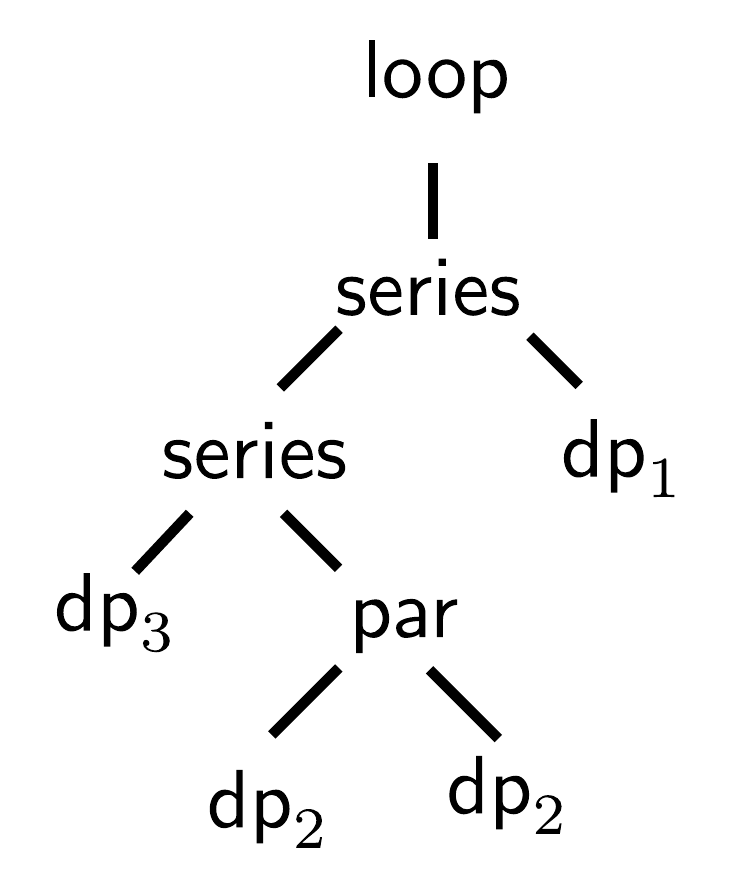}
\par\end{centering}
}

\caption{Co-design diagram equivalent to \eqref{example-1} and its tree representation.}
\end{figure}

\noindent The diagram contains three primitive DPIs: $\dprob_{1},$
$\dprob_{2}$ (used twice), and $\dprob_{3}$. Their $\ftor$ maps
are:
\begin{eqnarray*}
\ftor_{1}:{\colF\overline{\mathbb{N}}\times\overline{\mathbb{N}}\times\overline{\mathbb{N}}} & \rightarrow & {\colR\antichains\overline{\mathbb{N}}},\\
\left\langle \fun_{1},\fun_{2},\fun_{3}\right\rangle  & \mapsto & \{\fun_{1}+\fun_{2}+\fun_{3}\},\\
\ftor_{2}:{\colF\overline{\mathbb{N}}} & \rightarrow & {\colR\antichains\overline{\mathbb{N}}},\\
\fun & \mapsto & \{\lceil\sqrt{\fun}\,\rceil\},\\
\ftor_{3}:{\colF\overline{\mathbb{N}}} & \rightarrow & {\colR\antichains(\overline{\mathbb{N}}\times\overline{\mathbb{N}})},\\
\fun & \mapsto & \{\left\langle a,b\right\rangle \in\overline{\mathbb{N}}\times\overline{\mathbb{N}}:a+b=\fun\}.
\end{eqnarray*}
The tree decomposition (\figref{toydiagram_tree}) corresponds to
the expression
\begin{equation}
\dprob=\dploop(\dpseries(\dppar(\dprob_{2},\dprob_{2}),\dpseries(\dprob_{1},\dprob_{3}))).\label{eq:expression}
\end{equation}
Consulting \tabref{Correspondence}, from~(\ref{eq:expression})
one obtains an expression for~$\ftor$:
\begin{equation}
\ftor=\left((\ftor_{2}\oppar\ftor_{2})\opseries\ftor_{1}\opseries\ftor_{3}\right)^{\oploop}.\label{eq:h}
\end{equation}
This problem is small enough that we can write down an explicit expression
for~$\ftor$. By substituting in~(\ref{eq:h}) the definitions given
in~\defref{opseries}\textendash \ref{def:oploop}, we obtain that
evaluating~$\ftor(\F{c})$ means finding the least fixed point of
a map~$\Psi_{\F{c}}$: 
\[
\ftor:\F{c}\mapsto\lfp(\Psi_{\F{c}}).
\]
The map $\Psi_{\F{c}}:\R{\antichains(\overline{\mathbb{N}}\times\overline{\mathbb{N}})}\rightarrow\R{\antichains(\overline{\mathbb{N}}\times\overline{\mathbb{N}})}$
can be obtained from~(\ref{eq:phi}) as follows:
\[
\Psi_{\F{c}}:\R{R}\mapsto\Min\bigcup_{\left\langle x,y\right\rangle \in\R{R}}\upit\left\langle x,y\right\rangle \cap\qquad\qquad\qquad\qquad
\]
\[
\cap\left\{ \left\langle a,b\right\rangle \in\mathbb{N}^{2}:\left(a+b\geq\lceil\sqrt{x}\,\rceil+\lceil\sqrt{y}\,\rceil+\F{c}\right)\right\} .
\]

\begin{figure}
\begin{centering}
\includegraphics[bb=0bp 0bp 458bp 621bp,clip,scale=0.5]{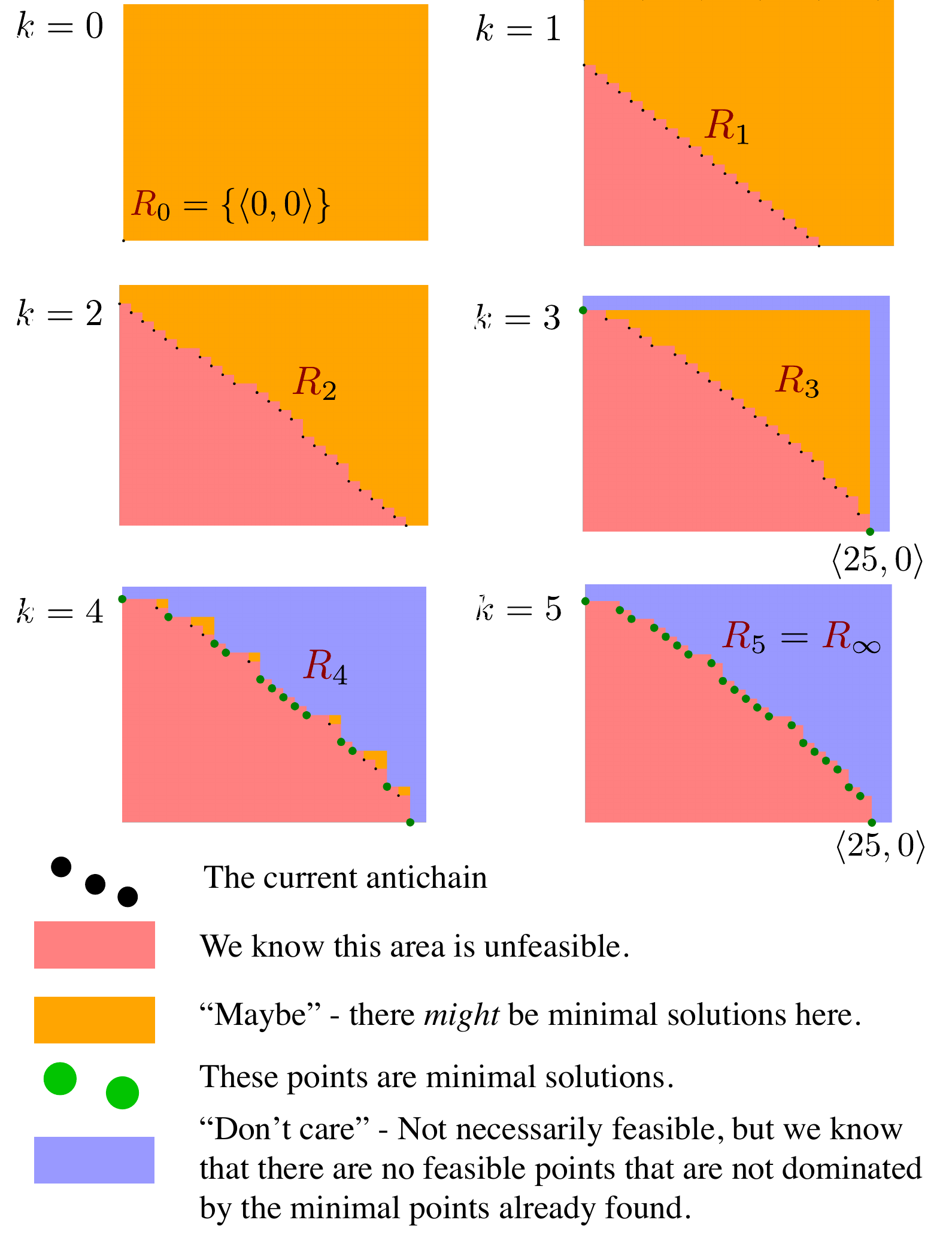}
\par\end{centering}
\caption{\label{fig:example24}Kleene ascent to solve the problem \eqref{example-1}
for $\F{c}=20$. The sequence converges in five steps to $\R{R}_{5}=\R{R}_{\infty}.$ }
\end{figure}

\noindent Kleene's algorithm is the iteration $\R{R}_{k+1}=\Psi_{\F{c}}(\R{R}_{k})$
starting from $\R{R}_{0}=\bot_{\R{\antichains(\overline{\mathbb{N}}\times\overline{\mathbb{N}})}}=\{\left\langle 0,0\right\rangle \}$.

\noindent For $\F{c}=0$, the sequence converges immediately: 
\[
\R{R}_{0}=\{\boldsymbol{\left\langle 0,0\right\rangle }\}=\ftor(\F{0}).
\]
 For $\F{c}=1$, the sequence converges at the second step:
\begin{align*}
\R{R}_{0} & =\{\left\langle 0,0\right\rangle \},\\
\R{R}_{1} & =\{\boldsymbol{\left\langle 0,1\right\rangle },\boldsymbol{\left\langle 1,0\right\rangle }\}=\ftor(\F{1}).
\end{align*}
For $\F{c}=2$, the sequence converges at the fourth step; however,
some solutions (in bold) converge sooner:
\begin{align*}
\R{R}_{0} & =\{\left\langle 0,0\right\rangle \},\\
\R{R}_{1} & =\{\left\langle 0,2\right\rangle ,\left\langle 1,1\right\rangle ,\left\langle 2,0\right\rangle \},\\
\R{R}_{2} & =\left\{ \boldsymbol{\left\langle 0,4\right\rangle },\left\langle 2,2\right\rangle ,\boldsymbol{\left\langle 4,0\right\rangle }\right\} ,\\
\R{R}_{3} & =\{\boldsymbol{\left\langle 0,4\right\rangle },\boldsymbol{\left\langle 3,3\right\rangle },\boldsymbol{\left\langle 4,0\right\rangle }\}=\ftor(\F{2}).
\end{align*}
The next values in the sequence are:
\begin{align*}
\ftor(\F{3}) & =\left\{ \boldsymbol{\left\langle 0,6\right\rangle },\boldsymbol{\left\langle 3,4\right\rangle },\boldsymbol{\left\langle 4,3\right\rangle },\boldsymbol{\left\langle 6,0\right\rangle }\right\} ,\\
\ftor(\F{4}) & =\left\{ \boldsymbol{\left\langle 0,7\right\rangle },\boldsymbol{\left\langle 3,6\right\rangle },\boldsymbol{\left\langle 4,4\right\rangle },\boldsymbol{\left\langle 6,3\right\rangle },\boldsymbol{\left\langle 7,0\right\rangle }\right\} .
\end{align*}
\figref{example24}~shows the sequence for $\F{c}=20$.

\subsubsection*{Guarantees of Kleene ascent}

\noindent Solving an MCDP with cycles reduces to computing a Kleene
ascent sequence~$\R{R}_{k}.$ At each instant~$k$ we have some
additional guarantees. 

For any finite~$k$, the resources ``below''~$\R{R}_{k}$ (the
set~$\ressp\ \backslash\upit\R{R}_{k}$,) are infeasible. (In \figref{example24},
those are colored in red.) 

If the iteration converges to a non-empty antichain~$\R{R}_{\infty}$,
the antichain~$\R{R}_{\infty}$ divides~$\ressp$ in two. Below
the antichain, all resources are infeasible. However, above the antichain
(purple area), it is not necessarily true that all points are feasible,
because there might be holes in the feasible set, as in \exaref{one}.
Note that this method does not compute the entire feasible set, but
rather only the \emph{minimal elements} of the feasible set, which
might be much easier to compute.

Finally, if the sequence converges to the empty set, it means that
there are no solutions. The sequence~$\R{R}_{k}$ can be considered
a certificate of infeasibility.

\subsection{Complexity of the solution }

\subsubsection*{1) Complexity of fixed point iteration}

Consider first the case of an MCDP that can be described as~$\dprob=\dploop(\dprob_{0})$,
where~$\dprob_{0}$ is an MCDP that is described only using the~$\dpseries$
and~$\dppar$ operators. Suppose that~$\dprob_{0}$ has resource
space~$\ressp$. Then evaluating~$\ftor$ for~$\dprob$ is equivalent
to computing a least fixed point iteration on the space of antichains~$\Aressp$.
This allows to give worst-case bounds on the number of iterations.

\begin{prop}
\label{prop:complexity}Suppose that~$\dprob=\dploop(\dprob_{0})$
and~$\dprob_{0}$ has resource space~$\ressp_{0}$ and evaluating~$\ftor_{0}$
takes at most~$c$ computation. Then we can obtain the following
bounds for the algorithm's resources usage:

\smallskip{}
\begin{tabular}{cc}
memory & $O(\posetwidth(\ressp_{0}))$\tabularnewline
number of steps & $O(\posetheight(\antichains\ressp_{0}))$\tabularnewline
total computation & $O(\posetwidth(\ressp_{0})\times\posetheight(\antichains\ressp_{0})\times c)$\tabularnewline
\end{tabular}

\end{prop}
\begin{proof}
The memory utilization is bounded by~$\posetwidth(\ressp_{0})$,
because the state is an antichain, and~$\posetwidth(\ressp_{0})$
is the size of the largest antichain. The iteration happens in the
space~$\antichains\ressp_{0}$, and we are constructing an ascending
chain, so it can take at most~$\posetheight(\antichains\ressp_{0})$
steps to converge. Finally, in the worst case the map~$\ftor_{0}$
needs to be evaluated once for each element of the antichain for each
step.
\end{proof}
These worst case bounds are strict. 
\begin{example}
Consider solving $\dprob=\dploop(\dprob_{0})$ with $\dprob_{0}$
defined by $\ftor_{0}\colon\left\langle \left\langle \right\rangle ,x\right\rangle \mapsto x+1$
with $x\in\overline{\mathbb{N}}$. Then the least fixed point equation
is equivalent to solving~$\min\{x\colon\Psi(x)\leq x\}$ with~$\Psi:x\mapsto x+1$.
The iteration~$R_{k+1}=\Psi(R_{k})$ converges to~$\top$ in~$\posetheight(\overline{\mathbb{N}})=\aleph_{0}$
steps.
\end{example}

\begin{rem}
Making more precise claims requires additional more restrictive assumptions
on the spaces involved. For example, without adding a metric on~$\ressp$,
it is not possible to obtain properties such as linear or quadratic
convergence. 
\end{rem}

\begin{rem}[Invariance to re-parameterization]
All the results given in this paper are invariant to any order-preserving
re-parameterization of all the variables involved. 
\end{rem}

\subsubsection*{2) Relating complexity to the graph properties}

\propref{complexity} above assumes that the MCDP is already in the
form $\dprob=\dploop(\dprob_{0})$, and relates the complexity to
the poset~$\ressp_{0}$. Here we relate the results to the graph
structure of an MCDP.

Take an MCDP~$\dprob=\left\langle \funsp,\ressp,\left\langle \cdpiN,\mathcal{E}\right\rangle \right\rangle $.
To put~$\dprob$ in the form $\dprob=\dploop(\dprob_{0})$ according
to the procedure in~\secref{Decomposition}, we need to find an arc
feedback set (AFS) of the graph~$\left\langle \cdpiN,\mathcal{E}\right\rangle $.
Given a AFS~$F\subset\mathcal{E}$, then the resource space~$\ressp_{0}$
for a~$\dprob_{0}$ such that~$\dprob=\dploop(\dprob_{0})$ is the
product of the resources spaces along the edges: $\ressp_{0}=\prod_{e\in F}\ressp_{e}.$

Now that we have a relation between the AFS and the complexity of
the iteration, it is natural to ask what is the optimal choice of
AFS\textemdash which, so far, was left as an arbitrary choice. The
AFS should be chosen as to minimize one of the performance measures
in~\propref{complexity}. 

Of the three performance measures in~\propref{complexity}, the most
fundamental appears to be~$\posetwidth(\ressp_{0})$, because that
is also an upper bound on the number of distinct minimal solutions.
Hence we can call it ``design complexity'' of the MCDP.
\begin{defn}
\label{def:design-complexity}Given a graph~$\left\langle \cdpiN,\mathcal{E}\right\rangle $
and a labeling of each edge~$e\in\mathcal{E}$ with a poset~$\ressp_{e}$,
the \emph{design complexity~}$\text{DC}(\left\langle \cdpiN,\mathcal{E}\right\rangle )$
is defined as 
\begin{equation}
\text{DC}(\left\langle \cdpiN,\mathcal{E}\right\rangle )=\min_{F\text{ is an AFS}}\posetwidth(\prod_{e\in F}\ressp_{e}).\label{eq:complexity}
\end{equation}
\end{defn}
In general, width and height of posets are not additive with respect
to products; therefore, this problem does not reduce to any of the
known variants of the minimum arc feedback set problem, in which
each edge has a weight and the goal is to minimize the sum of the
weights.

\subsubsection*{3) Considering relations with infinite cardinality}

This analysis shows the limitations of the simple solution presented
so far: it is easy to produce examples for which $\posetwidth(\ressp_{0})$
is infinite, so that one needs to represent a continuum of solutions.

\begin{example}
Suppose that the platform to be designed must travel \F{a distance~$d$
{[}m{]}}, and we need to choose the \R{endurance~$T$~{[}s{]}}
and \R{the velocity~$v$~{[}m/s{]}}. The relation among the quantities
is~${\colF d}\leq{\colR T\,v}.$ This is a design problem described
by the map
\begin{eqnarray*}
\ftor:{\colF\Rcomp} & \rightarrow & {\colR\antichains\Rcomp\times\Rcomp,}\\
{\colF d} & \mapsto & \{\langle{\colR T},{\colR v}\rangle\in{\colR\Rcomp\times\Rcomp}:\,{\colF d}={\colR T}\,{\colR v}\}.
\end{eqnarray*}
For each value of ${\colF d}$, there is a continuum of solutions.
\end{example}
One approach to solving this problem would be to discretize the functionality~$\funsp$
and the resources~$\ressp$ by sampling and/or coarsening. However,
sampling and coarsening makes it hard to maintain completeness and
consistency. 

One effective approach, outside of the scope of this paper, that allows
to use finite computation is to \emph{approximate the design problem}
\emph{itself}, rather than the spaces~$\funsp,\ressp$, which are
left as possibly infinite. The basic idea is that an infinite antichain
can be bounded from above and above by two antichains that have a
finite number of points. This idea leads to an algorithm that, given
a prescribed computation budget, can compute an inner and outer approximation
to the solution antichain~\cite{mcdp_icra_uncertainty_arxiv}.

\section{Extended Numerical Examples\label{sec:Numerical-examples}}

This example considers the choice of different battery technologies
for a robot. The goals of this example are: 1) to show how design
problems can be composed; 2) to show how to define hard constraints
and precedence between resources to be minimized; 3) to show how even
relatively simple models can give very complex trade-offs surfaces;
and 4) to introduce MCDPL, a formal language for the description of
MCDPs.

\begin{figure}
\begin{centering}
\subfloat[\label{fig:battery}Interface of battery design problem.]{\centering{}\includegraphics[scale=0.3]{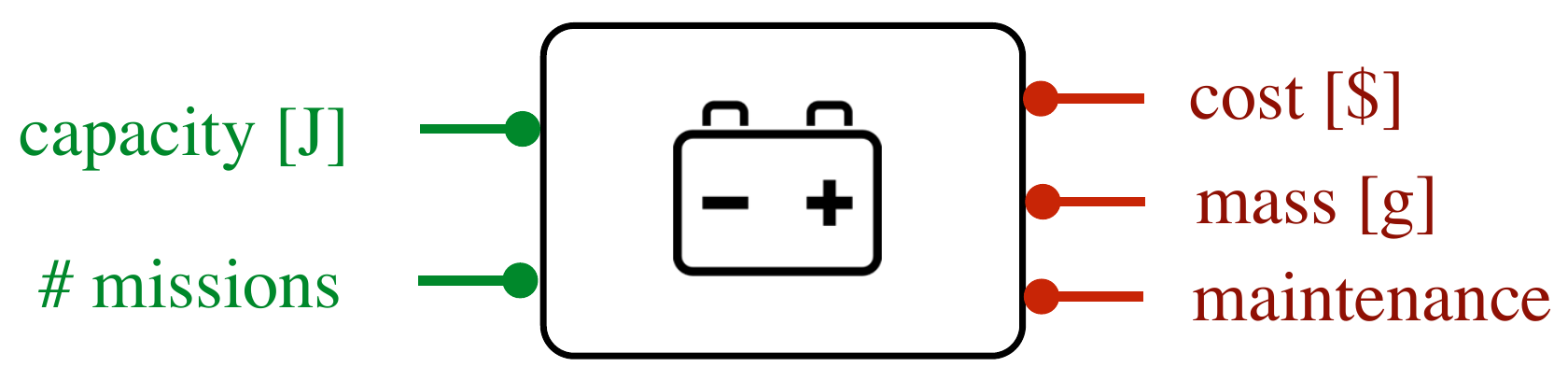}}
\par\end{centering}
\begin{centering}
\medskip{}
\subfloat[\label{fig:battery_nicad}MCDPL code equivalent to equations (\ref{eq:mass})\textendash (\ref{eq:cost}).]{\begin{centering}
\includegraphics[scale=0.55]{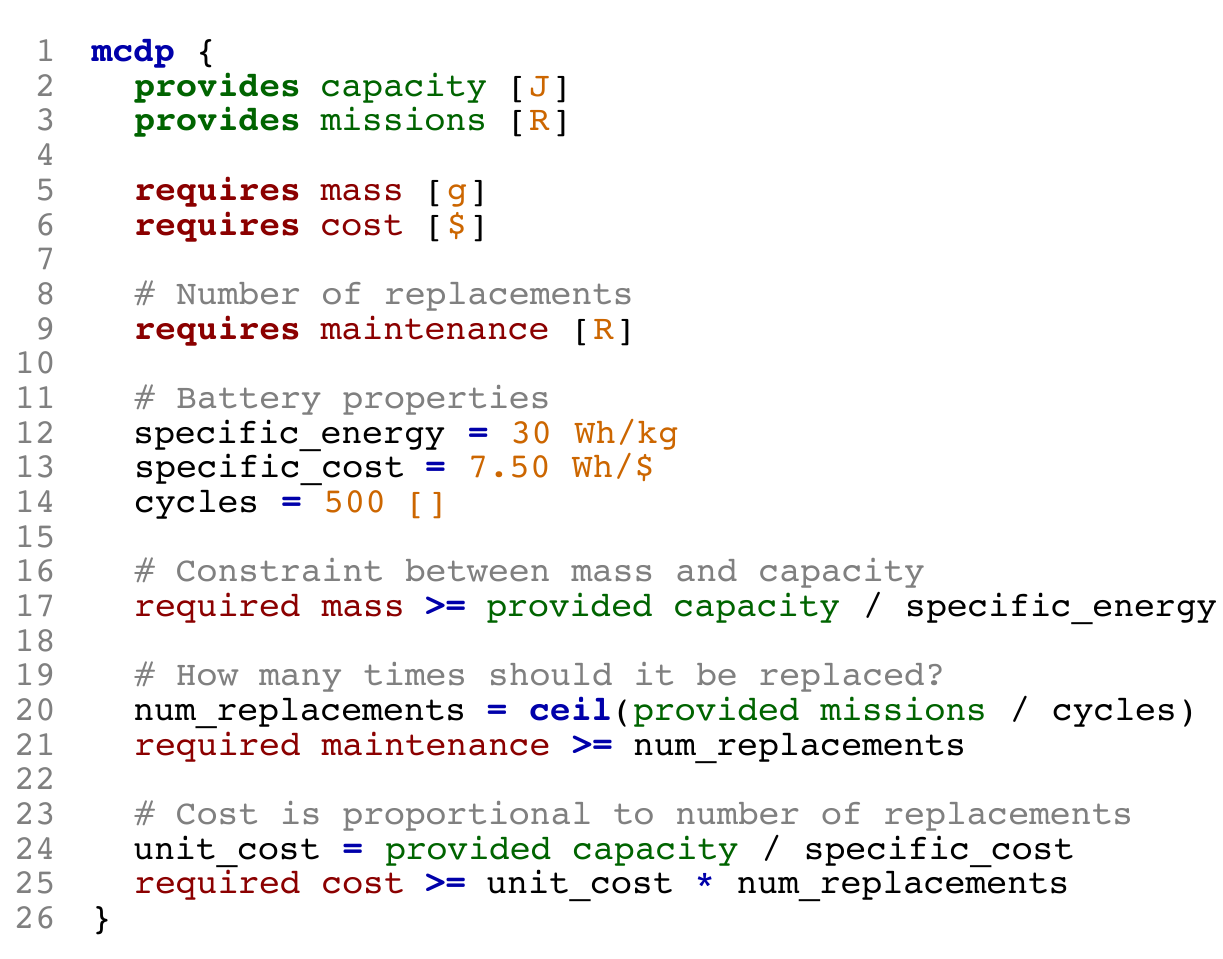}
\par\end{centering}
}
\par\end{centering}
\begin{centering}
\medskip{}
\subfloat[\label{fig:Co-design-diagram}Co-design diagram generated by PyMCDP
from code in panel (b).]{\begin{centering}
\hspace{2cm}\includegraphics[scale=0.35]{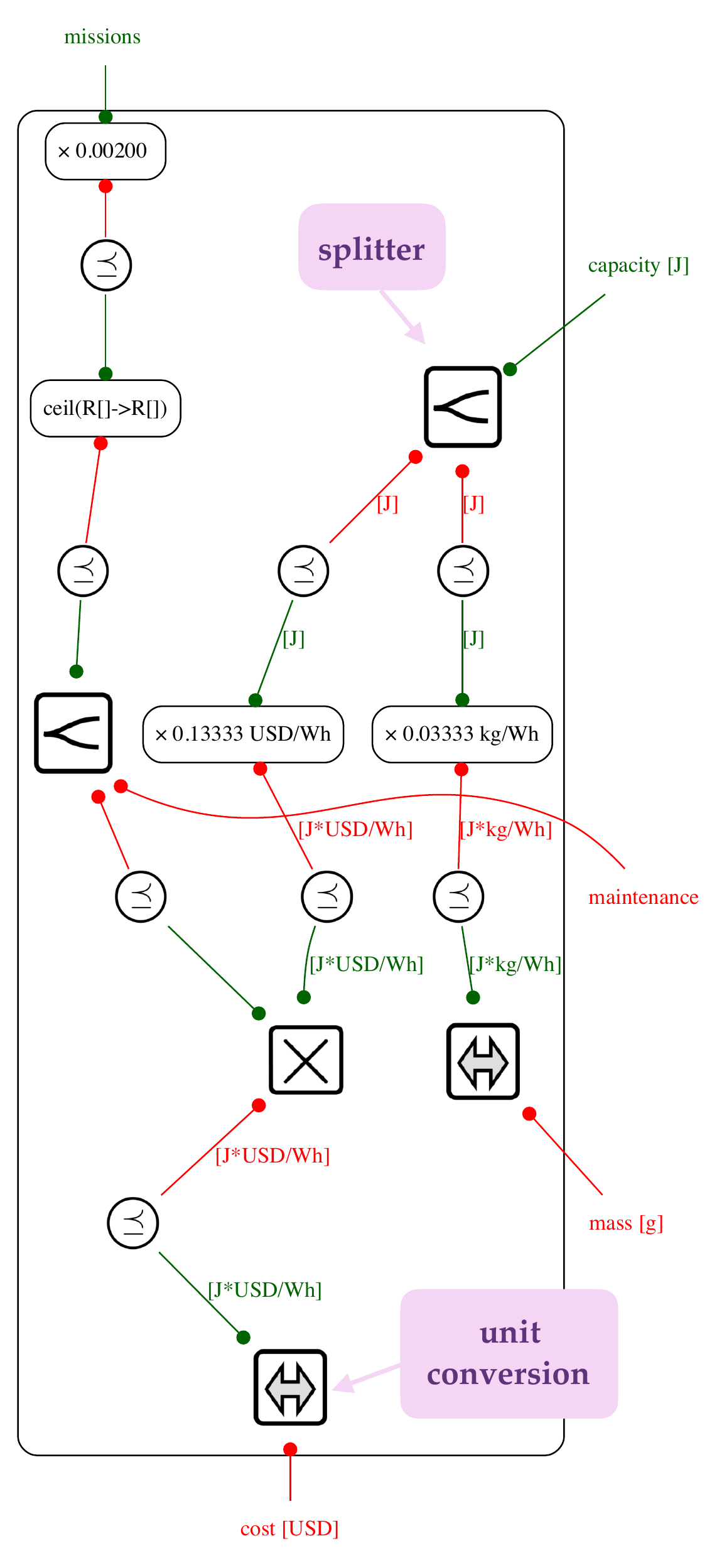}\hspace{2cm}
\par\end{centering}
}\medskip{}
\par\end{centering}
\begin{centering}
\subfloat[\label{fig:Tree-representation-of}Tree representation using $\dppar$/$\dpseries$
of diagram in panel (c).]{\begin{centering}
\includegraphics[width=7cm]{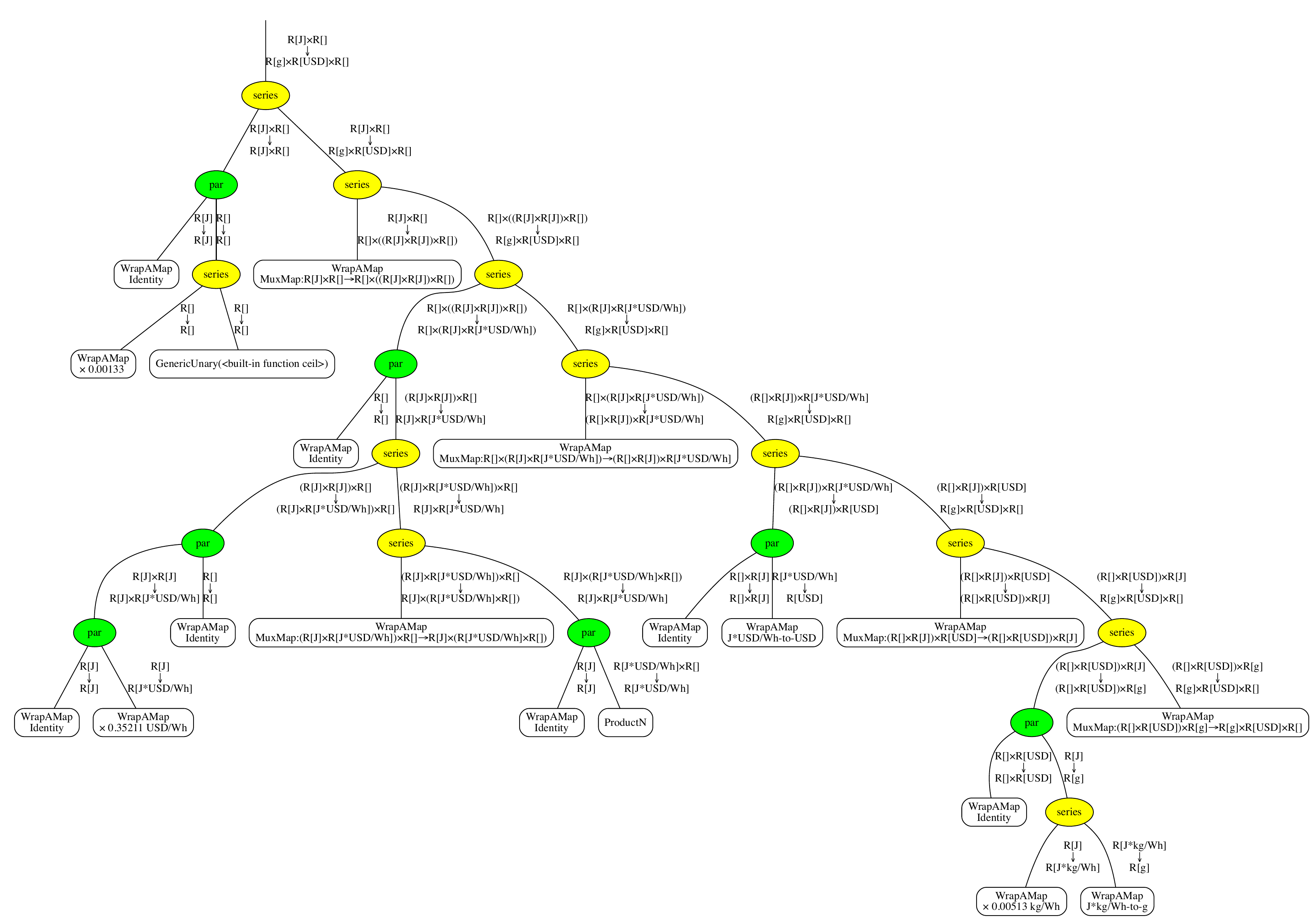}
\par\end{centering}
}
\par\end{centering}
\medskip{}

\medskip{}

\caption{Panel (c) shows the co-design diagram generated from the code in~(b).
Panel (d) shows a tree representation (series, parallel) for the diagram.
The edges show the types of functionality and resources. The leaves
are labeled with the Python class used internally by the interpreter
PyMCDP. }
\end{figure}

\subsubsection*{Language and interpreter/solver}

MCDPL is a modeling language to describe MCDPs and their compositions.
It is inspired by CVX and ``disciplined convex programming''~\cite{grant08graph}.
MCDPL is even more disciplined than CVX; for example, multiplying
by a negative number is a \emph{syntax} error. The figures are generated
by PyMCDP, an interpreter and solver for MCDPs, which implements the
techniques described in this paper.  The software and a manual are
available at \url{http://mcdp.mit.edu}. 

\subsubsection*{Model of a battery}

The choice of a battery can be modeled as a DPI (\figref{battery})
with functionalities \F{capacity {[}J{]}} and \F{number of missions}
and with resources \R{mass {[}kg{]}}, \R{cost {[}\${]}} and ``\R{maintenance}'',
defined as the number of times that the battery needs to be replaced
over the lifetime of the robot. 

Each battery technology is described by the three parameters specific
energy, specific cost, and lifetime (number of cycles):
\begin{align*}
\rho & \doteq\text{specific energy [Wh/kg]},\\
\alpha & \doteq\text{specific cost [Wh/\$]},\\
c & \doteq\text{battery lifetime [\# of cycles]}.
\end{align*}
The relation between functionality and resources is described by three
nonlinear monotone constraints: 
\begin{align}
\R{\text{mass}} & \geq\F{\text{capacity}}/\rho,\label{eq:mass}\\
\R{\text{maintenance}} & \geq\left\lceil \F{\text{missions}}/c\right\rceil ,\label{eq:maintenance}\\
\R{\text{cost}} & \geq\left\lceil \F{\text{missions}}/c\right\rceil (\F{\text{capacity}}/\alpha).\label{eq:cost}
\end{align}

\figref{battery_nicad} shows the MCDPL code that describes the
model corresponding to (\ref{eq:mass})\textendash (\ref{eq:cost}).
The diagram in \figref{Co-design-diagram} is automatically generated
from the code. \figref{Tree-representation-of}~shows a tree representation
of the diagram using the $\dpseries$/$\dppar$ operators. 

\subsubsection*{Competing battery technologies}

The parameters for the battery technologies used in this example are
shown in~~\tabref{batteries}. 

\begin{table}[H]
\begin{centering}
\caption{\label{tab:batteries}Specifications of common batteries technologies}
\par\end{centering}
\centering{}\setlength\extrarowheight{0.5pt}{\footnotesize{}}
\begin{tabular}{crrr}
\multirow{2}{*}{{\footnotesize{}\tableColors}\emph{\footnotesize{}technology}} & \emph{\footnotesize{}energy density} & \emph{\footnotesize{}specific cost} & \emph{\footnotesize{}operating life}\tabularnewline
 & {\footnotesize{}{[}Wh/kg{]}} & {\footnotesize{}{[}Wh/\${]}} & \# cycles\tabularnewline
{\footnotesize{}NiMH} & {\footnotesize{}100} & {\footnotesize{}3.41} & {\footnotesize{}500 }\tabularnewline
{\footnotesize{}NiH2} & {\footnotesize{}45} & {\footnotesize{}10.50} & {\footnotesize{}20000}\tabularnewline
{\footnotesize{}LCO} & {\footnotesize{}195} & {\footnotesize{}2.84} & {\footnotesize{}750}\tabularnewline
{\footnotesize{}LMO} & {\footnotesize{}150} & {\footnotesize{}2.84} & {\footnotesize{}500}\tabularnewline
{\footnotesize{}NiCad} & {\footnotesize{}30} & {\footnotesize{}7.50} & {\footnotesize{}500}\tabularnewline
{\footnotesize{}SLA} & {\footnotesize{}30} & {\footnotesize{}7.00} & {\footnotesize{}500}\tabularnewline
{\footnotesize{}LiPo} & {\footnotesize{}250} & {\footnotesize{}2.50} & {\footnotesize{}600}\tabularnewline
{\footnotesize{}LFP} & {\footnotesize{}90} & {\footnotesize{}1.50} & {\footnotesize{}1500}\tabularnewline
\end{tabular}{\footnotesize \par}
\end{table}

Each row of the table is used to describe a model as in \figref{battery_nicad}
by plugging in the specific values in lines 12\textendash 14.

Given the different models, we can defined their co-product (\figref{Co-product-of-battery})
using the MCDPL code in~\figref{batteries_code}. 

\begin{figure}[H]
\subfloat[Co-product of battery technologies\label{fig:Co-product-of-battery}]{\centering{}\includegraphics[scale=0.33]{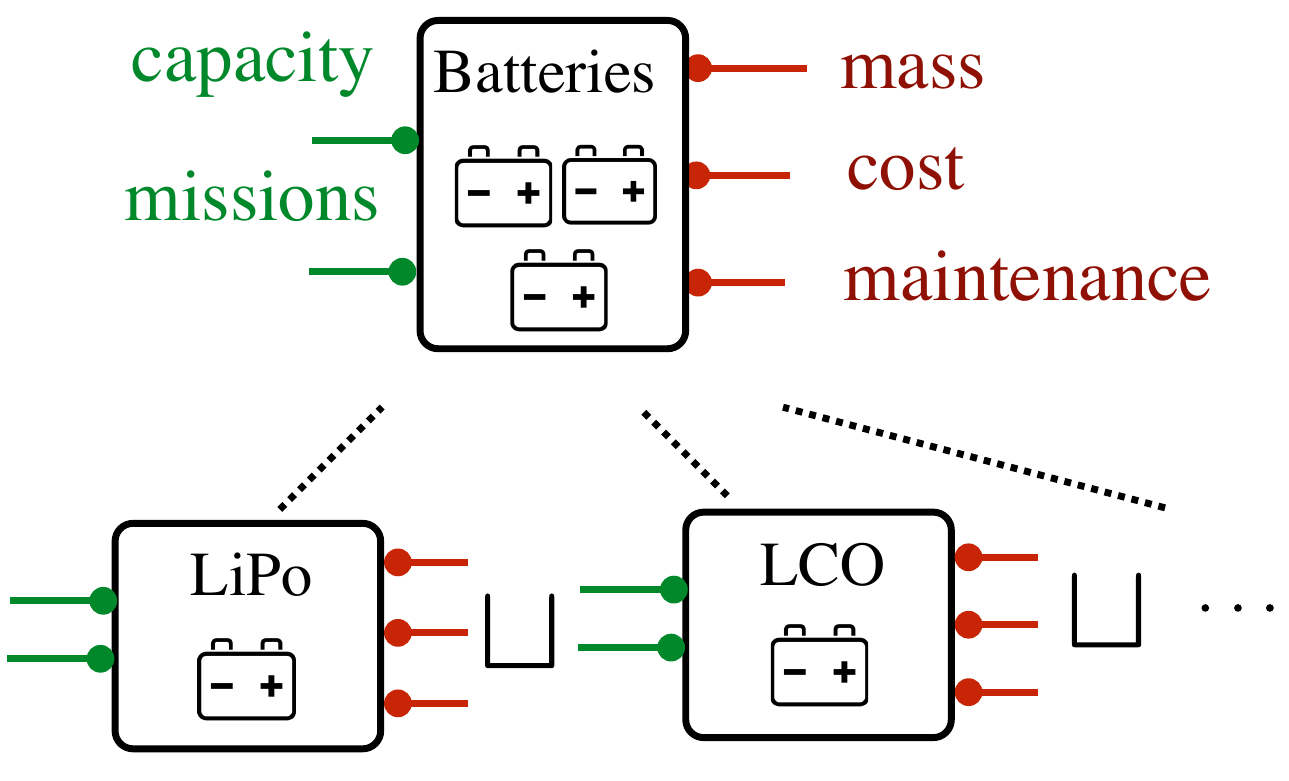}}\subfloat[\label{fig:batteries_code}Batteries.mcdp]{\begin{centering}
\includegraphics[scale=0.66]{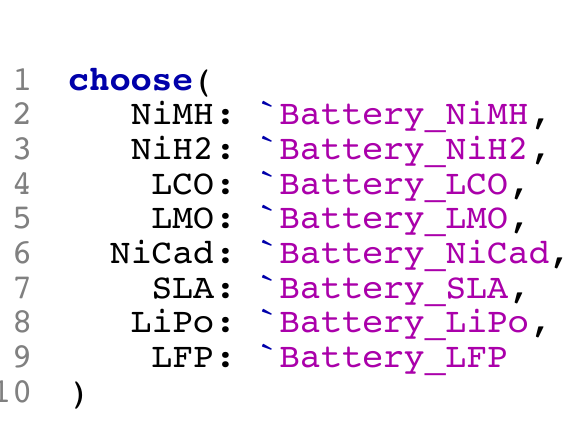}
\par\end{centering}
}

\smallskip{}

\caption{\label{fig:batteriesbig}The co-product of design problems describes
the choices among different technologies. The MCDPL keyword for the
co-product is ``choose''.}
\end{figure}

\subsubsection*{Introducing other variations or objectives}

The design problem for the battery has two functionalities (\F{capacity}
and \F{number of missions}) and three resources (\R{cost}, \R{mass},
and \R{maintenance}). Thus, it describes a family of multi-objective
optimization problems, of the type ``Given \F{capacity} and \F{missions},
minimize $\left\langle \R{\text{cost}},\R{\text{mass}},\R{\text{maintenance}}\right\rangle $''.
We can further extend the class of optimization problems by introducing
other hard constraints and by choosing which resource to prioritize.
This can be done by composition of design problems; that is, by creating
a larger DP that contains the original DP as a subproblem, and contains
some additional degenerate DPs that realize the desired semantics.

For example, suppose that we would like to find the optimal solution(s)
such that: 1) The mass does not exceed 3 kg; 2) The mass is minimized
as a primary objective, while cost/maintenance are secondary objectives.

This semantics can be described by the co-design diagram in~\figref{diagram},
which contains two new symbols. The DP labeled ``3 kg'' implements
the semantics of hard constraints. It has one functionality ($\funsp=\F{\Rcomp^{\text{kg}}}$)
and zero resources~($\ressp=\R{\One}$). The poset~$\One=\{\left\langle \right\rangle \}$
has exactly two antichains: $\emptyset$ and $\{\left\langle \right\rangle \}$.
These represent ``infeasible'' and ``feasible'', respectively.
The DP is described by the map

\quad\quad
\begin{minipage}[c]{5cm}
\begin{align*}
\ftor:\F{\mathbb{\overline{R}}_{+}^{\text{kg}}} & \rightarrow\R{\antichains\One},\\
\fun & \mapsto\begin{cases}
\R{\{\left\langle \right\rangle \}}, & \text{if }\fun\leq\text{3 kg},\\
\R{\emptyset}, & \text{if }\fun>\text{3 kg}.
\end{cases}
\end{align*}

\end{minipage}\quad\includegraphics[scale=0.45]{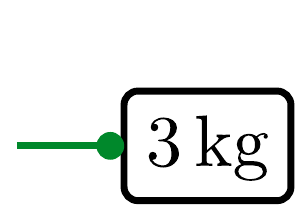}

\smallskip{}

\noindent The block labeled~``$\top$'' is similarly defined and
always returns ``feasible'', so it has the effect of ignoring \R{cost}
and \R{maintenance} as objectives. The only resource edge is the
one for \R{mass}, which is then the only objective.

The MCDPL code is shown in~\figref{diagram_code}. Note the intuitive
interface: the user can directly write ``mass required by battery
$\leq$ 3 kg'' and ``ignore maintenance required by battery'',
which is compiled to ``maintenance required by battery $\leq\top$''.

\begin{figure}
\begin{centering}
\subfloat[\label{fig:diagram}Co-design diagram that expresses hard constraints
for \R{mass}.]{\begin{centering}
\begin{minipage}[t]{8.6cm}
\begin{center}
\includegraphics[scale=0.4]{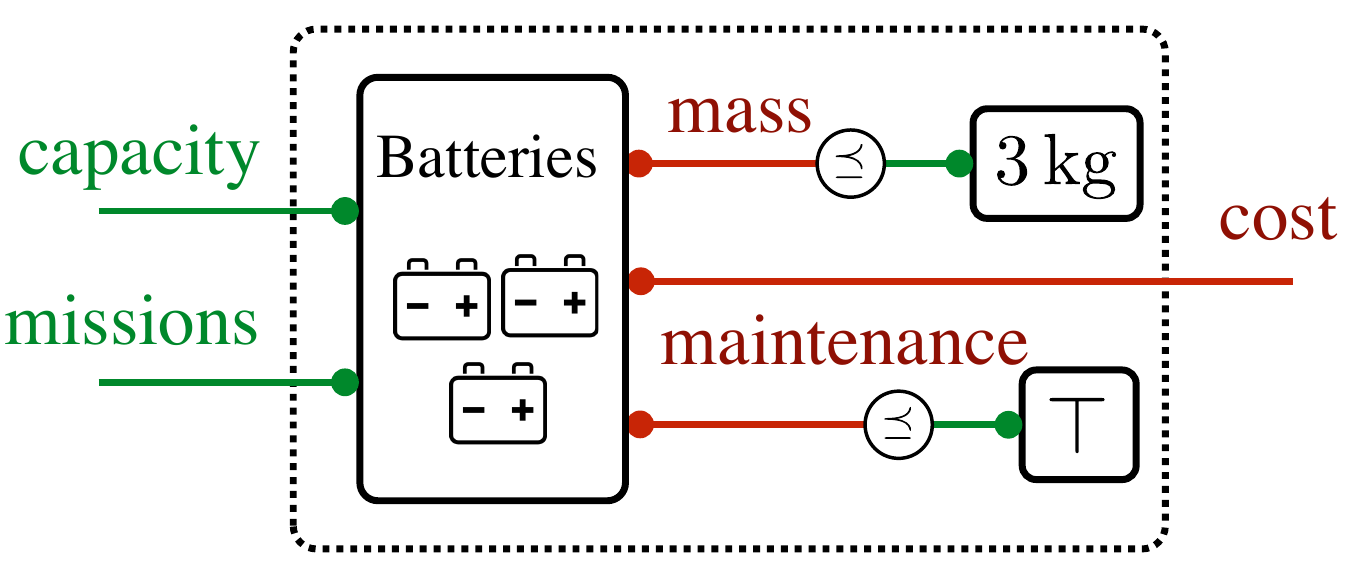}
\par\end{center}
\end{minipage}
\par\end{centering}
}
\par\end{centering}
\begin{centering}
\subfloat[\label{fig:diagram_code}MCDPL code equivalent to diagram in (a).]{\begin{centering}
\includegraphics[width=6cm]{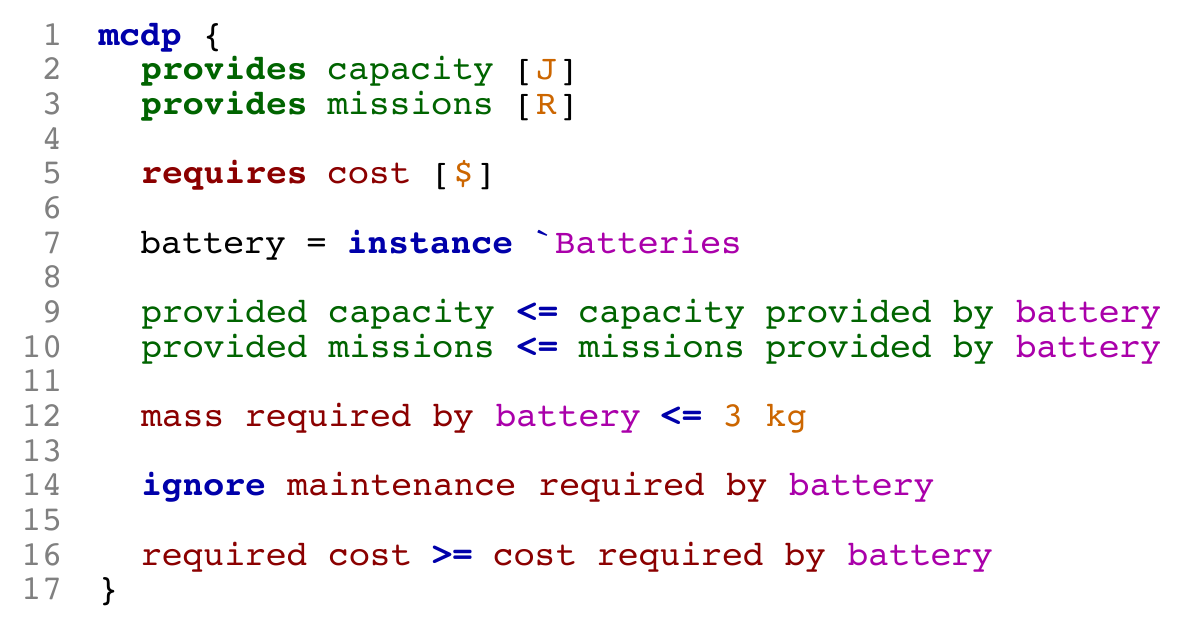}
\par\end{centering}
}
\par\end{centering}
\smallskip{}

\caption{Composition of MCDPs can express hard constraints and precedence of
objectives. In this case, there is a hard constraint on the \R{mass}.
Because there is only one outgoing edge for \R{mass}, and the \R{cost}
and \R{maintenance} are terminated by a dummy constraint ($x\posleq\top$),
the semantics of the diagram is that the objective is to minimize
the \R{mass} as primary objective.}
\end{figure}

This relatively simple model for energetics already shows the complexity
of MCDPs. \figref{mainbattery} shows the optimal choice of the battery
technology as a function of capacity and number of missions, for several
slight variations of the problem that differ in constraints and objectives.
For each battery technology, the figures show whether at each operating
point the technology is the optimal choice, and how many optimal choices
there are. Some of the results are intuitive. For example, \figref{min_mass}
shows that if the only objective is minimizing \R{mass}, then the
optimal choice is simply the technology with largest specific energy
(LiPo). The decision boundaries become complex when considering nonlinear
objectives. For example, \figref{min_cost} shows the case where the
objective is to minimize the \R{cost}, which, defined by~\eqref{cost},
is nonlinearly related to both \F{capacity} and \F{number of missions}.
When considering multi-objective problems, such as minimizing jointly
$\left\langle \R{\text{mass}},\R{\text{cost}}\right\rangle $~(\figref{min_mass_cost})
or~$\left\langle \R{\text{mass}},\R{\text{cost}},\R{\text{maintenance}}\right\rangle $
(\figref{min_mass_cost}), there are multiple non-dominated solutions.

\begin{figure}
\begin{centering}
\subfloat[\label{fig:drone}\label{fig:drone_dia}Co-design diagram corresponding
to (\ref{eq:drone_eq_first})\textendash (\ref{eq:drone_eq_last}).]{\centering{}\includegraphics[scale=0.38]{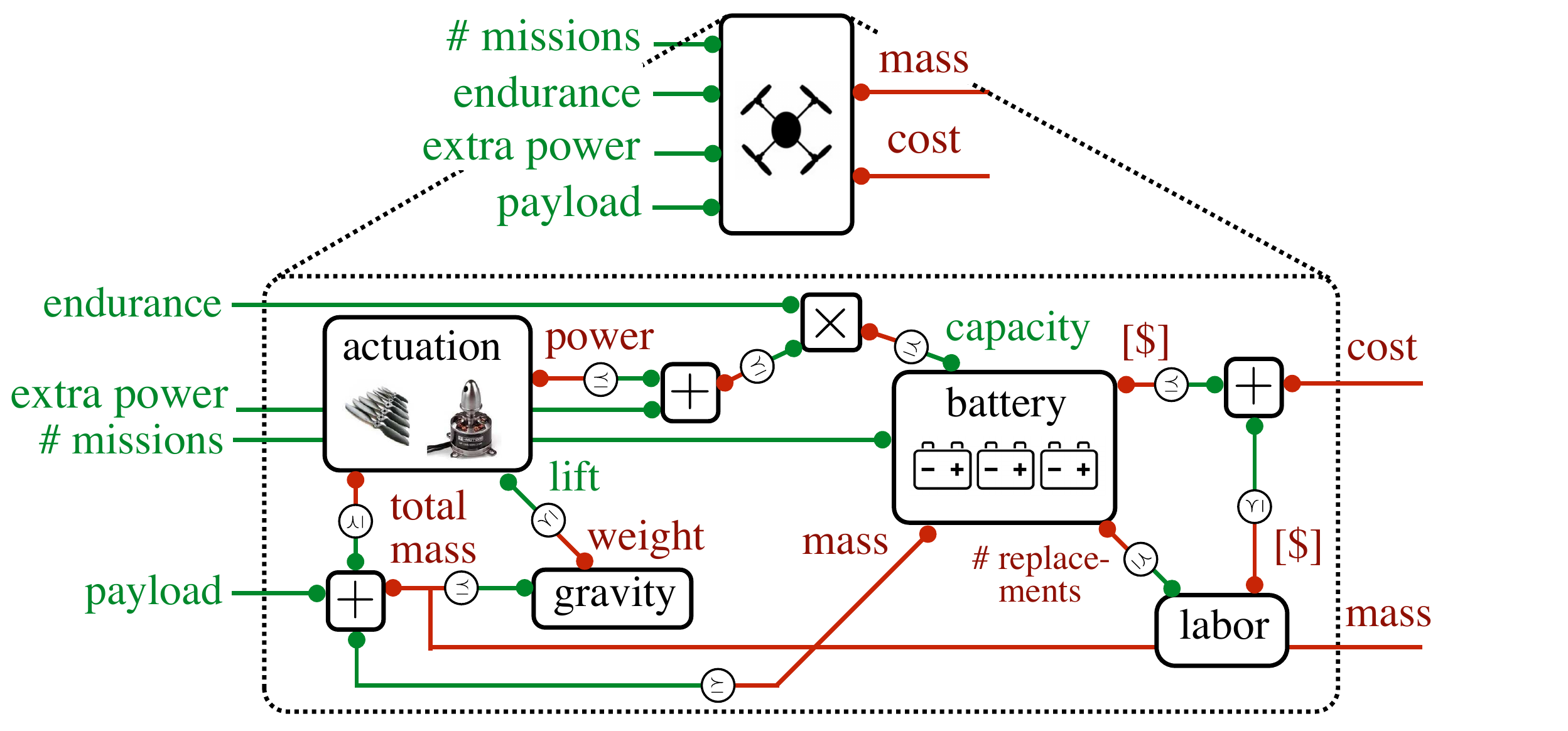}}
\par\end{centering}
\begin{centering}
\subfloat[\label{fig:drone_code}MCDPL code for (\ref{eq:drone_eq_first})\textendash (\ref{eq:drone_eq_last}).
The ``instance'' statements refer to previously defined models for
batteries (\figref{batteries_code}) and actuation (not shown). ]{\begin{centering}
\begin{minipage}[t]{8.6cm}
\begin{center}
\includegraphics[scale=0.55]{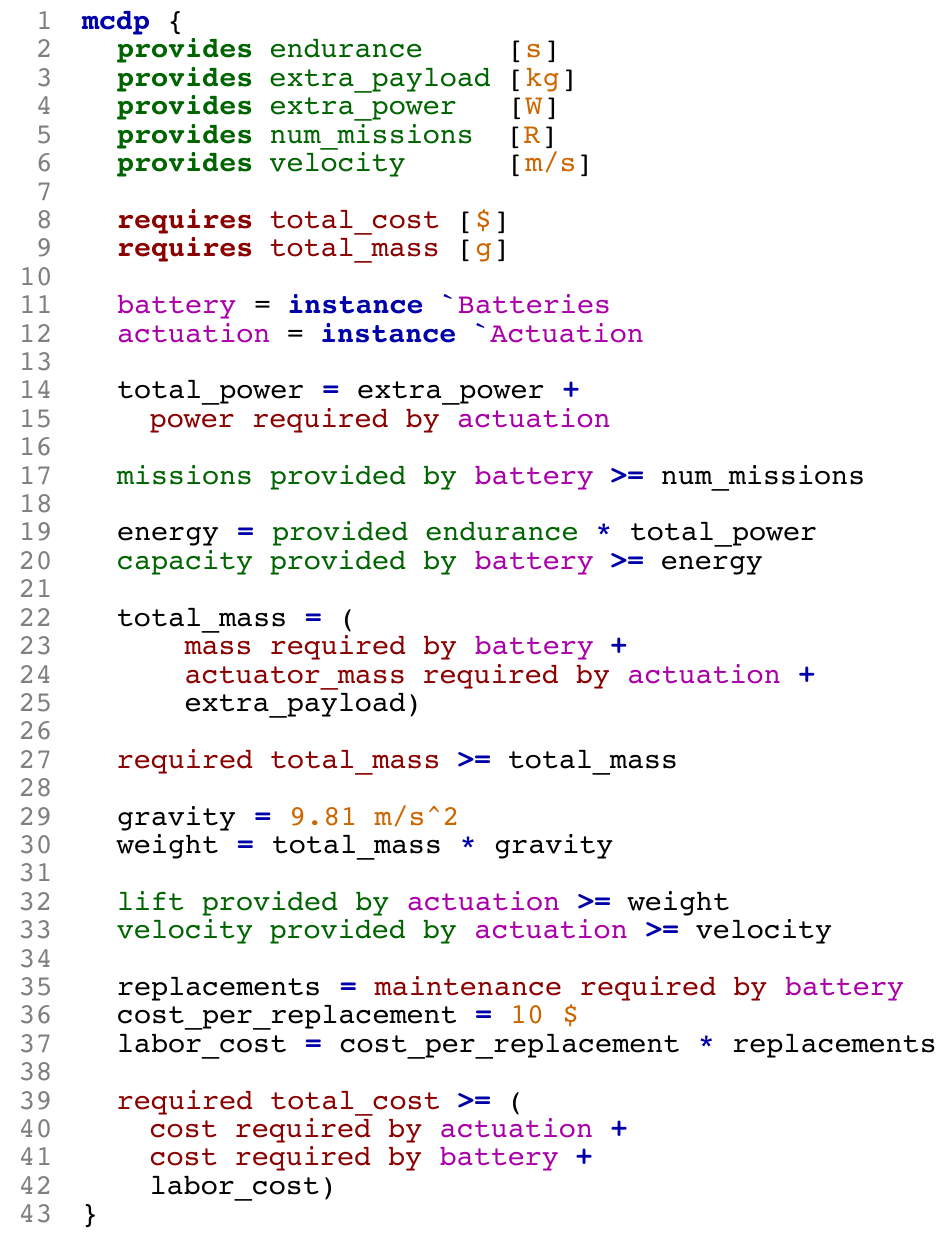}
\par\end{center}
\end{minipage}
\par\end{centering}
}
\par\end{centering}
\subfloat[\label{fig:drone_tree}Tree representation for the MCDP. Yellow/green
rounded ovals are $\dpseries$/$\dppar$ junctions. There is one co-product
junction, signifying the choice between different battery technologies,
and one~$\dploop$ junction, at the root of the tree.]{\begin{centering}
\includegraphics[width=8.6cm]{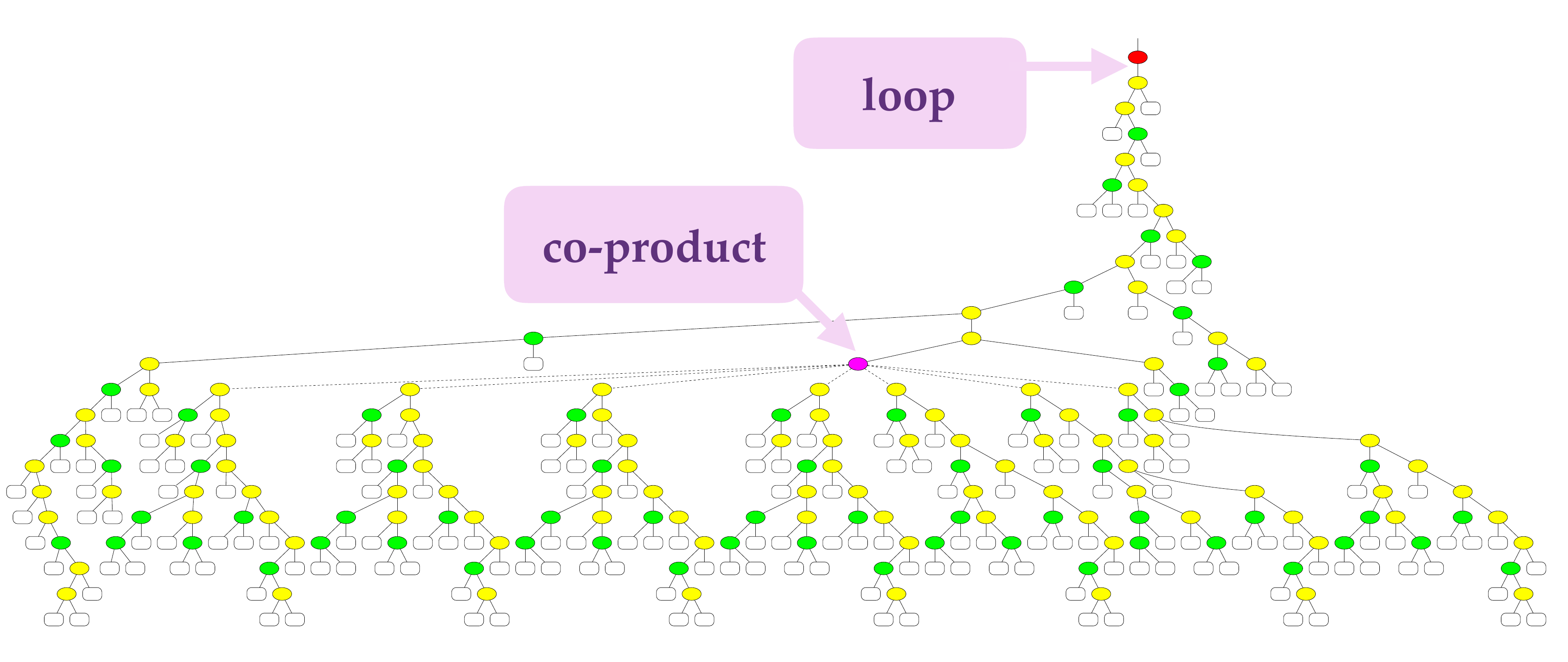}
\par\end{centering}
}

\subfloat[\label{fig:drone-endurance-missions}Relation between \F{endurance}
and \F{number of missions} and \R{cost} and \R{mass}. ]{\begin{centering}
\includegraphics[width=8.6cm]{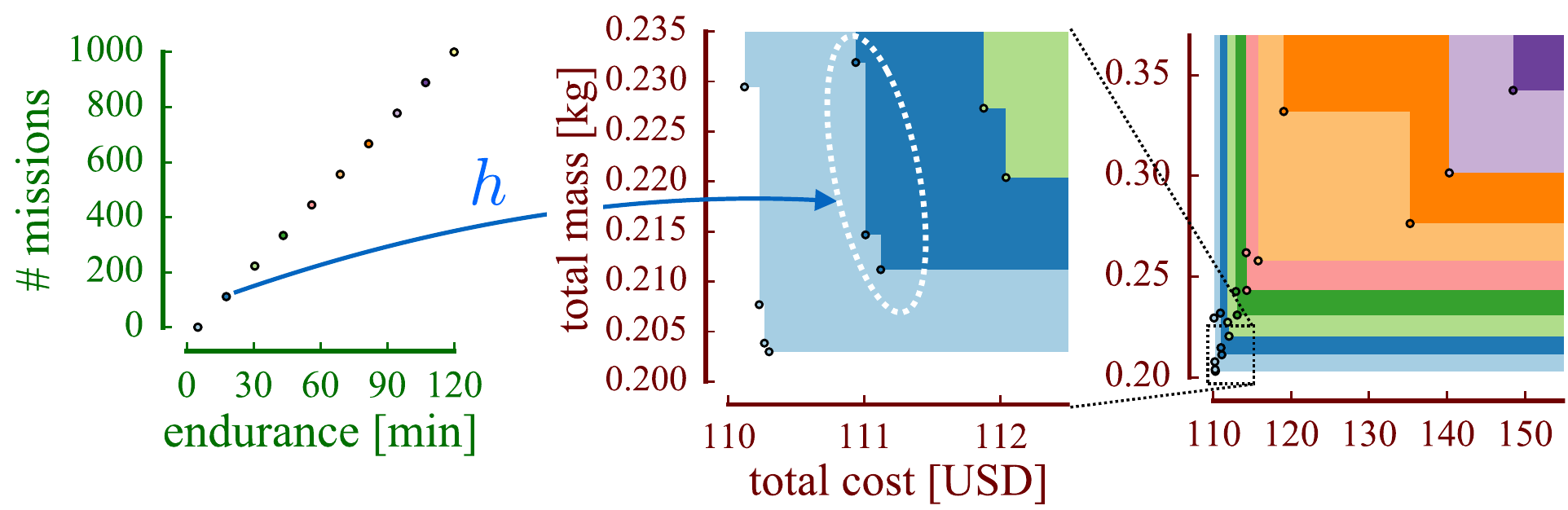}
\par\end{centering}
}

\subfloat[\label{fig:drone-endurance-payload}Relation between \F{endurance}
and \F{payload} and \R{cost} and \R{mass}. ]{\begin{centering}
\includegraphics[width=8.6cm]{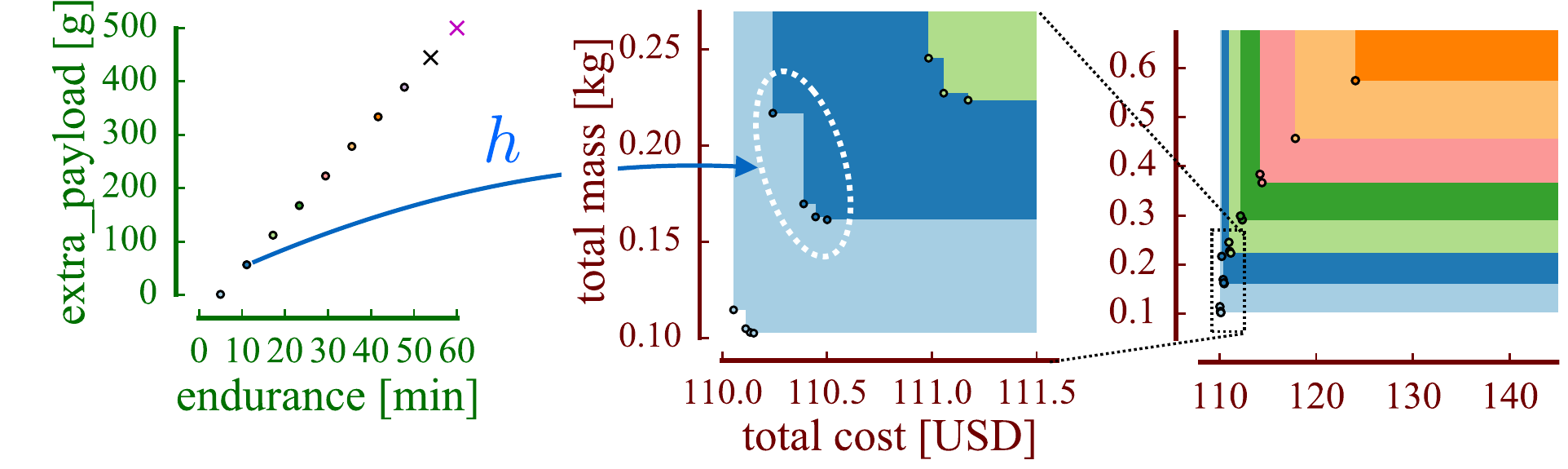}
\par\end{centering}
}

\caption{\label{fig:dronebigfig}In panel (c), the \F{payload} is fixed to
100 g and \F{extra power} is set to 1~W. In panel (d), the \F{number
of missions} is fixed to 400 and \F{extra power} is set to 1~W.
The last two values, marked with ``$\times$'', are not feasible.}
\end{figure}

\subsubsection*{From component to system co-design}

The rest of the section reuses the battery DP into a larger co-design
problem that considers the co-design of actuation together with energetics
for a drone~(\figref{drone_dia}). We will see that the decision
boundaries change dramatically, which shows that the optimal choices
for a component cannot be made in isolation from the system.

The functionality of the drone's subsystem considered~(\figref{drone})
are parametrized by \F{endurance}, \F{number of missions}, \F{extra
power} to be supplied, and \F{payload}. We model ``actuation''
as a design problem with functionality \F{lift {[}N{]}} and resources
\R{cost}, \R{mass} and \R{power}, and we assume that power is
a quadratic function of lift~(\figref{actuation}). Any other monotone
map could be used.

\captionsideleft{\label{fig:actuation}}{\includegraphics[scale=0.33]{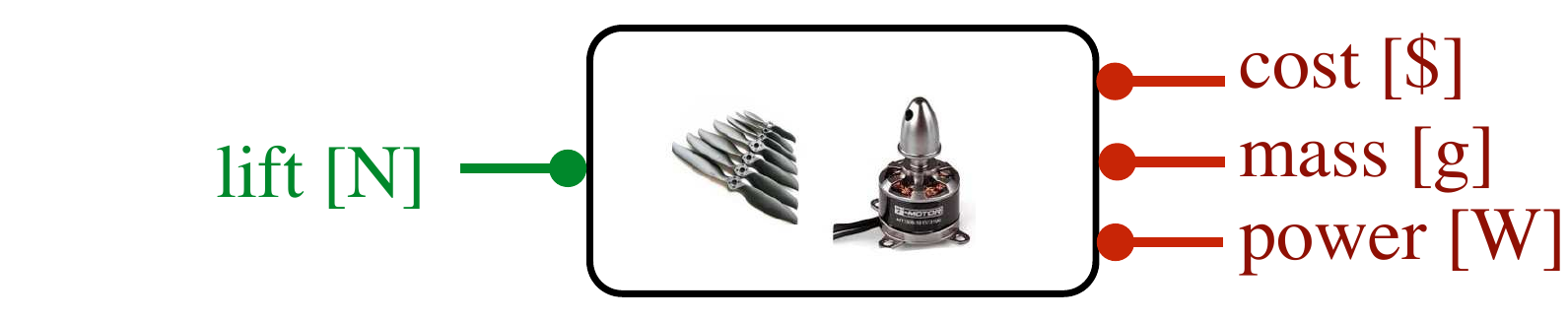}}

\noindent The co-design constraints that combine energetics and actuation
are{\small{}
\begin{align}
\text{battery }\F{\text{capacity}} & \geq\text{total power}\times\F{\text{endurance}},\label{eq:drone_eq_first}\\
\text{total power} & =\text{actuation }\R{\text{power}}+\F{\text{extra power}},\nonumber \\
\text{weight} & =\text{total mass}\times\text{gravity},\nonumber \\
\text{actuation}\,\F{\text{lift}} & \geq\text{weight},\nonumber \\
\text{labor cost} & =\text{cost per replacement}\times\text{battery }\R{\text{maintenance}},\nonumber \\
\R{\text{total cost}} & =\text{battery}\,\R{\text{cost}}+\text{actuation}\,\R{\text{cost}}+\text{labor cost},\nonumber \\
\R{\text{total mass}} & =\text{battery}\,\R{\text{mass}}+\text{actuation}\,\R{\text{mass}}+\F{\text{payload}}.\label{eq:drone_eq_last}
\end{align}
}The co-design graph contains recursive constraints: the power for
actuation depends on the total weight, which depends on the mass of
the battery, which depends on the capacity to be provided, which depends
on the power for actuation. The MCDPL code  for this model is shown
in \figref{drone_code}; it refers to the previously defined models
for ``batteries'' and ``actuation''.

The co-design problem is now complex enough that we can appreciate
the compositional properties of MCDPs to perform a qualitative analysis.
Looking at~\figref{drone}, we know that there is a monotone relation
between any pair of functionality and resources, such as \F{payload}
and \R{cost}, or \F{endurance} and \R{mass}, even without knowing
exactly what are the models for battery and actuation.

When fully expanded, the co-design graph (too large to display) contains
110 nodes and 110 edges. It is possible to remove all cycles by removing
only one edge (e.g., the $\R{\text{energy}}\leq\text{\F{\text{capacity}}}$
constraint), so the design complexity (\defref{design-complexity})
is equal to~$\posetwidth(\R{\overline{\mathbb{R}}_{+}})=1$. The
tree representation is shown in~\figref{drone_tree}. Because the
co-design diagram contains cycles, there is a~$\dploop$ operator
at the root of the tree, which implies we need to solve a least fixed
point problem. Because of the scale of the problem, it is not possible
to show the map~$\ftor$ explicitly, like we did in (\ref{eq:expression})
for the previous example. The least fixed point sequence converges
to 64 bits machine precision in 50-100 iterations.

To visualize the multidimensional relation
\[
\ftor\colon\F{\Rcomp\times\Rcomp^{\text{s}}\times\Rcomp^{\text{W}}\times\Rcomp^{\text{g}}}\rightarrow\R{\antichains(\Rcomp^{\text{kg}}\times\Rcomp^{\text{USD}})},
\]
we need to project across 2D slices. \figref{drone-endurance-missions}~shows
the relation when the functionality varies in a chain in the space
\F{endurance}/\F{missions}, and~\figref{drone-endurance-payload}
shows the results for a chain in the space \F{endurance}/\F{payload}.

Finally, \figref{drone_choice} shows the optimal choices of battery
technologies in the \F{endurance}/\F{missions} space, when one
wants to minimize \R{mass}, \R{cost}, or~$\left\langle \R{\text{mass}},\R{\text{cost}}\right\rangle $.
The decision boundaries are completely different from those in~\figref{mainbattery}.
This shows that it is not possible to optimize a component separately
from the rest of the system, if there are cycles in the co-design
diagram.

\begin{figure*}
\begin{centering}
\hfill\includegraphics[width=0.8\textwidth]{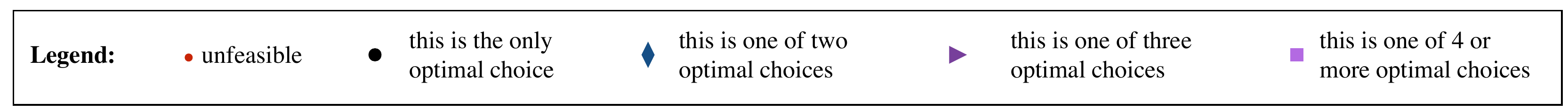}\rule{2cm}{0pt}
\par\end{centering}
\begin{centering}
\subfloat[\label{fig:min_joint_dia}]{\begin{centering}
\includegraphics[scale=0.33]{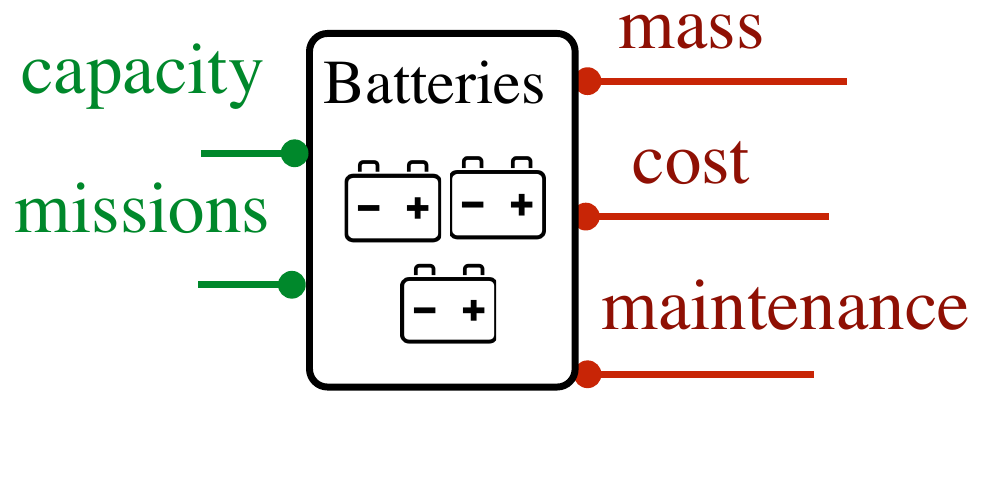}
\par\end{centering}
\vspace{5mm}

}\subfloat[\label{fig:min_joint}]{\begin{centering}
\includegraphics[scale=0.45]{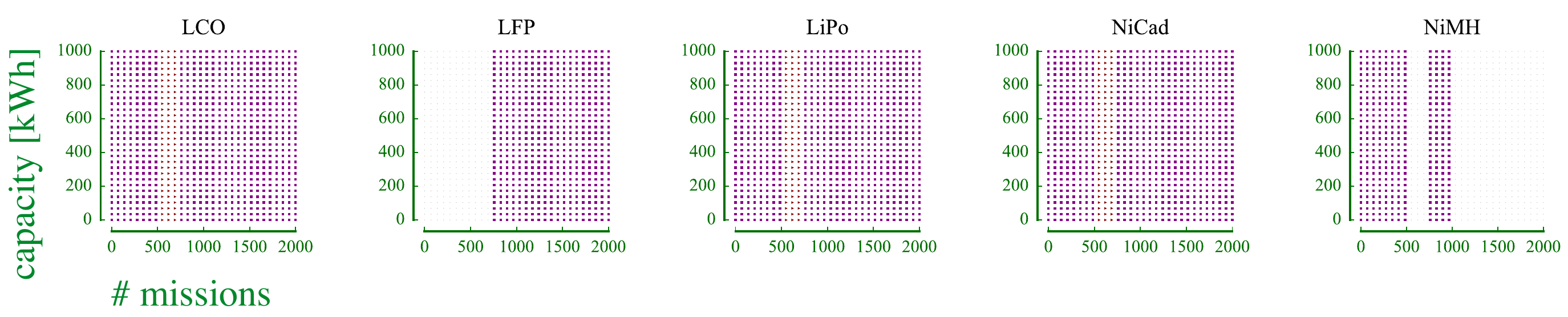}
\par\end{centering}
}
\par\end{centering}
\begin{centering}
\subfloat[\label{fig:min_cost_dia}]{\begin{centering}
\includegraphics[scale=0.33]{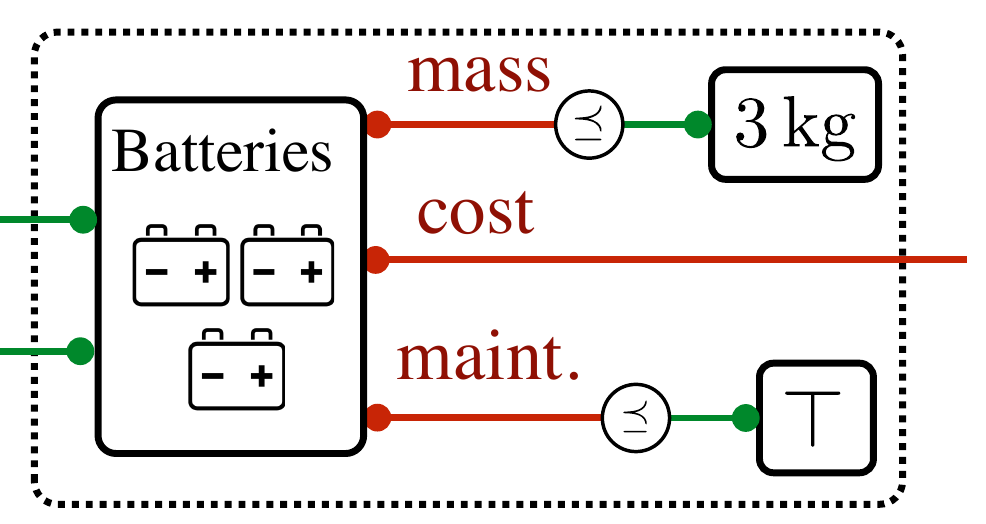}
\par\end{centering}
}\subfloat[\label{fig:min_cost}]{\begin{centering}
\includegraphics[scale=0.45]{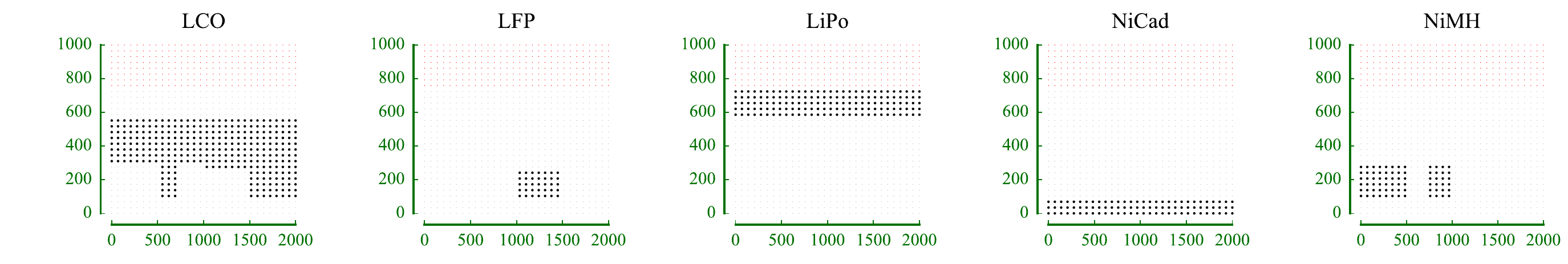}
\par\end{centering}
}
\par\end{centering}
\begin{centering}
\subfloat[]{\begin{centering}
\includegraphics[scale=0.33]{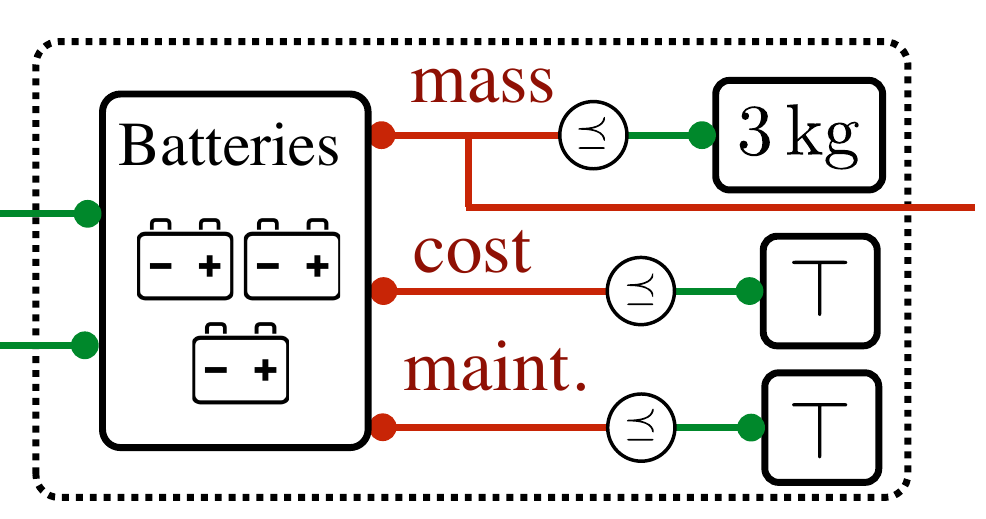}
\par\end{centering}
}\subfloat[\label{fig:min_mass}]{\begin{centering}
\includegraphics[scale=0.45]{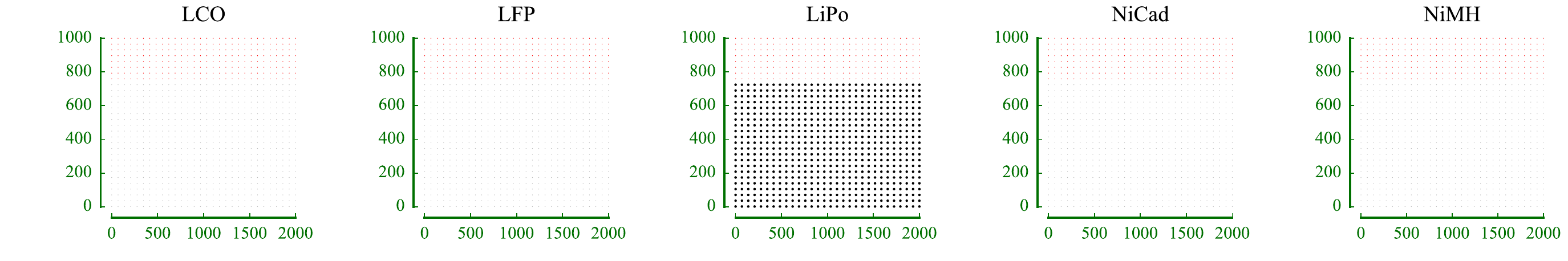}
\par\end{centering}
}
\par\end{centering}
\begin{centering}
\par\end{centering}
\begin{centering}
\subfloat[]{\begin{centering}
\includegraphics[scale=0.33]{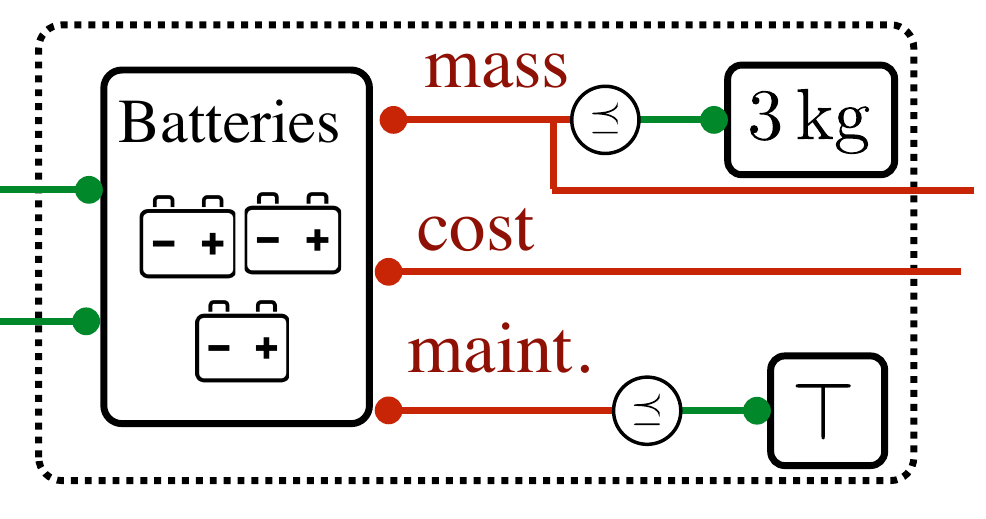}
\par\end{centering}
}\subfloat[\label{fig:min_mass_cost}]{\begin{centering}
\includegraphics[scale=0.45]{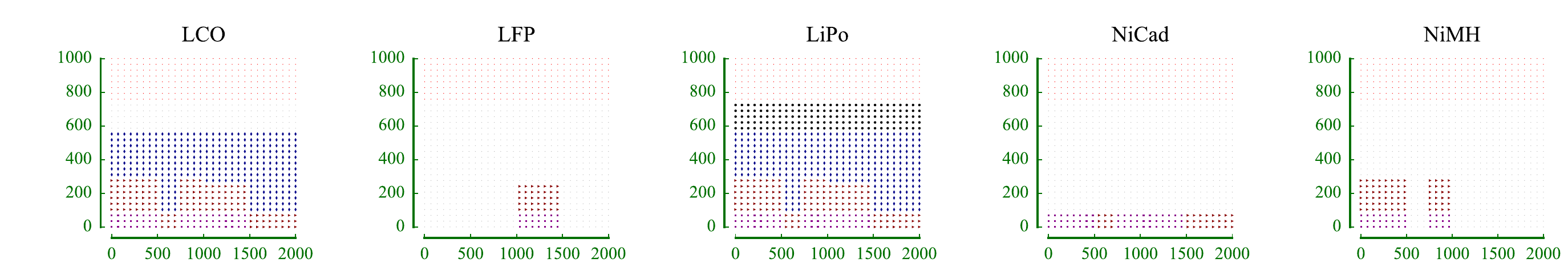}
\par\end{centering}
}
\par\end{centering}
\caption{\label{fig:mainbattery}This figure shows the optimal decision boundaries
for the different battery technologies for the design problem ``batteries'',
defined as the co-product of all battery technologies (\figref{batteriesbig}).
Each row shows a different variation of the problem. The first row
(panels \emph{a}\textendash \emph{b}) shows the case where the objective
function is the product of~$\left\langle \R{\text{mass}},\R{\text{cost}},\R{\text{maintenance}}\right\rangle $.
The shape of the symbols shows how many minimal solutions exists for
a particular value of the functionality~$\left\langle \F{\text{capacity}},\F{\text{missions}}\right\rangle $.
In this case, there are always three or more minimal solutions. The
second row (panels \emph{c}\textendash \emph{d}) shows the decision
boundaries when minimizing only the scalar objective~$\R{\text{cost}}$,
with a hard constraint on \R{mass}. The hard constraints makes some
combinations of the functionality infeasible. Note how the decision
boundaries are nonconvex, and how the formalisms allows to define
slight variations of the problem.}
\end{figure*}
\begin{figure*}
\begin{centering}
\subfloat[]{\begin{centering}
\includegraphics[scale=0.33]{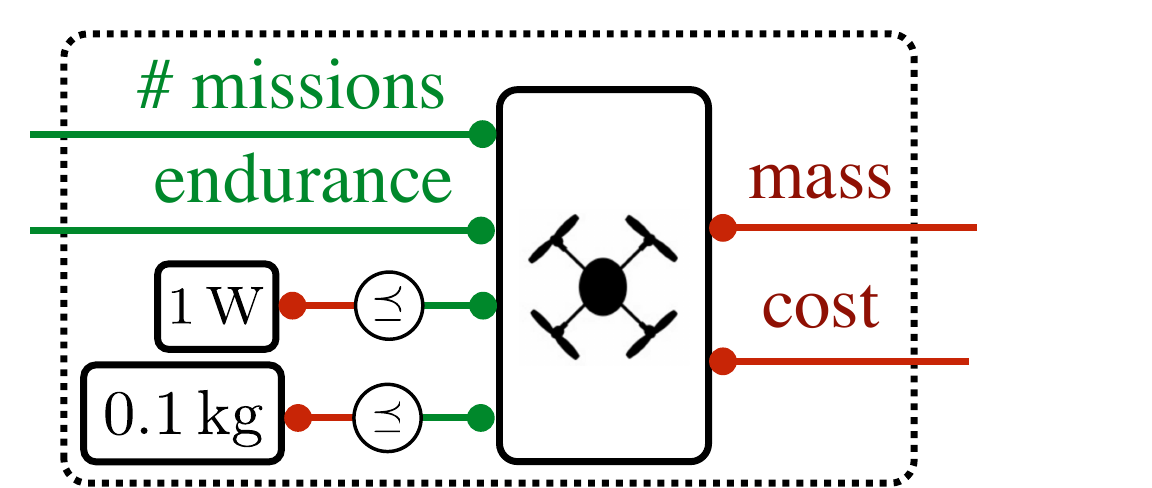}
\par\end{centering}
}\subfloat[]{\begin{centering}
\includegraphics[scale=0.45]{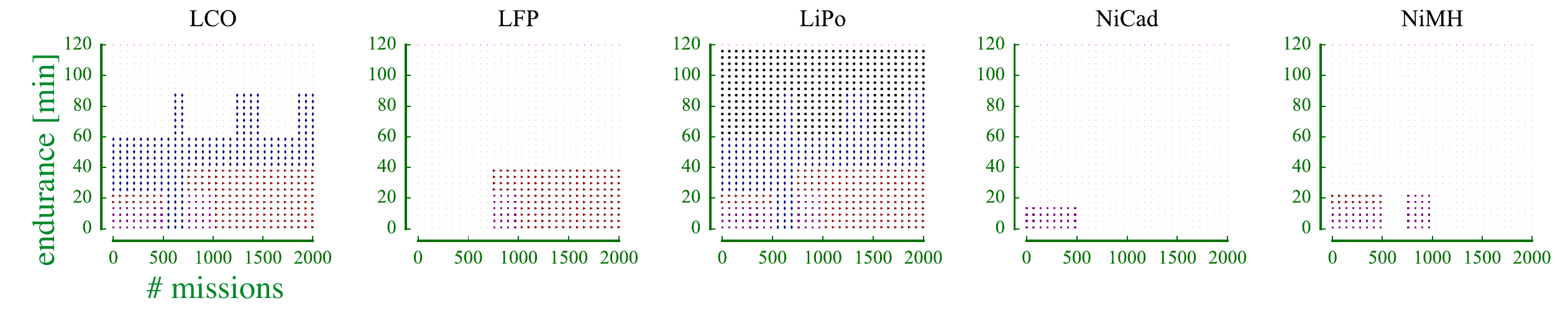}
\par\end{centering}
}
\par\end{centering}
\begin{centering}
\subfloat[]{\begin{centering}
\includegraphics[scale=0.33]{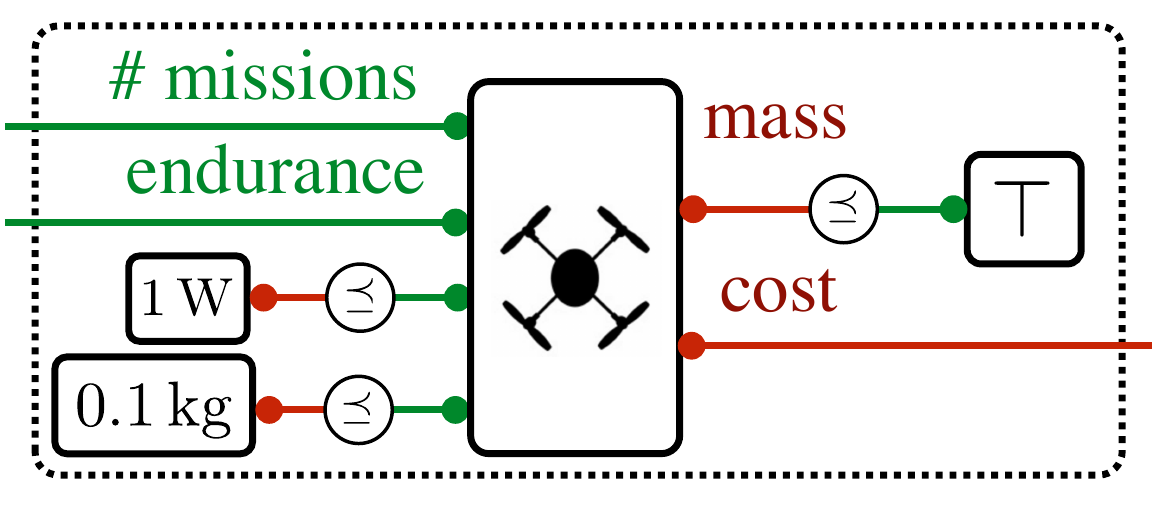}
\par\end{centering}
}\subfloat[]{\begin{centering}
\includegraphics[scale=0.45]{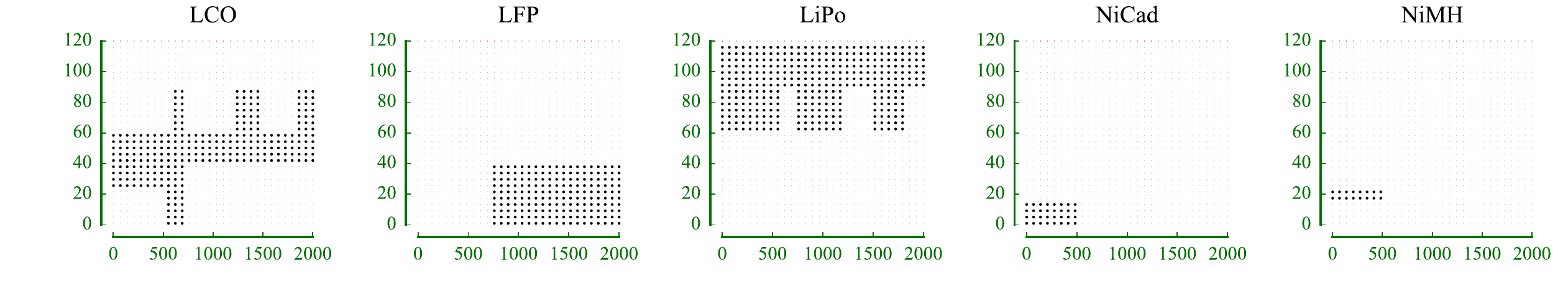}
\par\end{centering}
}
\par\end{centering}
\begin{centering}
\subfloat[]{\begin{centering}
\includegraphics[scale=0.33]{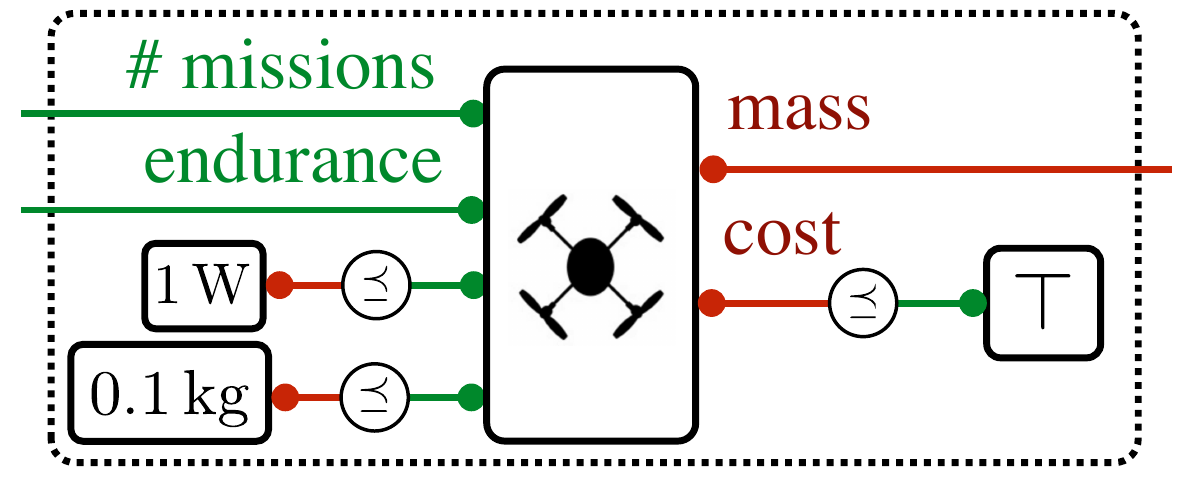}
\par\end{centering}
}\subfloat[]{\begin{centering}
\includegraphics[scale=0.45]{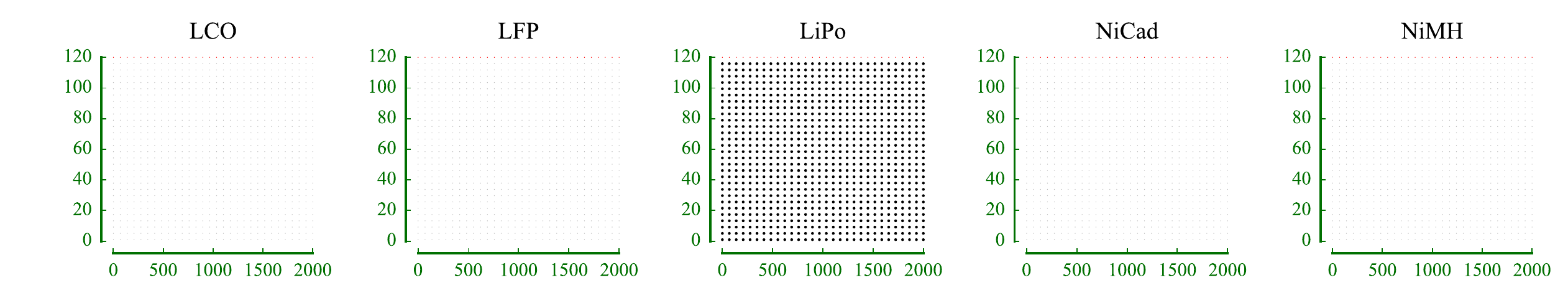}
\par\end{centering}
}
\par\end{centering}
\caption{\label{fig:drone_choice}This figure shows the decision boundaries
for the different values of battery technologies for the integrated
actuation-energetics model described in \figref{dronebigfig}. Please
see the caption of~\figref{mainbattery} for an explanation of the
symbols. Notice how in most cases the decision boundaries are different
than those in~\figref{mainbattery}: this is an example in which
one component cannot be optimized by itself without taking into account
the rest of the system.}
\end{figure*}

\section{Discussion of related work\label{sec:Discussion-of-related}}

\subsubsection*{Theory of design}

Modern engineering has long since recognized the two ideas of modularity
and hierarchical decomposition, yet there exists no general quantitative
theory of design that is applicable to different domains. Most of
the works in the theory of design literature study abstractions that
are meant to be useful for a human designer, rather than an automated
system. For example, a \emph{function structure }diagram~\cite[p. 32]{pahl07}
decomposes the function of a system in subsystems that exchange energy,
materials, and signals, but it is not a formal representation. From
the point of view of the theory of design, the contribution of this
work is that the \emph{design problem }abstraction developed, where
one takes functionality and resources as the interfaces for the subsystems,
is at the same time (1) mathematically precise; (2)~intuitive to
understand; and (3)~leads to tractable optimization problems.

This work also provides a clear answer to one long-standing issue
in the theory of design: the inter-dependence between subsystems,
(i.e., cycles in the co-design graph). Consider, as an example, Suh's
theory of \emph{axiomatic design~}\cite{suh01}, in which the first
``axiom'' is to keep the design requirements orthogonal (i.e., do
not introduce cycles). This work shows that it is possible to deal
elegantly with recursive constraints.

\subsubsection*{Partial Order Programming}

In ``Partial Order Programming''~\cite{parkerjr89partial} Parker
studies a hierarchy of optimization problems that can be represented
as a set of partial order constraints. The main interest is to study
partial order constraints as the means to define the semantics of
programming languages and for declarative approaches to knowledge
representation.

In Parker's hierarchy, MCDPs are most related to the class of problems
called \emph{continuous monotone partial order program} (CMPOP). CMPOPs
are the least specific class of problems studied by Parker for which
it is possible to obtain existence results and a systematic solution
procedure. MCDPs subsume CMPOPs. A CMPOP is an MCDP where: 1)~All
functionality and resources belong to the same poset~$\posA$ ($\funsp_{v}=\ressp_{v}=\posA$);
2)~Each functionality/resource relation is a simple map, rather than
a multi-valued relation; 3)~There are no dangling functionality edges
in the co-design diagram ($\funsp=\One$).

In a MCDP, each DP is described by a \scottcontinuous map~$\ftor:\funsp\rightarrow\Aressp$
which maps one functionality to a minimal set of resources. By contrast,
in a CMPOP an operator corresponds to a \scottcontinuous map~$\ftor:\funsp\rightarrow\ressp$.
The consequence is that a CMPOP has a unique solution~\cite[Theorem 8]{parkerjr89partial},
while an MCDP can have multiple minimal solutions (or none at all).

\subsubsection*{Abstract interpretation}

The methods used from order theory are the same used in the field
of \emph{abstract interpretation}~\cite{cousot14abstract}. In
that field, the least fixed point semantics arises from problems such
as computing the sets of reachable states. Given a starting state,
one is interested to find a subset of states that is closed under
the dynamics (i.e. a fixed point), and that is the smallest that contains
the given initial state (i.e. a \emph{least} fixed point). Reachability
and other properties lead to considering systems of equation of the
form
\begin{align}
x_{i} & =\varphi_{i}(x_{1},\,\dots\,,x_{i},\,\dots\,,x_{n}),\quad i=1,\dots,n,\label{eq:ai}
\end{align}
where each value of the index $i$ is for a control point of the program,
and~$\varphi_{i}$ are \scottcontinuous functions on the abstract
lattice that represents the properties of the program. In the simplest
case, each~$x_{i}$ could represent intervals that a variable could
assume at step~$i$. By applying the iterations, one finds which
properties can be inferred to be valid at each step.

We can repeat the same considerations we did for Parker's CMPOPs vs
MCDPs. In particular, in MCDP we deal with multi-valued maps, and
there is more than one solution.

In the field of abstract interpretation much work has been done towards
optimizing the rate of convergence. The order of evaluation in~\eqref{ai}
does not matter. Asynchronous and ``chaotic'' iterations were proposed
early~\cite{cousot77asynchronous} and are still object of investigation~\cite{bourdoncleefficient}.
To speed up convergence, the so called ``widening'' and ``narrowing''
operators are used~\cite{cortesi11widening}. The ideas of chaotic
iteration, widening, narrowing, are not immediately applicable to
MCDPs, but it is a promising research direction.

\section{Conclusions}

This paper described a mathematical theory of co-design, in which
the primitive objects are design problems, defined axiomatically as
relations between functionality, resources, and implementation. Monotone
Co-Design Problems (MCDPs) are the interconnection of design problems
whose functionality and resources are complete partial orders and
the relation is \scottcontinuous. These were shown to be non-convex,
non-differentiable, and not even defined on continuous spaces. Yet,
MCDPs have a systematic solution procedure in the form of a least
fixed point iteration, whose complexity depends on a measure of interdependence
between design problems\textemdash it is easier to design a system
composed of subsystems that are only loosely coupled. Based on this
theory, it is possible to create modeling languages and optimization
tools that allow the user to quickly define and solve multi-objective
design problems in heterogeneous domains. 

\textbf{\small{}}{\small \par}

\footnotesize

\setcounter{page}{1}

\printbibliography

\end{document}